\documentclass[a4paper,12pt,english]{article}
\usepackage[utf8]{inputenc}   
\usepackage{lmodern}
\usepackage[english]{babel}   
\usepackage[babel]{microtype}
\usepackage{amsfonts,amsmath,latexsym,amsthm,fixmath}
\usepackage{amssymb,ifthen}
\usepackage{color,graphicx}
\usepackage{url}
\usepackage[
    pdftex,
    pdfstartview={XYZ 0 1000 1.0},
    bookmarks=false,
    colorlinks,
    breaklinks=true,
    citecolor=citecolor,
    filecolor=blue,
    linkcolor=linkcolor,
    urlcolor=urlcolor,
    pdfborder={0 0 0}
]{hyperref}
\definecolor{urlcolor}{rgb}{0, 0.5, 0}
\definecolor{citecolor}{rgb}{.5,0,.25}
\definecolor{linkcolor}{rgb}{0,0,1}

\usepackage{geometry}
\graphicspath{{figures/}}
\def\svgpath{figures/}

\newtheorem{thm}{Theorem}

\newtheorem{prop}[thm]{Proposition}
\newtheorem{cor}[thm]{Corollary}
\newtheorem{lem}[thm]{Lemma}
\newtheorem{claim}{Claim}

\newtheorem*{nonumberlem}{Lemma}
\newtheorem*{nonumberprop}{Proposition}

\theoremstyle{remark} 
\newtheorem{remark}{Remark}

\def\R{\mathbb{R}}
\def\Z{\mathbb{Z}}
\def\D{\mathbb{D}}
\def\dP{\mathbb{P}} 
\def\define#1{\textbf{#1}}
\def\Ipath#1#2{[{#1}\stackrel{#2}{\to}]}

\def\Dpoint#1#2{({#1},{#2})}

\def\Focc#1{[{#1},{#1+1}]}
\def\Bocc#1{[{#1},{#1-1}]}
\def\Imath{i}
\def\Jmath{j}
\def\Ibar{i}
\def\Jbar{j}
\def\Overline{}

\def\zipper{unzip}
\def\bhom{\overset{*}{\sim}}
\def\immersion{perturbation}
\def\animmersion{a perturbation}
\def\Animmersion{A perturbation}

\newcommand{\inv}{\ensuremath{^{-1}}}

\newcommand{\includesvg}[2][]{%
\def\tempa{#1}\def\tempb{}%
\ifx\tempa\tempb\else\let\svgwidth\tempa\fi
\input{\svgpath#2.pdf_tex}%
}

 \title{Computing the Geometric Intersection Number of Curves\thanks{This work was supported by the LabEx PERSYVAL-Lab ANR-11-LABX-0025-01.}}

\author{
Vincent Despr\'e%
    \thanks{\'Equipe INRIA-GAMBLE, Universit\'e de Nancy, France,
\protect\url{Vincent.Despre@inria.fr}}
  \and
Francis Lazarus%
    \thanks{GIPSA-Lab, CNRS, Grenoble, France,
\protect\url{Francis.Lazarus@grenoble-inp.fr}}
}

\begin{document}
\maketitle

\begin{abstract}
The geometric intersection number of a curve on a surface is the minimal number of self-intersections of any homotopic curve, i.e. of any curve obtained by continuous deformation. 
Given a curve $c$ represented by a closed walk of length at most $\ell$ on a combinatorial surface of complexity $n$ we describe simple algorithms  to (1) compute the geometric intersection number of $c$ in $O(n+ \ell^2)$ time, (2) construct a curve homotopic to $c$ that realizes this geometric intersection number in $O(n+\ell^4)$ time, (3) decide if the geometric intersection number of $c$ is zero, i.e. if $c$ is homotopic to a simple curve, in $O(n+\ell\log\ell)$ time. The algorithms for (2) and (3) are restricted to orientable surfaces, but the algorithm for (1) is also valid on non-orientable surfaces.

To our knowledge, no exact complexity analysis had yet appeared on those problems. An optimistic analysis of the complexity of the published algorithms for problems (1) and (3) gives at best a $O(n+g^2\ell^2)$ time complexity on a genus $g$ surface without boundary. No polynomial time algorithm was known for problem (2) for surfaces without boundary. Interestingly, our solution to problem (3) provides a quasi-linear algorithm to a problem raised by Poincar\'e more than a century ago. Finally, we note that our algorithm for problem (1) extends to computing the geometric intersection number of two curves of length at most $\ell$ in $O(n+ \ell^2)$ time.
\end{abstract}
\newpage
\tableofcontents
\newpage
\section{Introduction}
Let $S$ be a surface. Two closed curves $\alpha,\beta: \R/\Z\to S$ are \define{freely homotopic}, written $\alpha\sim \beta$, if there exists a continuous map $h: [0,1]\times\R/\Z\to S$ such that $h(0,t)=\alpha(t)$ and $h(1,t)=\beta(t)$ for all $t\in \R/\Z$. The curves $\alpha$ and $\beta$ being homotopic or not, their number of intersections is
\[ |\alpha\cap\beta| = |\{(t,t')\mid t, t' \in \R/\Z \text{ and } \alpha(t) = \beta(t')\}|.
\]
Their \define{geometric intersection number} only depends on their free homotopy classes and is defined as
\[i(\alpha,\beta) = \min_{\alpha'\sim\alpha, \beta'\sim\beta} |\alpha'\cap\beta'| 
\]
This minimum is finite and obtained with curves that intersect transversally.
Likewise, the number of self-intersections of $\alpha$ is given by
\[\frac{1}{2}|\{(t,t')\mid t\neq t' \in \R/\Z \text{ and } \alpha(t) = \alpha(t')\}|,
\] 
and its minimum over all the curves freely homotopic to $\alpha$ is its \define{geometric self-intersection number} $i(\alpha)$. Note the one half factor that comes from the identification of $(t,t')$ with $(t',t)$.

The geometric intersection number is an important parameter that allows to stratify the set of homotopy classes of curves on a surface. The surface is usually endowed with a hyperbolic metric, implying that each homotopy class is identified by its unique geodesic representative.  Extending a former result by Mirzakhani~\cite{m-gnscg-08}, Sapir~\cite{s-bnnsc-15,m-cmcgo-16} has recently provided upper and lower bounds for the number of closed geodesics with bounded length and bounded geometric intersection number. Chas and Lalley~\cite{cl-sicts-12} also proved that the distribution of the geometric intersection number with respect to the word length approaches the Gaussian distribution as the length grows to infinity. Other more experimental results were obtained with the help of a computer to show the existence of length-equivalent homotopy classes with distinct geometric intersection numbers~\cite{c-sinle-14}.
Hence, for both theoretical and practical reasons, various aspects of 
the computation of geometric intersection numbers have been studied in the past including the algorithmic ones. Nonetheless, all the previous approaches rely on rather complex mathematical arguments and to our knowledge no exact complexity analysis has yet appeared.
 In this paper, we make our own the words of Dehn who noted that the metric on words (on some basis of the fundamental group of the surface) can advantageously replace the hyperbolic metric~\cite{l-ttgpd-10}. We propose a combinatorial framework that leads to simple algorithms of  low complexity to compute the geometric intersection number of curves or to test if this number is zero. Our approach is based on the computation of canonical forms as recently introduced for the purpose of testing whether two curves are homotopic~\cite{lr-hts-12,ew-tcsr-13}. Canonical forms are instances of combinatorial geodesics who share nice properties with the geodesics of a hyperbolic surface. On hyperbolic surfaces each homotopy class contains a unique geodesic that moreover minimizes the number of self-intersections. 
Although a combinatorial geodesic is generally not unique in its homotopy class, it must stay at distance one from its canonical representative and a careful analysis of its structure leads to the first result of the paper.
 \begin{thm}\label{th:main-result}
   Given two curves represented by closed walks of length at most $\ell$ on a combinatorial surface of complexity $n$ we can compute the geometric intersection number of each curve or of the two curves in $O(n+ \ell^2)$ time.
 \end{thm}
As usual the complexity of a combinatorial surface stands for its total number of vertices, edges and faces. A key point in our algorithm is the ability to compute the primitive root of a canonical curve $c$ in linear time. This is a curve $r$ that is not homotopic to a proper power of any other curve and such that $c\sim r^k$ for some integer $k$. We next provide an algorithm to compute an actual curve immersion -- its combinatorial description is part of our combinatorial framework --  that minimizes the number of self-intersections in its homotopy class. While the combinatorial surface in Theorem~\ref{th:main-result} may be non-orientable, the next two results only apply to orientable surfaces. 
\begin{thm}\label{th:compute-immersion}
  Given a curve $c$ represented by a closed walk of length $\ell$ on an orientable combinatorial surface of complexity $n$ we can compute a combinatorial immersion with $i(c)$ crossings in $O(n+\ell^4)$ time.
\end{thm}
We also propose a nearly optimal algorithm that answers an old problem studied by Poincar\'e~\cite[\S 4]{p-ccal-04}:  decide if the geometric intersection number of a curve is null, that is if the curve is homotopic to a simple curve. 
\begin{thm}\label{th:simple-curve}
  Given a curve represented by a closed walk of length $\ell$ on an orientable combinatorial surface of complexity $n$ we can decide  if the curve is homotopic to a simple curve in  $O(n+ \ell\log\ell)$ time\footnote{In the preliminary version of this work, as it appeared in the proceddings of the 33rd International Symposium on Computational Geometry (SoCG 2017), the announced complexity had an extra $\log\ell$ factor.}. In the affirmative we can construct an embedding of $c$ in the same amount of time.
\end{thm}
We emphasize that our results represent significant progress with respect to the state of the art. No precise analysis appeared in the previously proposed algorithms~\cite{bs-ascs-84,c-sccs-69,c-wns2-72,cl-pggin-87,l-pggin2-87,gs-mcmcr-97,p-cails-02,gkz-amnip-05} concerning Theorems~\ref{th:main-result} or~\ref{th:simple-curve}. An optimistic analysis of what seems the most efficient approach~\cite[Th. 3.7]{l-pggin2-87}, although particularly complex, gives at best a quadratic time complexity for computing the geometric intersection number on an orientable genus $g$ surface without boundary, assuming that the curves are primitive and expressed as words in a canonical basis of the fundamental group. Note that these assumptions are not required in Theorems~\ref{th:main-result},~\ref{th:compute-immersion} or ~\ref{th:simple-curve}.
Schaefer et al.~\cite{sss-cdtgi-08} propose an efficient computation of the geometric intersection number of a set of curves represented by normal coordinates in a triangulated surface. However, their approach is limited to \emph{simple} input curves.  Apart from a recent algorithm by Aretinnes~\cite{a-cavrm-15}, which is restricted to surfaces with \emph{nonempty} boundary, we know of no polynomial time algorithm for Theorem~\ref{th:compute-immersion}. Finally, Theorem~\ref{th:simple-curve} states the first quasi-linear algorithm for detecting homotopy classes of simple curves since the problem was raised by Poincar\'e more than a century ago~\cite[\S 4]{p-ccal-04}. A related problem was tackled by Chang \emph{et al.}~\cite[Th. 8.2]{cex-dwsp-15} (see also~\cite {aaet-rwsp-17} for the planar case) who describe an algorithm to decide if a given closed path on a combinatorial surface is \emph{weakly simple}, i.e. admits an infinitesimal perturbation that is an embedding. Note that a path may go several times through the same vertex or edge and still have infinitesimal perturbations that are simple. However, the algorithm in~\cite{cex-dwsp-15} does not always detect if the path is homotopic to a simple curve since it is only authorized an infinitesimal perturbation. 

In the next section we review some of the previous relevant works. Section~\ref{sec:strategy} presents our general simple strategy to compute the geometric intersection number. This strategy will be applied in a combinatorial framework introduced in Section~\ref{sec:framework}. 
The proof of Theorem~\ref{th:main-result} restricted to primitive curves on orientable surfaces is then given in Section~\ref{sec:counting-on-oriented}. The general case for curves that may be non-primitive on possibly non-orientable surfaces is treated in Section~\ref{sec:counting}. We remark that the number of crossings of two given combinatorial curves is not always well-defined as parts of curves may overlap along shared edges. We resolve this ambiguity with the help of combinatorial perturbations as presented in Section~\ref{sec:perturbation}. The computation of a minimally crossing perturbation is then presented in Section~\ref{sec:computing-immersions} with the proof of Theorem~\ref{th:compute-immersion}. We finally propose a simple algorithm to detect and embed curves homotopic to simple curves (Theorem~\ref{th:simple-curve}) in Section~\ref{sec:simple-curves}. 

\section{Historical notes}
In the fifth supplement to its \emph{Analysis situs} Poincar\'e~\cite[\S 4]{p-ccal-04} describes a method to decide whether a given closed curve $\gamma$ on a surface can be continuously deformed to a simple curve. For this purpose, he considers the surface as the quotient $\mathbb{D}/\Gamma$ of the Poincar\'e disk $\mathbb{D}$ by a (Fuchsian) group $\Gamma$ of hyperbolic transformations. The endpoints of a lift of $\gamma$ in the Poincar\'e disk are related by a hyperbolic transformation whose axis is a hyperbolic line $L$ representing the unique geodesic homotopic to $\gamma$. He concludes that $\gamma$ is homotopic to a simple curve if and only if all the transforms of $L$ by $\Gamma$ are pairwise disjoint or equal. This method was turned into an algorithm by Reinhart~\cite{r-ajccs-62} who worked out the explicit computations in the Poincar\'e disk using the usual representation of hyperbolic transformations by two-by-two matrices. The entries of the matrices being algebraic, the computation could indeed be performed accurately on a computer. The ability to recognize curves that are  \define{primitive}, i.e. whose homotopy class cannot be expressed as a proper power of another class, happens to be crucial in this algorithm though computationally expensive.

Birman and Series~\cite{bs-ascs-84} subsequently proposed an algorithm for the case of surfaces with nonempty boundary that avoids manipulating algebraic numbers. While their arguments appeal to a hyperbolic structure, their algorithm is purely combinatorial. Intuitively, a surface with boundary deform retracts onto a fat graph (in fact a fat bouquet of circles) whose universal covering space embeds as a fat tree in the Poincar\'e disk. The successive lifts of a curve $\gamma$ trace a bi-infinite path in this tree. The limit points of this path belongs to the circle $\partial\mathbb{D}$ at infinity and coincide with the ideal endpoints of the axis of the hyperbolic transformation corresponding to any of the lifts of the curve in the path. The question as to whether two lifts give rise to intersecting axes can thus be reduced to test if the corresponding bi-infinite paths have separating limit points on $\partial\mathbb{D}$. As for Reinhart, Birman and Series assume that the homotopy class of $\gamma$ is given by a word $W$ on some given set of generators of the fundamental group of the surface. In turn, the above test on bi-infinite paths boils down to consider the cyclic permutations of $W$ and $W^{-1}$ in a cyclic lexicographic order and to check if this ordering is  well-parenthesized with respect to some pairing of the words.

Cohen and Lustig~\cite{cl-pggin-87} further observed that the approach of Birman and Series on surfaces with boundary could be extended to count the geometric intersection number of curves. 
In a second paper Lustig~\cite{l-pggin2-87} tackles the case where the curves are taken on a surface without boundary. Like Poincar\'e he considers a closed surface with negative Euler characteristic as the quotient of the Poincar\'e disk by a group of transformations isomorphic to the fundamental group of the surface. The main contribution of the paper is to define a canonical representative for every free homotopy class $\alpha$ given as a word in some fixed system of loops generating the fundamental group. Lustig first notes that there is no obvious way of choosing a canonical form among the words representing $\alpha$. In particular, the shortest words are far from unique. He rather represents (a lift of) $\alpha$ by a path in the union $G\cup N\cup H$ of three tessellations of $\mathbb{D}$, where the edges of $G$ are all the lifts of the generating loops, $N$ is the dual tessellation, and the edges of $H$ joins the vertices of $G$ with the vertices of $N$. (Although the graphs of $G$, $N$ and $H$ are embedded their union is not as every edge of $G$ crosses its dual edge in $N$.) Lustig gives a purely combinatorial characterization of canonical paths and argues that the method in his first paper~\cite{cl-pggin-87} can be applied to the canonical representative of $\alpha$. Overall the two papers add up to 60 pages with essential arguments from hyperbolic geometry and the complexity analysis of the whole procedure remains to be done.

Other approaches were developed without assuming any hyperbolic structure.
Based on the notion of winding number, Chillingsworth~\cite{c-sccs-69,c-wns2-72} provides an algorithm to test whether a curve is homotopic (this time with fixed basepoint) to a simple curve on a surface with nonempty boundary. He also proposed an algorithm for determining when a given set of \emph{simple} closed curves can be made disjoint by (free) homotopy~\cite{c-afdsc-71}.  While the winding number relies on a differentiable structure of the surface, Zieschang~\cite{z-aekf-65,z-aekf2-69} appeals to  the topological structure only. He used the connection between the automorphisms of a topological surface and the automorphisms of its fundamental group in order to detect the homotopy classes of simple closed curves. We also mention some works by Schaefer et al.~\cite{sss-cdtgi-08} to compute the geometric intersection number of two \emph{simple} curves given by their normal coordinates in a triangulated surface. Their algorithm is based on repeated applications of Dehn twists and is claimed to have polynomial time complexity. 

A related work by Hass and Scott~\cite{hs-ics-85} is concerned with curves that have excess intersection, i.e. that can be homotoped so as to reduce their number of intersections. Hass and Scott introduce various types of monogons and bigons that can be either embedded, singular or weak.
A \emph{singular monogon} of a curve $\gamma: \R/\Z\to S$ is a contractible subpath of $\gamma$ whose endpoints define a self-intersection of $\gamma$.
A \emph{bigon} of $\gamma$ is defined by two self-intersections joined by two homotopic subpaths (with fixed endpoints) $\gamma|_\sigma$ and $\gamma|_{\sigma'}$ with $\sigma, \sigma' \subset \R/\Z$. 
The bigon is said \emph{singular} if the defining segments $\sigma,\sigma'$ are disjoint and \emph{weak} otherwise. Hass and Scott prove that a curve with excess self-intersection on an orientable surface must have a singular monogon or a singular bigon. Their result directly suggests an algorithm to compute the geometric intersection number of a curve: iteratively remove monogons or untie bigons until there are no more. The final configuration must have the minimal number of self-intersections. Designing an efficient procedure to find monogons and bigons remains the crux of this approach. Arettines~\cite{a-cavrm-15} has proposed an algorithm based on this approach to compute a minimal configuration of a single curve on an orientable surface with \emph{nonempty} boundary. 
Note that the method cannot be extended to compute the geometric intersection number of a curve on a non-orientable surface since Hass and Scott give a counter-example of a curve with excess self-intersections but with no singular monogon or bigon. The method also cannot be extended to compute the geometric intersection number of two curves on an orientable surface. Indeed, Hass and Scott give two counter-examples to the fact that two curves with excess intersections should have a singular bigon. (A singular bigon between two curves is a pair of homotopic subpaths, one of each curve.) One of their counter-examples contains a curve with excess self-intersections and the other one contains a non-primitive curve. Our counter-example, Figure~\ref{fig:no-singular-bigon}, shows that even assuming each curve to be primitive and in minimal configuration, we may have excess intersections without singular bigons. 
\begin{figure}[h]
  \centering
\includesvg[.5\linewidth]{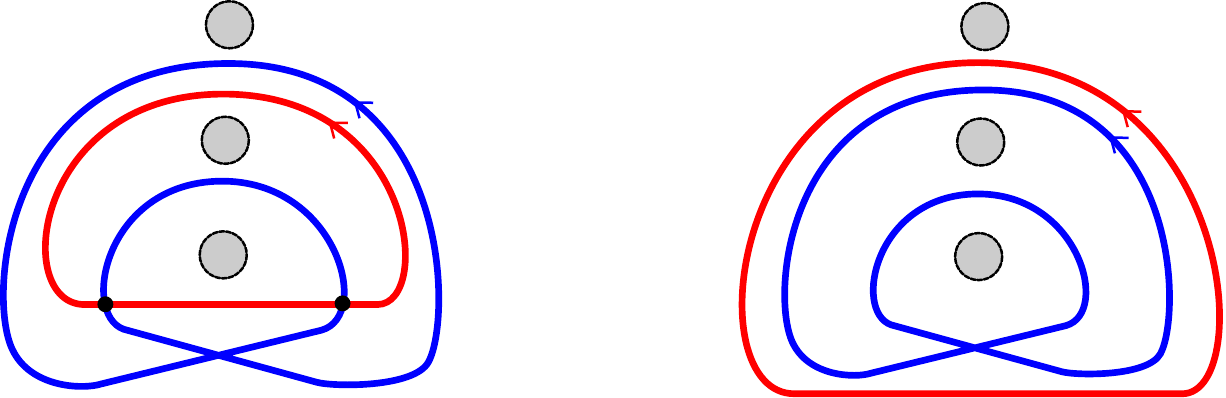}
  \caption{The plain circles represent non-contractible curves. The two curves $\gamma$ and $\delta$ on the left have homotopic disjoint curves $\gamma'$ and $\delta'$. The curves $\gamma$ and $\delta$  thus have excess intersection although there is no singular bigon between the two. If $A=\gamma(0)=\delta(0)$ and $B=\gamma(u)=\delta(v)$ we nonetheless have $\delta|_{[0,v]}\sim \gamma|_{[0,1+u]}$  where $\gamma|_{[0,1+u]}$ is the concatenation of $\gamma$ with $\gamma|_{[0,u]}$. In particular, $\gamma|_{[0,1+u]}$ wraps more that once around $\gamma$.
}
  \label{fig:no-singular-bigon}
\end{figure}
Nonetheless, it was proved~\cite{hs-scs-94,gs-mcmcr-97,p-cails-02} that starting from any configuration of curves one may reach a configuration with a minimal number of intersections by applying a finite sequence of elementary moves involving monogons, bigons and trigons similar to the Reidemeister moves in knot theory. A surprising consequence was obtained by Neumann-Coto~\cite{n-csgs-01}. Define a \emph{cut and paste} on a family of curves by cutting the curves at some of their intersection points and glueing the resulting arcs in a different order. Neumann-Coto proves that any set of primitive curves can be brought to a homotopic set with minimal (self-)intersections by a set of cut and paste operations. Note that each intersection of two curve pieces can be re-arranged in three different manners, including the original one, by a cut and paste. Hence, if the curves have $I$ (self-)intersections we may find a minimal configuration out of the $3^I$ possible re-configurations! 

We also mention the algebraic approach of Gon\c{c}alves et al.~\cite{gkz-amnip-05} based on previous works by Turaev~\cite{t-iltdm-79} who introduced intersection forms over the integral group ring of the fundamental group of the surface. This approach is based on the transformation of each curve into a certain algebraic sum of homotopy classes. A simple analysis of the time complexity for the computation of this sum leads to $O(\ell^5)$ operations for a curve with $\ell$ self-intersections. See~\cite{cl-pggin-87,gkz-amnip-05} for more historical notes.

\section{Our strategy for counting intersections}\label{sec:strategy}
Here, we assume some familiarity with basic hyperbolic geometry in the Poincar\'e disk and with the notion of (universal) covering of a surface. We refer the reader to Stillwell~\cite{s-ctcgt-93}, especially Chapter 1 and 6, for a gentle introduction. 
Following Poincar\'e's original approach we   
represent a surface $S$ with negative Euler characteristic as the hyperbolic quotient surface $\D/\Gamma$ where $\Gamma$ is a discrete group of hyperbolic motions of the Poincar\'e disk $\D$.
We denote by $p: \D\to \D/\Gamma = S$ the universal covering map. Any closed curve $\alpha:\R/\Z\to S$ gives rise to its infinite power $\alpha^\infty:\R\to\R/\Z \to S$ that wraps around $\alpha$ infinitely many times. A \define{lift} of $\alpha^\infty$ is any curve $\tilde{\alpha}:\R\to \D$ such that $p\circ \tilde{\alpha} = \alpha^\infty$ where the parameter of $\tilde{\alpha}$ is defined up to an integer translation (we thus identify the curves $t\mapsto \tilde{\alpha}(t+k)$, $k\in\Z$). \define{For conciseness, we shall write from now on ``a lift of $\alpha$'' where we actually mean a lift of $\alpha^\infty$.}
Note that $p\inv(\alpha)$ is the union of all the images $\Gamma\cdot \tilde{\alpha}$ of $\tilde{\alpha}$ by the motions in $\Gamma$. 
The curve $\tilde{\alpha}$ has two limit points on the boundary of $\D$ which can be joined by a unique hyperbolic line $L$.
The projection $p(L)$ wraps infinitely many times around the unique geodesic homotopic to $\alpha$. In particular, the set of pairs of limit points  of all lifts of $\alpha$ only depends on the homotopy class of $\alpha$.

No two motions of $\Gamma$ have a limit point in common unless they are powers of the same motion. This can be used to show that when $\alpha$ is primitive, its lifts are uniquely identified by their limit points~\cite{fm-pmcg-12}.
Let $\alpha$ and $\beta$ be two primitive curves. We fix a lift $\tilde{\alpha}$ of $\alpha$ and denote by $\tau\in\Gamma$ the hyperbolic motion sending  $\tilde{\alpha}(0)$ to $\tilde{\alpha}(1)$.  Note that $\tau$ leaves $\tilde{\alpha}$ globally invariant. Let $\Gamma\cdot \tilde{\beta}$ be the set of lifts of $\beta$. We consider the subset
\[B=\{\tilde{\beta}'\in \Gamma\cdot \tilde{\beta}\mid \text{ the limit points of } \tilde{\beta}' \text{ and }  \tilde{\alpha} \text{ alternate along } \partial\D\},
\]
and we denote by $B/\tau$ the set of equivalence classes of lifts generated by the relations $\tilde{\beta}'\approx \tau(\tilde{\beta}')$. In other words, two lifts are equivalent if one is the image of the other by some power of $\tau$.
\begin{lem}[\cite{r-ajccs-62}]\label{lem:strategy}
 $i(\alpha,\beta) = |B/\tau|$.
\end{lem}
\begin{proof}
  Put $I(\alpha,\beta) =\{(t,t')\mid t, t' \in \R/\Z \text{ and } \alpha(t) = \beta(t') \text{ and if } \alpha=\beta: t\neq t'\}$. 
Define a map $\varphi: I(\alpha,\beta)\to (\Gamma\cdot \tilde{\beta})/\tau$ as follows. 
Given $(u\bmod 1, v\bmod 1)\in I(\alpha,\beta)$ there is, by the unique lifting property of coverings, a unique lift $\tilde{\beta}'$ of $\beta$ that satisfies $\tilde{\beta}'(v) = \tilde{\alpha}(u)$. We set $\varphi(u\bmod 1, v\bmod 1)$ to the class of this lift. Note that changing $u$ to $u+k$, $k\in\Z$, leads to the lift $\tau^k(\tilde{\beta}')$, so that $\varphi$ is well-defined. We have $B/\tau \subset \varphi(I(\alpha,\beta))$. Indeed, if $\tilde{\beta}'\in B$ then $\tilde{\beta}'$ and $\tilde{\alpha}$ must intersect at some point $\tilde{\beta}'(v) = \tilde{\alpha}(u)$. It follows that $\varphi(u\bmod 1, v\bmod 1)$ is the class of $\tilde{\beta}'$. As an immediate consequence,
\[ |I(\alpha,\beta)| \geq |B/\tau|
\]
When $\alpha$ and $\beta$ are geodesics all their lifts are hyperbolic lines and $\varphi$ is a bijection onto its image $B/\tau$. We conclude that $|I(\alpha,\beta)|$ is minimized among all homotopic curves, so that
$i(\alpha,\beta) = |B/\tau|$.
\end{proof}
When $\alpha$ and $\beta$ are hyperbolic geodesics, their lifts being hyperbolic lines have alternating limit points exactly when they have a non-empty intersection and they have a unique intersection point in that case. This point projects to a crossing of $\alpha$ and $\beta$ that actually identifies the corresponding element of $B/\tau$.
When $\alpha$ and $\beta$ are not geodesic the situation is more ambiguous and their lifts may have multiple intersection points. 
Those intersection points project to crossings of $\alpha$ and $\beta$, so that the elements of $B/\tau$ are now identified with subsets of crossing points (with odd cardinality) rather than single crossing points. The induced partition is generated by the following relation: two crossings are equivalent if they are connected by a pair of homotopic subpaths of $\alpha$ and $\beta$, namely one of the two subpaths of $\alpha$ and one of the two subpaths of $\beta$ cut by the two crossings. Indeed, if the two crossings are projections of intersections of two lifts, then the paths between those intersections in each lift project to homotopic paths. Conversely, homotopic paths lift to paths with common endpoints that can be seen as subpaths of two intersecting lifts. In order to compute the above partition, we thus essentially need an efficient procedure for testing if two paths are homotopic. This homotopy test can be performed in linear time according to Theorem~\ref{th:canonical}. Since a combinatorial curve of length $\ell$ may have $O(\ell^2)$ crossings, we directly obtain an algorithm with time complexity $O(\ell^5)$  to compute the above partition.

In practice, we shall work in a combinatorial framework as presented in the next section.
In this framework, curves are identified with paths in the 1-skeleton of a combinatorial surface. It appears more convenient to assume that all the faces of the combinatorial surface are quadrilaterals. We thus need to reduce the given surface to a system of quads\footnote{Some topologists would call it a \emph{square-tiled surface}.} as described in Section~\ref{sec:quad-system}. Each curve may now be replaced by a homotopic combinatorial geodesic in such a system of quads. The situation becomes somehow intermediate between the ideal hyperbolic case and the most general situation; the unique intersection point of hyperbolic lines is replaced by an elongated pair of homotopic paths. See Figure~\ref{fig:strategy}.
By taking advantage of the specific structure of combinatorial geodesics in the system of quads we can indeed avoid computing the above partition and directly identify $B/\tau$ with certain pairs of homotopic subpaths of the combinatorial geodesics. See Proposition~\ref{prop:counting-crossings}. This leads to a more efficient algorithm with quadratic complexity.
\begin{figure}
  \centering
  \includesvg[\linewidth]{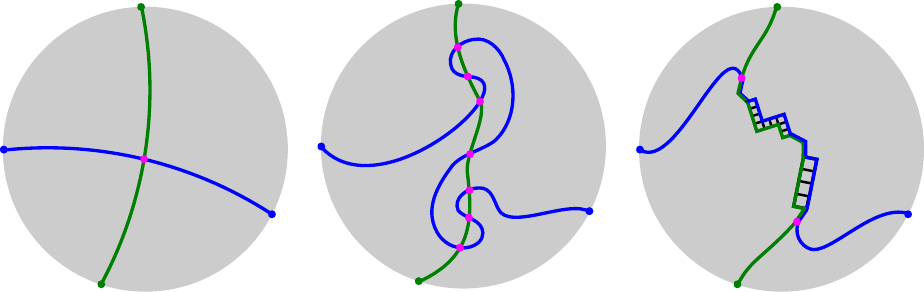}
  \caption{Left, two intersecting hyperbolic lines corresponding to lifts of geodesic curves, say $\alpha$ and $\beta$. Middle, lifts of non-geodesic curves homotopic to $\alpha$ and $\beta$ may intersect several times. Right, lifts of combinatorial geodesics.}
  \label{fig:strategy}
\end{figure}

\section{Combinatorial framework}\label{sec:framework}
In order to study the computational aspects of the intersection number we need to specify the encoding of curves on surfaces. A common representation of the homotopy class of a curve on a surface consists of a word in some set of generators of the fundamental group of the surface. Typically, when the fundamental group is given by a combinatorial presentation with a single relation, one can build a two dimensional complex that is homeomorphic to the surface and such that the word for the curve corresponds to a homotopic closed path in this complex. We shall use a more general notion of combinatorial surfaces where curves a represented by closed walks.
We introduce below this combinatorial framework. In order to mimic as closely as possible hyperbolic geometry we further restrict our framework to system of quads which are combinatorial surfaces composed only of quadrilaterals. We finally study the structure of geodesics in such systems.

\subsection{Combinatorial curves on surfaces}
\paragraph{Combinatorial surfaces.}
As usual in computational topology, we model a surface by a cellular embedding of a graph $G$ in a compact topological surface $S$. Such a cellular embedding can be encoded by a \define{combinatorial surface} composed of the graph $G$ itself together with a rotation system~\cite{mt-gs-01} that records for every vertex of the graph the \define{direct} cyclic order of the incident arcs. When $S$ is orientable, this order can be chosen consistently for all vertices. We call this order clockwise. Nonorientable surfaces can be encoded by  providing an additional signature with every edge that indicates if the direct orientations at its endpoints are consistent or not. The edge is said \define{half-twisted} when the orientations are inconsistent. The \define{facial walks} are obtained from the rotation system by the face traversal procedure as described in~\cite[p.93]{mt-gs-01}.
In order to handle surfaces with boundaries we allow every face of $G$ in $S$ to be either an open disk or an annulus (open on one side). In other words $G$ is a cellular embedding in the closure of $S$ obtained by attaching a disk to every boundary of $S$. We record this information by storing a boolean for every facial walk of $G$ indicating whether the associated face is perforated or not. All the considered graphs may have loops and multiple edges. A directed edge will be called an \define{arc} and each edge corresponds to two opposite arcs. We denote by $a\inv$ the arc opposite to an arc $a$. 
Every combinatorial surface $\Sigma$ can be reduced by first contracting the edges of a spanning tree and then deleting edges incident to distinct plain, i.e. non-perforated, faces. The resulting \define{reduced surface} has a single vertex. A reduced surface without boundary has a single face, but in general a reduced surface may have several faces. The combinatorial surface $\Sigma$ and its reduced version encode different cellular embeddings on a same topological surface.

\paragraph{Combinatorial curves.}
Consider a combinatorial surface with its graph $G$. A combinatorial curve (or path) $c$ is a  walk in $G$, i.e. an alternating sequence of vertices and arcs, starting and ending with a vertex, such that each vertex in the sequence is the target vertex of the previous arc and the source vertex of the next arc. We generally omit the vertices in the sequence. A combinatorial curve is closed when additionally the first and last vertex are equal. When no confusion is possible we shall drop the adjective combinatorial. A closed curve is \define{2-sided} if its number of half-twisted arcs is even. It is otherwise \define{1-sided}.
The \define{length} of $c$ is its total number of arc occurrences, which we denote by $|c|$. If $c$ is closed, we write $c(\Ibar)$, $\Ibar\in \Z/|c|\Z$, for the vertex of index $\Ibar$ of $c$ and $c\Focc{\Imath}$ for the  arc joining $c(\Ibar)$ to $c(\Ibar+1)$.  A \define{double point} of $c$ is a pair of indices $\Dpoint{\Imath}{\Jmath}\in \Z/|c|\Z\times \Z/|c|\Z$ such that $\Ibar\neq \Jbar$ and $c(\Ibar)=c(\Jbar)$. Likewise, given another closed curve $d$, we define a double point of $(c,d)$ as a pair $\Dpoint{\Imath}{\Jmath}\in \Z/|c|\Z\times \Z/|d|\Z$ with $c(\Ibar)=d(\Jbar)$.

For convenience we set $c[\Ibar+1,\Ibar]=c[\Ibar,\Ibar+1]\inv$ to allow the traversal of $c$ in reverse direction. In order to differentiate the arcs from their occurrences we denote by $[\Ibar,\Ibar\pm 1]_c$ the corresponding occurrence of the arc $c[\Ibar,\Ibar\pm 1]$ in $c^{\pm 1}$, where $c\inv$ is obtained by traversing $c^{1}:=c$ in the opposite direction. More generally, for any non-negative integer $\ell$ and any sign $\varepsilon\in\{-1,1\}$,
The sequence of indices 
\[(\Ibar,\Overline{\Imath+\varepsilon},\Overline{\Imath+2\varepsilon},\dots,\Overline{\Imath+\varepsilon\ell})
\]
is called an \define{index path} of $c$ of length $\ell$. The index path can be \define{forward} ($\varepsilon =1$) or \define{backward} ($\varepsilon =-1$) and can be longer than $c$ so that an index may appear more than once in the sequence. We denote this path by $\Ipath{\Imath}{\varepsilon\ell}_{c}$. Its \define{image path} is given by the arc sequence
\[c\Ipath{\Imath}{\varepsilon\ell} := 
(c[\Ibar,\Overline{\Imath+\varepsilon}],c[\Overline{\Imath+\varepsilon},\Overline{\Imath+2\varepsilon}],\dots,c[\Overline{\varepsilon(\ell-1)},\Overline{\varepsilon\ell}]).
\]
The image path of a length zero index path is just a vertex. There is an evident notion of inclusion between index paths: a forward index path $\Ipath{\Imath}{\ell}_{c}$ is included in a forward index path $\Ipath{\Jmath}{k}_{c}$ if there exists a natural number $s$ such that $s+\ell \leq k$ and $i=j+s$. This inclusion relationship extends to both forward and backward paths, replacing each backward index path $\Ipath{\Imath}{-\ell}_{c}$ by $\Ipath{\Imath-\ell}{\ell}_{c}$ to perform the comparison.

A closed curve is \define{contractible} if it is homotopic to a trivial curve (i.e., a curve reduced to a single vertex). We will  implicitly assume that a  homotopy has fixed endpoints when applied to paths and is free when applied to closed curves. If $\alpha, \beta$ are two closed paths, we shall sometimes write $\alpha\, \bhom\,\beta$ to denote homotopy with fixed basepoint  as opposed to $\alpha\sim\beta$ for the free homotopy relation.

\subsection{Systems of quads}\label{sec:quad-system}
\paragraph{Reduction to a system of quads.}
Let $\Sigma$ be a combinatorial surface with negative Euler characteristic. Following Lazarus and Rivaud~\cite{lr-hts-12} we describe the reduction of  $\Sigma$ to a standard quadrangulation, called a \define{system of quads} by Erickson and Whittlesey~\cite{ew-tcsr-13}. We first consider the case of a surface without boundary, i.e., when $\Sigma$ has no perforated faces. After reducing $\Sigma$ to a surface $\Sigma'$ with a single vertex $v$ and a single face $f$ this system of quads is obtained by adding a vertex $w$ at the center of $f$, adding edges between $w$ and all occurrences of $v$ in the facial walk of $f$, and finally deleting the edges of $\Sigma'$. See left and middle Figure~\ref{fig:quadSystem}.
\begin{figure}[h]
  \def\ea{\alpha}
  \def\eb{\beta}
  \def\ec{\gamma}
  \def\ed{\delta}
  \def\ee{\varepsilon}
  \def\ef{\zeta}
  \def\eg{\eta}
  \def\eh{\theta}
  \centering
  \includesvg[\linewidth]{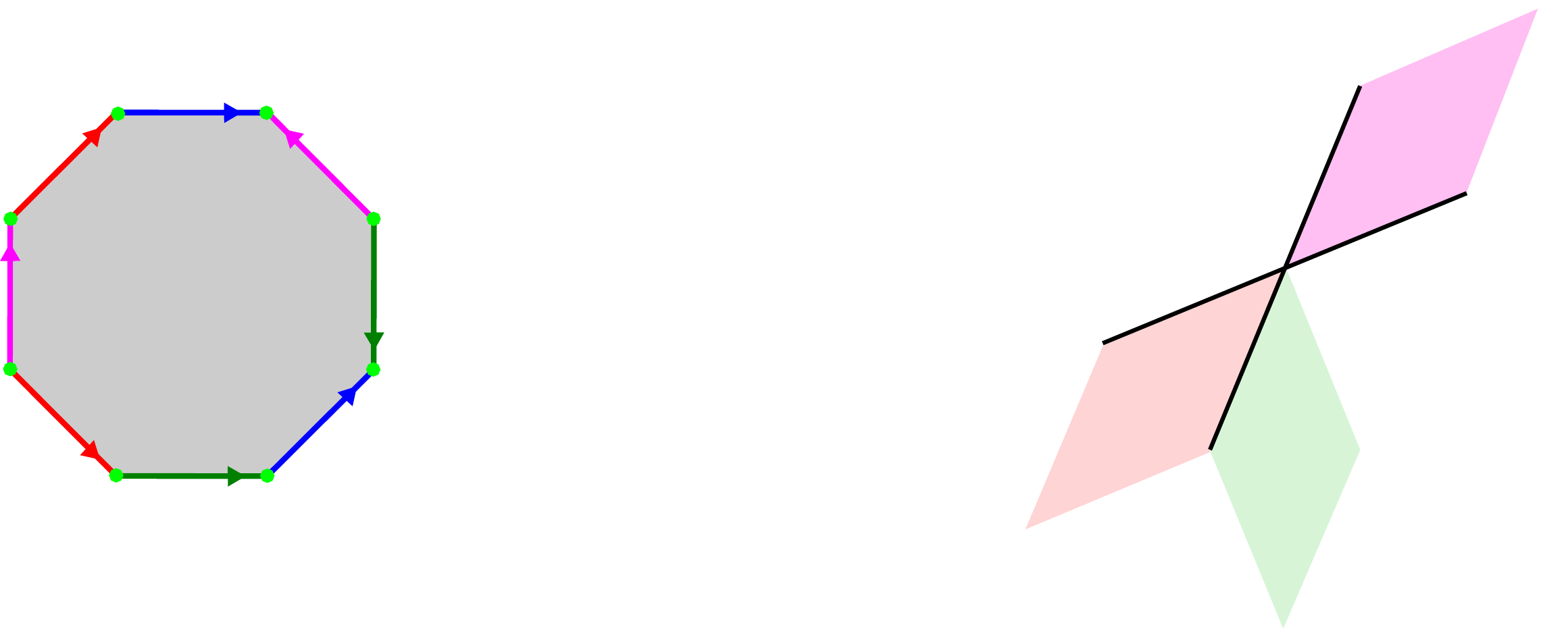}
  \caption{Left, a planar unfolding of a reduced orientable surface of genus two after cutting along its edges. Middle, the corresponding system of quads with its radial graph. Each face is a quadrilateral obtained by gluing two triangular faces along a deleted edge of the reduced surface. Right, a disk diagram for the contractible curve $(\eb,\ea\inv,\ed,\eb\inv,\ea,\ee\inv,\ef,\ef\inv,\eg,\ef\inv,\ee,\ec\inv,\eb)$. The unbounded face corresponds to the perforated face in $\Delta$.}
  \label{fig:quadSystem}
\end{figure}
The graph of the resulting system of quads, called the \define{radial graph}~\cite{lr-hts-12}, is bipartite. It contains two vertices, namely $v$ and $w$, and $4g$ edges (resp. $2g$ edges) if the surface is orientable (resp. non-orientable), where $g$ is the genus of $\Sigma$. All its faces are quadrilaterals. A similar construction applies when $\Sigma$ has perforated faces. In this case, the reduced surface $\Sigma'$ may have several plain faces, each surrounded by perforated faces. 
We essentially perform the same sequence of operations to each plain face. However, when deleting edges of $\Sigma'$ incident to both a perforated and a plain face, we merge the two faces to obtain a larger perforated face\footnote{In general, we cannot afford merging the plain faces with the adjacent perforated faces in a systematic manner to get a surface with perforated faces only. Indeed, when transforming a path of $\Sigma$ into a homotopic path of this reduced version this may increase the complexity of the path by a factor of $g$.}. In the end, the plain faces are quadrilaterals and every vertex is incident to a perforated face. In order to get a bipartite graph, we eventually subdivide the remaining loop edges by introducing a vertex in the middle. 

Note that the system of quads can be deduced from $\Sigma$ by a sequence of edge contractions, deletions, insertions and  subdivisions. Every cycle $c$ of $\Sigma$ can thus be modified accordingly to give a cycle $\varphi(c)$ in the system of quads. Clearly, there exist cellular  embeddings $\rho, \rho'$ of the graphs of $\Sigma$ and of the system of quads in a same topological surface such that $\rho(c)\sim\rho'(\varphi(c))$. 
\begin{lem}[\cite{dg-tcs-99,lr-hts-12}]\label{lem:systems-of-quads}
  Let $n$ be the number of edges of the combinatorial surface $\Sigma$. The above construction of a system of quads can be performed in $O(n)$ time so that for every closed curve $c$ of length $\ell$ in $\Sigma$, we can compute in $O(\ell)$ time a curve $c'$ of length at most $2\ell$ in the system of quads such that $c'\sim \varphi(c)$.
\end{lem}
For the rest of the paper, unless stated otherwise, we shall assume that \define{all surfaces have negative Euler characteristic}. 

\paragraph{Diagrams.}
A \define{disk diagram} over the combinatorial surface $\Sigma$ is a combinatorial sphere $\Delta$ with one perforated face together with a labelling of the arcs of $\Delta$ by the arcs of $\Sigma$ such that
\begin{enumerate}
\item opposite arcs receive opposite labels, 
\item the facial walk of each plain face of $\Delta$ is labelled by the facial walk of some plain face of $\Sigma$.
\end{enumerate}
The diagram is \define{reduced} when no edge of $\Delta$ is incident to two plain faces labelled by the same facial walk (with opposite orientations) of $\Sigma$. An \define{annular diagram} is defined similarly by a combinatorial sphere with two distinct perforated faces. A vertex of a diagram that is not incident to any perforated face is said \define{interior}.
\begin{lem}[van Kampen, See~{\cite[Sec. 2.4]{ew-tcsr-13}}] \label{lem:diagram}
  A cycle of $\Sigma$ is contractible if and only if it is the label of the facial walk of the perforated face of a reduced disk diagram over $\Sigma$. Two cycles are freely homotopic if and only if the facial walks of the two perforated faces of a reduced annular diagram over $\Sigma$ are labelled by these two cycles respectively.
\end{lem}
Figure~\ref{fig:quadSystem} (right) shows a disk diagram for a contractible curve. Note that two plain faces that are adjacent and consistently oriented in a reduced diagram must be labelled by adjacent faces that are consistently oriented in $\Sigma$. Moreover, the degree of an interior vertex of the diagram is a multiple of the degree of  the corresponding vertex in $\Sigma$.  In the sequel, \define{all the considered diagrams will be supposed reduced}. 

\paragraph{Spurs, brackets and canonical curves.}
Thanks to Lemma~\ref{lem:systems-of-quads} we may assume that our combinatorial surface $\Sigma$ is a system of quads. Moreover, the construction of this system of quads with the assumption on the Euler characteristic implies that \define{all interior vertices have degree at least 6}. For surfaces without boundary, this follows from the fact that the two vertices of $\Sigma$ have degree $4g$ in the orientable case and $2g$ otherwise. For surfaces with non-empty boundary, the claim is trivial since the system of quads has no interior vertices.
Let $(a_1,a_2)$ be a pair of arcs sharing their origin vertex $v$ and such that the face corners between $a_1$ and $a_2$ in the direct cyclic order around $v$ belong to 
plain faces.
Following the terminology of Erickson and Whittlesey~\cite{ew-tcsr-13}, we define the \define{turn} of $(a_1,a_2)$ as the number of those face corners. Hence, if $v$ is an interior vertex of degree $d$ in $\Sigma$, the turn of $(a_1,a_2)$ is an integer modulo $d$  that is zero when $a_1=a_2$. If $v$ is a boundary vertex and one of the faces  between $a_1$ and $a_2$ (in the direct order) is
perforated, while none of the faces between $a_2$ and $a_1$ is perforated, then the turn of $(a_1,a_2)$ is minus the number of face corners between $a_2$ and $a_1$. The turn is undefined when one of the faces between $a_1$ and $a_2$  and one of the faces between $a_2$ and $a_1$ are perforated.
When $\Sigma$ is oriented, the \define{turn sequence} of a subpath $(a_i,a_{i+1},\dots,a_{i+j-1})$ of a closed curve of length $\ell$ is the  sequence of $j+1$ turns of $(a_{i+k}\inv,a_{i+k+1})$ for $-1\leq k< j$, where indices are taken modulo $\ell$. The subpath may have length $\ell$, thus leading to a sequence of $\ell+1$ turns. On a non-necessarily oriented surface the sign of the turn of $(a_{i+k}\inv,a_{i+k+1})$ should be changed according to the parity of the number of half-twisted arcs among $(a_i,a_{i+1},\dots,a_{i+k})$. This amounts to untwisting the path $(a_i,a_{i+1},\dots,a_{i+k})$ before computing the turn sequence. A \define{spur} in a curve is a subpath of the form $(a,a\inv)$. Hence,
the turn of $(a_{i+k}\inv,a_{i+k+1})$ is zero precisely when $(a_{i+k},a_{i+k+1})$ is a spur. A \define{bracket} is any subpath whose turn sequence has the form $12^*1$ or $\bar{1}\bar{2}^*\bar{1}$ where $t^*$ stands for a possibly empty sequence of turns $t$ and $\bar{x}$ stands for $-x$.

The next theorem follows from Lemma~\ref{lem:diagram} and a  combinatorial version of Gauss-Bonnet theorem~\cite{gs-sctag-90}.
\begin{thm}[\cite{gs-sctag-90,ew-tcsr-13}]\label{th:four-vertex}
A nontrivial contractible closed curve on a system of quads must have either a spur or four brackets. Moreover, if the curve is the label of the boundary walk
(i.e., of the facial walk of the perforated face)
of a disk diagram with at least one interior vertex, then the curve must have either a spur or five brackets.
\end{thm}

Lazarus and Rivaud~\cite{lr-hts-12} have introduced a canonical form for every nontrivial free homotopy class of closed curves in an oriented system of quads. In particular, two curves are freely homotopic if and only if their canonical forms are equal (up to a circular shift of their vertex indices). This canonical form is defined as the rightmost generator of a certain cylindrical covering surface of the system of quads. It was further characterized by Erickson and Whittlesey~\cite{ew-tcsr-13} in terms of turns and brackets. The \define{canonical form} (of the free homotopy class) of a curve is the homotopic curve that contains no spurs or brackets and whose turning sequence contains  no $-1$'s and contains at least one turn that is not $-2$.
\begin{thm}[\cite{lr-hts-12},\cite{ew-tcsr-13}]\label{th:canonical}
  Up to a circular shift of its vertex indices, the canonical form of
  a combinatorial closed curve of length $\ell$ on an oriented system
  of quads is uniquely defined. It can be computed in $O(\ell)$ time.
\end{thm}

\subsection{Geodesics}\label{subsec:geodesics}
A \define{combinatorial geodesic} of length $\ell$ is a curve that contains no spurs or brackets, except maybe a bracket of length $\ell-1$ if the curve is \emph{exceptional}. See the notion of exceptional curve after Corollary~\ref{cor:equal-length}. Given a curve on a system of quads, orientable or not, we can compute a homotopic geodesic by iteratively removing spurs and brackets as much as possible.  The resulting curve are said \emph{reduced} in \cite{ew-tcsr-13}. The proof therein applies to orientable systems of quads but trivially extends to the non-orientable case.
\begin{thm}[{\cite[Lem. 4.3]{ew-tcsr-13}}]\label{th:compute-geodesic}
  Given a combinatorial closed curve of length $\ell$ on a system of quads, a homotopic geodesic can be computed in $O(\ell)$ time. 
\end{thm}
On an oriented surface, the canonical form is an instance of combinatorial geodesic; this is the rightmost homotopic geodesic. The definitions of a geodesic and of a canonical form extend naturally to paths: a \define{geodesic path} is a path  that contains no spurs or brackets. It is 
\define{canonical} (or in canonical form) if its turning sequence contains  no $-1$'s. Theorem~\ref{th:canonical} extends trivially to paths: every path has a unique homotopic geodesic in canonical form.

Although we cannot claim in general the uniqueness of geodesics in a homotopy class, homotopic geodesics are almost equal and have the same minimal length. Specifically, define a \define{(quad) staircase} as a planar sequence of quads obtained by stitching an alternating sequence of rows and columns of quads  to get the shape of a staircase. Assuming that the staircase goes up from left to right, we define the \define{initial tip} of a quad staircase as the lower left vertex of the first quad in the sequence. The \define{final tip}  is defined as the upper right vertex of the last quad. See Figure~\ref{fig:lentille-plate}. A \define{closed staircase} is obtained by identifying the two vertical arcs incident to the initial and final tips of a staircase.
\begin{thm}\label{th:geodesic}
  Let $c,d$ be two nontrivial homotopic combinatorial geodesics. 
If $c,d$ are closed curves, then they label the two boundary cycles of an annular diagram composed of a single closed staircase or of an alternating sequence of paths (possibly reduced to a vertex) and quad staircases connected through their tips. 
Likewise, if $c,d$ are paths, then 
the closed curve $c\cdot d\inv$ labels the boundary of a disk diagram composed of an alternating sequence of paths (possibly reduced to a vertex) and quad staircases connected through their tips. 
\end{thm}
\begin{proof}
We only detail the proof when $c,d$ are paths. See Figure~\ref{fig:lentille-plate}.
\begin{figure}
  \centering
  \includesvg[.6\textwidth]{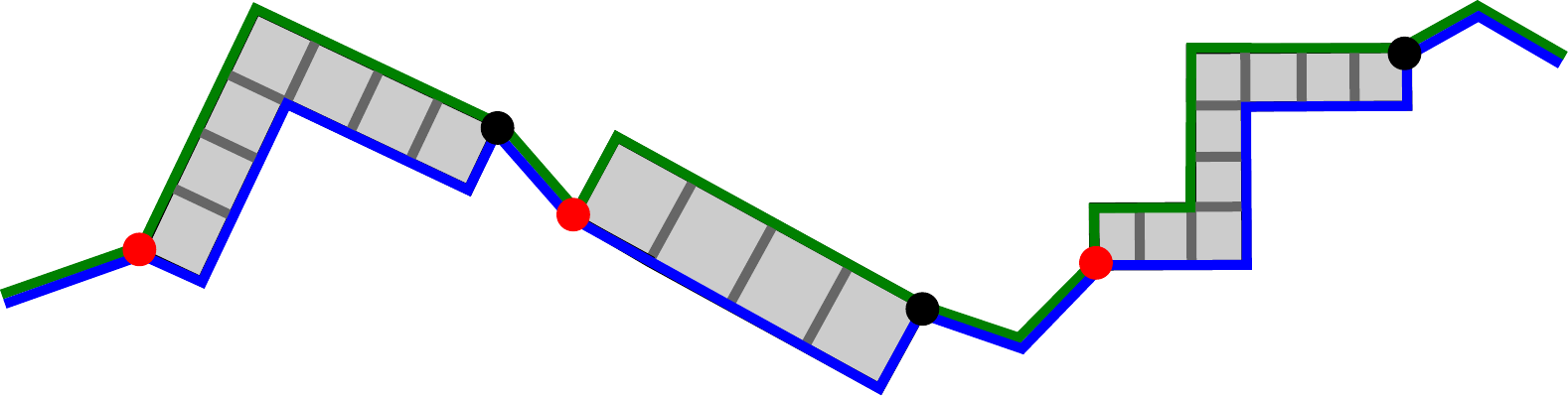}
  \caption{A disk diagram for two homotopic paths $c$ and $d$ composed of paths and staircases.}
  \label{fig:lentille-plate}
\end{figure}
The similar case of closed curves is covered in~\cite{ew-tcsr-13}. By Lemma~\ref{lem:diagram}, $c\cdot d\inv$ is the label of the facial walk of the perforated face of a disk diagram $\Delta$. This diagram has a cactus-like structure composed of 2-cells subdivided into quads and connected by trees. A vertex $v$ such that $\Delta\setminus v$ is not connected is called a \emph{cut vertex}. We also consider as cut vertices the endpoints of $c$ in $\Delta$. A 2-cell of $\Delta$ must have more than one cut vertex on its boundary. Otherwise, this boundary is entirely labelled by a subpath of either $c$ or $d$ and Theorem~\ref{th:four-vertex} implies the existence of four brackets, one of which (in fact two) must avoid the cut vertex, hence be contained in the interior of this subpath. This would contradict that $c$ and $d$ are geodesic. Moreover, because $c$ and $d$ have no spur, $\Delta$ cannot have more than two degree one vertices, namely the common endpoints of $c$ and $d$. It follows that $\Delta$ is an alternating sequence of paths and 2-cells. In particular, each 2-cell has exactly two cut vertices. No 2-cell in this sequence has an interior vertex. For otherwise, by the second part of Theorem~\ref{th:four-vertex}, the boundary of this 2-cell would contain five brackets one of which would be contained in the interior of either $c$ or $d$ and this would again contradict that $c$ and $d$ are geodesic.
It follows that the dual of a 2-cell, viewed as an assembling of quads, is a tree. We finally remark that this tree must be a path with a staircase shape. Indeed, any other shape would imply the existence of a bracket in either $c$ or $d$.
\end{proof}
Since the tips of a staircase divide its boundary into two sides of equal length, and since removing a bracket or a spur shortens a curve, we easily deduce the following.
\begin{cor}\label{cor:equal-length}
 With the hypothesis of Theorem~\ref{th:geodesic}, $c$ and $d$ have equal length which is minimal among homotopic curves.
\end{cor}
It will often be useful to think of a disk diagram bounded by homotopic geodesic paths as a subset of the universal cover $\tilde{\Sigma}=\D$ of the system of quads. Indeed, by the remark after Lemma~\ref{lem:diagram}, the sequence of faces along a staircase is labelled by a sequence of faces arranged the same way in $\tilde{\Sigma}$. Beware, however, that two successive staircases in a diagram may correspond to opposite orientations in $\tilde{\Sigma}$. Some staircases should thus be flipped in order to get an exact image of the diagram in $\tilde{\Sigma}$. A consequence of this representation is that homotopic geodesic paths have a unique disk diagram. This follows directly from the unique lifting property of coverings.

\paragraph{Exceptional curves}
A geodesic curve of length $\ell$ is \define{exceptional} if it has a bracket of length $\ell-1$. An exceptional curve must be 1-sided, since otherwise the curve would have a bracket of length 1. Figure~\ref{fig:exceptional-curve} shows the lift of an exceptional curve; it is contained in an infinite straight strip of quadrilaterals that covers a M\"obius strip infinitely many times.
\begin{figure}[h]
  \centering
  \includesvg[.6\linewidth]{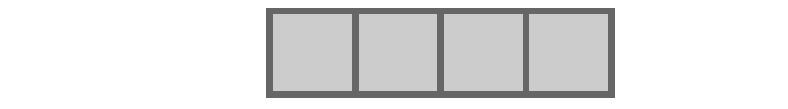}
  \caption{The lift of an exceptional curve of length 4 has brackets of length 3.}
  \label{fig:exceptional-curve}
\end{figure}
\begin{cor}\label{cor:no-monogon}
  A geodesic has no nontrivial index path whose image path is contractible.
\end{cor}
\begin{proof}
  Let $c$ be a geodesic. If $c$ is not exceptional, it follows from the definition of geodesics that the image path of any index path of $c$ has no spur or brackets, hence is also a (non-contractible) geodesic. If $c$ is exceptional, the possible brackets of its index paths can be flattened to get a homotopic non-trivial geodesic. 
\end{proof}
The next two remarks follow directly from the characterization of geodesics and canonical forms in terms of spurs, brackets and turns.
\begin{remark}\label{rem:sub-geodesic}
  The image path of any index path of a combinatorial geodesic that is not  exceptional  is geodesic. If the combinatorial geodesic is in canonical form, so is the image path.
\end{remark}

\begin{remark}\label{rem:geodesic-power}
 Likewise, any power $c^k$ of a combinatorial closed geodesic $c$ that is not  exceptional is also a combinatorial geodesic. Moreover, if $c$ is in canonical form, so is $c^k$.
\end{remark}

\section{Computing intersection numbers on oriented surfaces}\label{sec:counting-on-oriented}
Here, we assume that the considered system of quads is orientable and consistently oriented. In this case we can rely on canonical forms to obtain a relatively simple algorithm for the computation of the geometric intersection number. We shall see in Proposition~\ref{prop:counting-crossings} that the set $B/\tau$ of Lemma~\ref{lem:strategy}, composed of equivalence classes of lifts, can be identified with certain \emph{crossing double paths}. Those are pairs of index paths with coincident images, see Section~\ref{subsec:double-paths}. The asymptotic complexity of the resulting algorithm is the same as for the general case, where the system of quads may be non-orientable. The general case leads to a slightly more complicated algorithm and is deferred to Section~\ref{sec:counting}. 

The next technical Lemma will be used in Proposition~\ref{prop:counting-crossings} to analyze the intersection of canonical curves.
Let $c_R$ and $c\inv_L$ be canonical paths such that $c_R\sim c_L$. In other words, $c_L$ is the leftmost geodesic homotopic to $c_R$. By Theorem~\ref{th:geodesic}, there is a disk diagram $\Delta$ composed of quad staircases and paths and whose left and right boundaries $\Delta_L$ and $\Delta_R$ are labelled by $c_L$ and $c_R$ respectively. 
A \define{spoke} is a non-boundary edge of $\Delta$.
\begin{lem}\label{lem:double-path-struct}
Let $v,w$ be two vertices, one on each side of $\Delta$. Then $\Delta$ contains a path  $p$ from $v$ to $w$ labelled by a canonical path. 
Moreover, $p$ can be uniquely decomposed as either $\lambda.\rho$, $\rho.\lambda$,   $\lambda.e.\rho$  or $\rho.e.\lambda$, where
\begin{enumerate}
\item\label{it:lambda}  $\lambda$ is a subpath (possibly reduced to a vertex) of $\Delta_L$ or $\Delta\inv_L$,
\item $\rho$ is a subpath (possibly reduced to a vertex) of $\Delta_R$ or $\Delta\inv_R$, 
\item $e$ is a spoke,
\item\label{it:neg-L} if $\lambda$ is a subpath of $\Delta\inv_L$ of positive length then $p\cap \Delta\inv_L = \lambda$,
\item\label{it:pos-R} if $\rho$ is a subpath of $\Delta_R$ of positive length then $p\cap \Delta_R = \rho$.
\item\label{it:both-neg} if $\rho$ is a subpath of $\Delta\inv_R$ and $\lambda$ is a subpath of $\Delta_L$ then either $\rho$ is reduced to a vertex and $p\cap \Delta_R = \rho$ or $\lambda$ is reduced to a vertex and $p\cap \Delta\inv_L = \lambda$.
\end{enumerate}
\end{lem}
\begin{proof}
  We assume that $v$ is on the left side and $w$ on the right side of $\Delta$, the other case being symmetric. Let $i,j$ be such that $v=\Delta_L(i)$ and $w=\Delta_R(j)$. We first consider the case where $j\geq i$. If $\Delta_L$ and $\Delta_R$ coincide at $v$, then we trivially obtain the desired decomposition as $p=\Delta_L\Ipath{i}{0}.\Delta_R\Ipath{i}{j-i}$. Otherwise, $v$ may be incident to $0,1$ or $2$ spokes.
  \begin{itemize}
  \item If $v$ is not incident to a spoke then $\Delta_L(i-1)$ is either the initial tip of a staircase or is incident to a spoke $e$. We set $q = \Delta_L\Ipath{i}{-1}.\Delta_R\Ipath{i-1}{j-i+1}$ in the first case and $q =\Delta_L\Ipath{i}{-1}.e.\Delta_R\Ipath{i-1}{j-i+1}$ otherwise. The path $q$ has no spurs or $-1$ turns but may start with a bracket. If not, by the characterization of canonical geodesic paths in Section~\ref{subsec:geodesics}, $q$ is already canonical and we can set $p=q$. Otherwise, we short cut the bracket in $q$ to obtain a canonical path $p$ satisfying the above points~\ref{it:lambda} to~\ref{it:pos-R}. See Figure~\ref{fig:structure-geodesic}. 
    \begin{figure}
      \centering
      \includesvg[\textwidth]{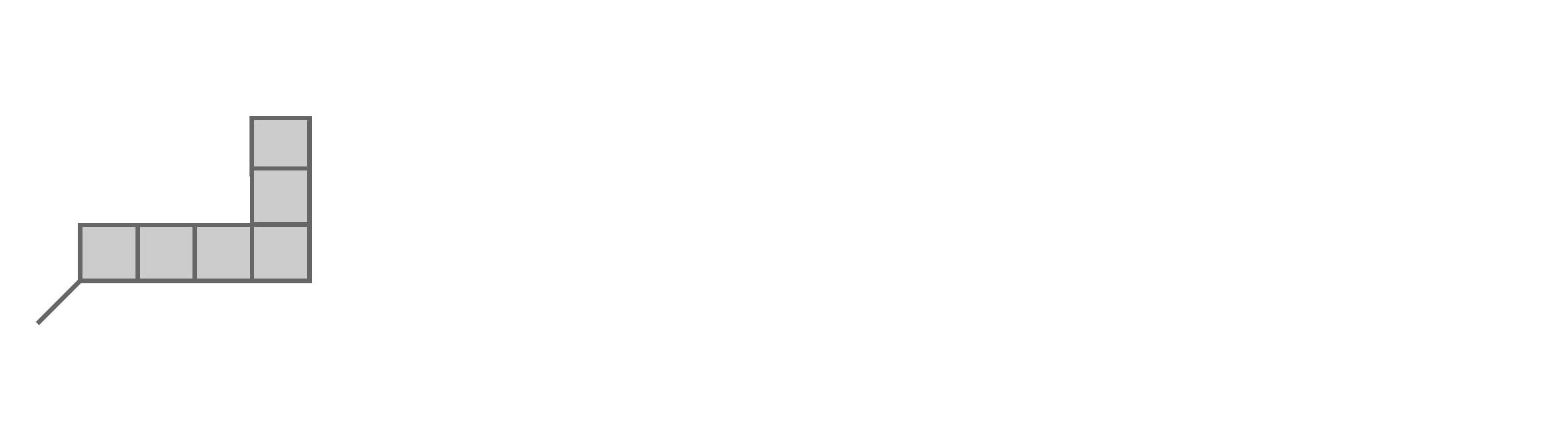}
      \caption{The canonical path $p$ from $v$ to $w$ when $v$ is not incident to a spoke.}
      \label{fig:structure-geodesic}
    \end{figure}
\item If $v$ is incident to exactly one spoke $e$,
then $e$ connects $v$ to either  $\Delta_R(i-1)$ or $\Delta_R(i+1)$. In the former case we set $q=e.\Delta_R\Ipath{i-1}{j-i+1}$. As above, $q$ may be canonical or starts with a bracket and we easily obtain the path $p$ with the desired properties. See Figure~\ref{fig:structure-geodesic-I}. 
    \begin{figure}
      \centering
      \includesvg[\textwidth]{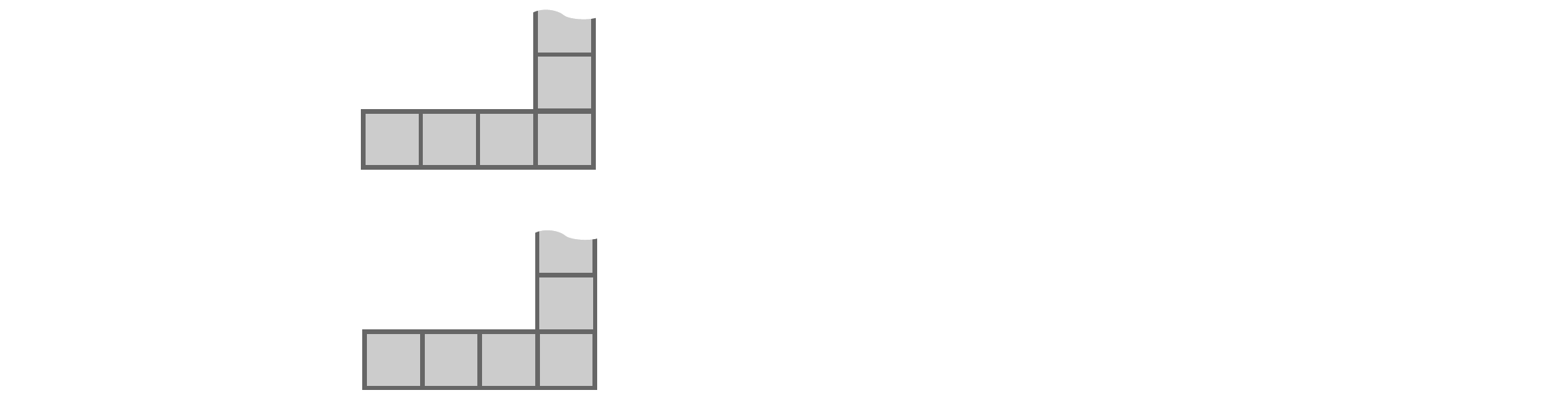}
      \caption{The canonical path $p$ from $v$ to $w$ when $v$ is incident to exactly one spoke.}
      \label{fig:structure-geodesic-I}
    \end{figure}
If $e$ is incident to $\Delta_R(i+1)$ and
$j>i$ the path $p=\Delta_L\Ipath{i}{0}.e.\Delta_R\Ipath{i+1}{j-i-1}$ has the required properties. 
When $e$ is incident to $\Delta_R(i+1)$ and $j=i$, then either 
 $\Delta_L(i-1)$ is incident to a spoke $e'$ and we must have $p=\Delta_L\Ipath{i}{-1}.e'.\Delta_R\Ipath{i}{0}$, or $\Delta_L(i-1)$  is the initial tip of a staircase and we must have $p=\Delta_L\Ipath{i}{-1}.\Delta_R\Ipath{i-1}{1}$. 
\item If $v$ is incident to two spokes, then one of them, $e_-$, connects $v$ to $\Delta_R(i-1)$ and the other $e_+$ connects $v$ to $\Delta_R(i+1)$. We can directly set $p=e_-.\Delta_R\Ipath{i-1}{1}$ if $j=i$ and $p=e_+.\Delta_R\Ipath{i+1}{j-i-1}$ if $j>i$.
  \end{itemize}
For Point~\ref{it:both-neg} in the Lemma, we note that we cannot have $p = \lambda.\rho$ or $p=\lambda.e.\rho$  with both $\lambda$ and $\rho$ being  subpaths of positive length of $\Delta_L$  and $\Delta\inv_R$ respectively. Indeed, $p$ would have a spur or a $\bar{1}$ turn in the first case and a bracket in the other case. Moreover, if $\rho$ is reduced to a vertex  and $\lambda$ is a subpath of $\Delta_L$,  we may assume that  the intersection $p\cap\Delta_R$ is reduced to $\rho$. Otherwise, the canonical path between any other intersection point and $\rho$ would have to follow $\Delta_R$ and we could express $p$ so that $\rho$ is a subpath of positive length of $\Delta_R$. An analogous argument holds to show that we can assume $p\cap\Delta\inv_L=\lambda$ if $\lambda$ is reduced to a vertex.

We next consider the case $i>j$. When $\Delta_L$ and $\Delta_R$ coincide at $w$, we obtain the desired decomposition as $p=\Delta_L\Ipath{i}{j-i}.\Delta_R\Ipath{i}{0}$. Otherwise, $w$ may be incident to $0,1$ or $2$ spokes. Similar arguments as in the case $j\geq i$ allow to conclude the proof. 
\end{proof}

\subsection{Crossing double-paths}\label{subsec:double-paths}
Let $c,d$ be two combinatorial closed curves on a combinatorial surface. 
A \define{double-path} of $(c,d)$ of length $\ell$ is a pair of forward index paths $(\Ipath{\Imath}{\ell}_c,\Ipath{\Jmath}{\ell}_d)$ with the same
image path $c\Ipath{\Imath}{\ell}=d\Ipath{\Jmath}{\ell}$. If $\ell=0$ then  the double path is just a double point. A double path of $c$ 
is defined similarly, taking $c=d$ and assuming $\Imath\neq \Jmath$. 
The next Lemma follows from Remark~\ref{rem:sub-geodesic}.
\begin{lem}\label{lem:double-path}
  Let $\Ipath{\Imath}{\ell}_c$ and $\Ipath{\Jmath}{k}_d$ be forward index paths of two canonical curves $c$ and $d$ such that the image paths $c\Ipath{\Imath}{\ell}$ and $\Ipath{\Jmath}{k}$ are homotopic. Then $k=\ell$ and $(\Ipath{\Imath}{\ell}_c,\Ipath{\Jmath}{\ell}_d)$ is a double path.
\end{lem}
A double path $(\Ipath{\Imath}{\ell}_c,\Ipath{\Jmath}{\ell}_d)$ gives rise to a sequence of $\ell+1$ double points $\Dpoint{\Imath+k}{\Jmath+k}$ for $k\in [0,\ell]$. A priori a double point could occur several times in this sequence. The next two lemmas claim that this is not possible when the curves are primitive.
\begin{lem}\label{lem:short-double-path-c}
  A double path of a primitive combinatorial curve $c$ cannot contain a double point more than once in its sequence. In particular, a double path of $c$ must be strictly shorter than $c$.
\end{lem}
\begin{proof}
Suppose that a double path $\dP$ of $c$ contains two occurrences of a double point $(\Ibar,\Jbar)$. Because the couples $(\Ibar,\Jbar)$ and $(\Jbar,\Ibar)$ represent the same double point there are two cases to consider. 
\begin{itemize}
\item If $\dP$ contains the couple $(\Ibar,\Jbar)$ twice then it must contain a subsequence of length $|c|$ starting with $(\Ibar,\Jbar)$. We thus have $c\Ipath{\Imath}{|c|}=c\Ipath{\Jmath}{|c|}$. This implies that $c$ is equal to some nontrivial circular permutation of itself. It is a simple exercise to check that $c$ must then be a proper power of some other curve, contradicting that $c$ is primitive. 
\item Otherwise $\dP$ contains $(\Ibar,\Jbar)$ and $(\Jbar,\Ibar)$. Let $\ell$ be the distance between these two occurrences in $\dP$. We thus have $c\Ipath{\Imath}{\ell} = c\Ipath{\Jmath}{|c|-\ell}$ from which we deduce that  $c$ is a square (and $\ell=|c|/2$), contradicting that $c$ is primitive. 
\end{itemize}
\end{proof}

\begin{lem}\label{lem:short-double-path-c-d}
  Let $c$ and $d$ be two non-homotopic primitive combinatorial curves. A double path of $(c,d)$ cannot contain a double point more than once in its sequence. Moreover, the length of a double path of $(c,d)$ must be less than $|c|+|d|-1$.
\end{lem}
\begin{proof}
  Suppose that a double path of $(c,d)$ contains two occurrences of a double point. After shortening the double path if necessary, we may assume that these two occurrences are the first and the last double points of the double path. Its length must accordingly be a nonzero integer multiple $p$ of $|c|$ as well as a nonzero integer multiple $q$ of $|d|$. It follows that for some circular permutations $c'$ of $c$ and $d'$ of $d$ we have $c'^p = d'^q$. By a classical result of combinatorics on words~\cite[Prop. 1.3.1]{l-cw-97} this implies that $c'$ and $d'$ are powers of a same curve, in contradiction with the hypotheses in the lemma. In fact, by a refinement due to Fine and Wilf~\cite[Prop. 1.3.5]{l-cw-97} it suffices that $c'^p$ and $d'^q$ have a common prefix of length $|c|+|d|-1$ to conclude that $c'$ and $d'$ are powers of a same curve. This proves the second part of the lemma.
\end{proof}

A double path whose index paths cannot be extended is said \define{maximal}. 
As an immediate consequence of Lemmas~\ref{lem:short-double-path-c} and ~\ref{lem:short-double-path-c-d} we have:
\begin{cor}\label{cor:partition-double-paths}
  The maximal double paths of a primitive curve or of two primitive curves in canonical form induce a partition of the double points of the curves. 
\end{cor}
Let $(\Ibar,\Jbar)$ and $\Dpoint{\Imath+\ell}{\Jmath+\ell}$ be the first and the last double points of a maximal double path of $(c,d)$, possibly with $c=d$. When $\ell\geq 1$ the arcs $c[\Ibar,\Ibar-1]$, $d[\Jbar,\Jbar-1]$, $c[\Ibar,\Ibar+1]$ must be pairwise distinct because canonical curves have no spurs, and similarly for the three arcs $c[\Overline{\Imath+\ell},\Overline{\Imath+\ell+1}]$, $d[\Overline{\Jmath+\ell},\Overline{\Jmath+\ell+1}]$, $c[\Overline{\Imath+\ell},\Overline{\Imath+\ell-1}]$. We declare the maximal double path to be a \define{crossing double path} if the circular ordering of the first three arcs at $c(\Ibar)$ and the circular ordering of the last three arcs at $c(\Overline{\Imath+\ell})$ are either both clockwise or both counterclockwise with respect to the rotation system of the system of quads. When $\ell=0$, that is when the maximal double path is reduced to the double point $(i,j)$, we require that the arcs  $c[i,i-1],d[j,j-1], c[i,i+1], d[j,j+1]$ are \emph{pairwise distinct} and appear in this circular order, or its opposite, around the vertex $c(i)=d(j)$.

\subsection{Proof of Theorem~\ref{th:main-result} in the primitive and orientable case}\label{subsec:proof-orientable}
Let $c,d$ be primitive combinatorial curves such that $d$ is canonical and let $c_R$ and $c_L\inv$ be the canonical curves homotopic to $c$ and $c\inv$ respectively. We denote by $\Delta$ the annular diagram corresponding to $c_R$ and $c_L$. When the two boundaries $\Delta_R$ and $\Delta_L$ of $\Delta$ have a common vertex we implicitly assume that $c_R$ and $c_L$ are indexed so that this vertex corresponds to the same index along $\Delta_R$ and $\Delta_L$.
We consider the following set of double paths:
\begin{itemize}
\item $\mathcal{D}_+$ is the set of crossing double paths of positive length of $c_R$ and $d$,
\item $\mathcal{D}_0$ is the set of crossing double paths $(i,j)$ of zero length of $c_R$ and $d$ such that either 
  \begin{itemize}
\item the two boundaries of $\Delta$ coincide at $\Delta_L(i) = \Delta_R(i)$ and $d[j-1,j] = c_L[i-1,i]$ or $d[j,j+1] = c_L[i,i+1]$, or
\item one of $d[j,j-1]$ or $d[j,j+1]$ is the label of a spoke $(\Delta_R(i),\Delta_L(i'))$ of $\Delta$ and $d[j-2,j-1] = c_L[i'-1,i']$ in the first case or $d[j+1,j+2] = c_L[i',i'+1]$ in the other case.
  \end{itemize}
\item $\mathcal{D}_-$ is the set of crossing double paths  $(\Ipath{\Imath}{\ell}_{c_L\inv},\Ipath{\Jmath}{\ell}_d)$ ($\ell \geq 0$) of $c_L\inv$ and $d$ such that \emph{none} of the following situations occurs:
  \begin{itemize}
\item the two boundaries of $\Delta$ coincide at $\Delta_L\inv(i) = \Delta_R(i')$ and $d[j-1,j] = c_R[i'-1,i']$,
\item  the two boundaries of $\Delta$ coincide at $\Delta_L\inv(i+\ell) = \Delta_R(i')$  and $d\Focc{j+\ell} = c_R\Focc{i'}$,
\item $d[j-1,j]$ is the label of a spoke $(\Delta_L\inv(i), \Delta_R(i'))$ of $\Delta$ and $d[j-2,j-1] = c_R[i'-1,i']$,
\item $d[j+\ell,j+\ell+1]$ is the label of a spoke $(\Delta_L\inv(i+\ell), \Delta_R(i'))$ of $\Delta$ and $d[j+\ell+1,j+\ell+2] = c_R[i',i'+1]$.
  \end{itemize}
\end{itemize}
In case $c\sim d$ or $c\sim d\inv$, which we can detect by comparing their canonical forms, we enforce $c_R=d$, inverting $d$ if necessary. This way, recalling that the index paths of a double path of $c$ must be distinct by definition, the maximal double paths have finite length by Lemma~\ref{lem:short-double-path-c}.  Figure~\ref{fig:Dplus-D0-Dminus} depicts some configurations.
\begin{figure}
  \centering
  \includegraphics[width=\textwidth]{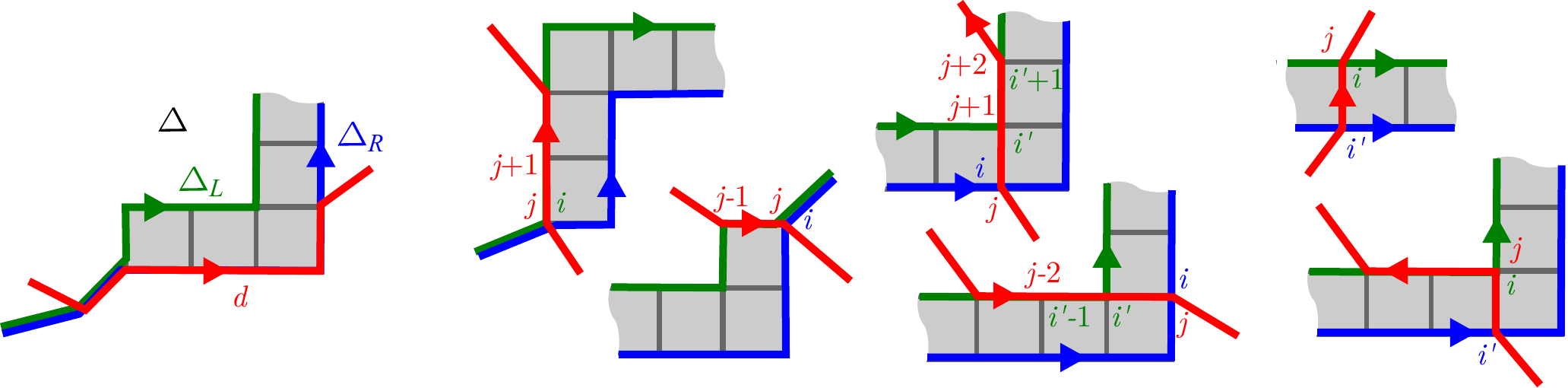}
  \caption{Left, a typical crossing double path in $\mathcal{D}_+$. Middle,  four configurations in $\mathcal{D}_0$. Right, two configurations in $\mathcal{D}_-$.}
  \label{fig:Dplus-D0-Dminus}
\end{figure}
Referring to Section~\ref{sec:strategy}, we view the underlying surface of the system of quads, call it $\Sigma$, as a quotient $\D/\Gamma$ of the Poincar\'e disk. The system of quads lifts to a quadrangulation of $\D$ and the lifts of a combinatorial curve in $\Sigma$ are combinatorial bi-infinite paths in this quadrangulation. By Remark~\ref{rem:geodesic-power}, if the combinatorial curve is geodesic (resp. canonical) so are its lifts. In this case, each lift is simple by Corollary~\ref{cor:no-monogon}. We fix a lift $\tilde{c_R}$ of $c_R$ and consider the set $B/\tau$ of Lemma~\ref{lem:strategy} corresponding to the classes of lifts of $d$ whose limit points alternate with the limit points of $\tilde{c_R}$ along $\partial\D$. 
\begin{prop}\label{prop:counting-crossings}
  $B/\tau$ is in 1-1 correspondence with the disjoint union $\mathcal{D}_+\cup \mathcal{D}_0\cup  \mathcal{D}_-$.
\end{prop}
\begin{proof}
  Let $\tilde{c_L}$ be the lift of $c_L$ with the same limit points as $\tilde{c_R}$. These two lifts  project onto the boundaries of the annular diagram $\Delta$ and thus form an infinite strip $\tilde{\Delta}$ of width at most 1 in $\D$ composed of paths and quad staircases (possibly a single infinite staircase). We shall define a correspondence between $B/\tau$ and $\mathcal{D}_+\cup \mathcal{D}_0\cup  \mathcal{D}_-$.
To this end we consider 
a lift $\tilde{d}$ of $d$ whose limit points alternate with those of $\tilde{c_R}$. In other words, $\tilde{d}\in B$. The lift $\tilde{d}$ must cross $\tilde{\Delta}$. Let $i$ and $j$ be respectively the smallest and largest index $k$ such that $\tilde{d}(k)$ is in $\tilde{\Delta}$. By Remark~\ref{rem:sub-geodesic}, the corresponding subpath $\tilde{d}[i,j]$ of $\tilde{d}$ is canonical. Since $\D$ is simply connected, $\tilde{d}[i,j]$ is homotopic to any path joining the same extremities and we can apply Lemma~\ref{lem:double-path-struct} to show that $\tilde{d}[i,j]$ is actually contained in $\tilde{\Delta}$ and that it can be decomposed as either $\lambda.\rho$, $\rho.\lambda$,   $\lambda.e.\rho$  or $\rho.e.\lambda$ where $e$ is a spoke of $\tilde{\Delta}$, $\rho=\tilde{c_R}\Ipath{a}{r}$ and $\lambda=\tilde{c_L}\Ipath{b}{\ell}$ for some $a,b,r,\ell\in \Z$. 
  \begin{itemize}
  \item If $r>0$ then by Point~\ref{it:pos-R} of Lemma~\ref{lem:double-path-struct} we have $\tilde{d}\cap \tilde{c_R}=\rho$ so that this intersection defines a maximal double path of  $\tilde{d}$ and $\tilde{c_R}$. It must be crossing since  $\tilde{d}\in B$. Its projection on $\Sigma$ is a crossing double path of length $r>0$ of $d$ and $c_R$, hence in $\mathcal{D}_+$, to which we map $\tilde{d}$.
\item If $r=0$ and $\ell>0$, then Point~\ref{it:both-neg} of Lemma~\ref{lem:double-path-struct} implies $\tilde{d}\cap \tilde{c_R}=\rho$. As above, $\rho$ must define a crossing double path of  length zero of $\tilde{d}$ and $\tilde{c_R}$. We map $\tilde{d}$ to the projection of this crossing double point on $\Sigma$ and remark that this projection is in $\mathcal{D}_0$.
\item Otherwise, we must have $r\leq 0$ and $\ell\leq 0$ by Point~\ref{it:both-neg} of Lemma~\ref{lem:double-path-struct}. If $\ell <0$ then Point~\ref{it:neg-L} of Lemma~\ref{lem:double-path-struct} implies $\tilde{d}\cap \tilde{c_L}\inv = \lambda$. This intersection defines a crossing double path of $\tilde{d}$ and $\tilde{c_L}\inv$ and we map $\tilde{d}$ to its projection on $\Sigma$. If $\ell=0$ and $r<0$ then Point~\ref{it:both-neg} of Lemma~\ref{lem:double-path-struct} implies $\tilde{d}\cap \tilde{c_L}\inv=\lambda$, which also holds true if $r=\ell=0$. In both cases $\lambda$ corresponds to a crossing double path of  length zero of $\tilde{d}$ and $\tilde{c_L}\inv$ and we map $\tilde{d}$ to its projection on $\Sigma$. We finally remark that this last projection or the above one belong to $\mathcal{D}_-$.
  \end{itemize}

Because $\tilde{\Delta}$ is left globally invariant by $\tau$ (the hyperbolic motion that sends $\tilde{c_R}(0)$ to $\tilde{c_R}(|c_R|)$), we have $\tau(\tilde{d})\cap\tilde{\Delta}= \tau(\tilde{d})\cap\tau(\tilde{\Delta})=\tau(\tilde{d}\cap\tilde{\Delta})$.
It follows that $\tilde{d}$ and $\tau(\tilde{d})$ are mapped to the same crossing double path by the above rules. We thus have a well defined map $B/\tau\to \mathcal{D}_+\cup \mathcal{D}_0 \cup \mathcal{D}_-$. The uniqueness of the decomposition in Lemma~\ref{lem:double-path-struct} implies that this map is 1-1.
In order to check that the map is onto we  consider a maximal crossing double path $\dP=(\Ipath{\Imath \mod |c_R|}{\ell}_{c_R},\Ipath{\Jmath\mod |d|}{\ell}_d)$ of $c_R$ and $d$ in $\mathcal{D}_+$. By the unique lifting property of coverings there is a unique lift $\tilde{d}$ of $d$ such that $\tilde{d}(j)=\tilde{c_R}(i)$ and $\dP$ lifts to a crossing double path of $\tilde{d}$ and $\tilde{c_R}$. By Lemma~\ref{lem:double-path} this double path is the only intersection of $\tilde{d}$ and $\tilde{c_R}$ so these lifts must have alternating limit points. In other words $\tilde{d}$ is in $B$ and is mapped to $\dP$. A similar argument applies to the crossing double paths of $\mathcal{D}_0$ and $\mathcal{D}_-$.
\end{proof}
This leads to a simple algorithm for computing combinatorial crossing numbers.
\begin{cor}\label{cor:crude-result}
  Let $c,d$ be primitive curves of length at most $\ell$ on an orientable combinatorial surface with complexity $n$. The crossing numbers $i(c,d)$ and $i(c)$ can be computed in $O(n+\ell^2)$ time.
\end{cor}
\begin{proof}
By Lemma~\ref{lem:systems-of-quads} we may assume that the surface is a system of quads. By Theorem~\ref{th:canonical} we may compute the canonical forms of $c, c\inv$ and $d$ in $O(\ell)$ time. According to Proposition~\ref{prop:counting-crossings}, we have
\[i(c,d) = |\mathcal{D}_+|+|\mathcal{D}_0|+|\mathcal{D}_-|
\]
The set $\mathcal{D}_+$ can be constructed in $O(\ell^2)$ time. Indeed, 
since the maximal double paths of $c$ and $d$ form disjoint sets of double points by Corollary~\ref{cor:partition-double-paths}, we just need to traverse the grid $\Z/|c|\Z\times \Z/|d|\Z$ and group the double points into maximal double paths. Those correspond to diagonal segments in the grid that can be computed in time proportional to the size of the grid. We can also determine which double paths are crossing in the same amount of time. Likewise, we can construct the sets $\mathcal{D}_0$ and $\mathcal{D}_-$ in $O(\ell^2)$ time. 
\end{proof}
The cases where $c$ and $d$ are non-primitive or where the Euler characteristic of the input surface is non-negative are dealt with in Section~\ref{subsec:end-of-proof}.

\section{Computing intersection numbers: the general case}\label{sec:counting}
Here, we drop the orientability assumption to consider systems of quads that may be non-orientable. We first relate the intersection number of two primitive combinatorial curves with their number of \emph{crossing pairs of index paths} as defined below. See Lemma~\ref{lem:counting}.
 Let $c,d$ be primitive combinatorial closed geodesics on a system of quads $\Sigma$. We assume that neither of them is exceptional as  defined above Corollary~\ref{cor:no-monogon}. If $c$ is homotopic to $d$ or to its inverse, which we can detect by considering the canonical forms of their squares (see the discussion before Proposition~\ref{prop:primitive-root}), we enforce $c=d$. As in Section~\ref{subsec:proof-orientable}, we view the underlying surface of $\Sigma$ as a quotient $\D/\Gamma$ of the Poincar\'e disk. The system of quads lifts to a quadrangulation of $\D$ and the lifts of $c,d$ in $\Sigma$ are combinatorial bi-infinite paths in this quadrangulation. Note that enforcing $c=d$ whenever $c$ was initially homotopic to $d$ or its inverse ensures that the lifts of $c$ and $d$ have pairwise distinct limit points, unless those lifts are equal (in which case we obviously have $c=d$).
We fix a lift $\tilde{c}$ of $c$ and consider any lift $\tilde{d}$ of $d$ whose limit points alternate with the limit points of $\tilde{c}$ along $\partial\D$. 
It follows from 
Remark~\ref{rem:geodesic-power} and Corollary~\ref{cor:equal-length} that $\tilde{c}$ and $\tilde{d}$ are geodesics that each minimizes the length between any two of its vertices. Let $i$ be  the smallest  index along $\tilde{c}$ such that $\tilde{c}(i)\in \tilde{d}$ and let $\ell$ be maximal such that $\tilde{c}(i+\ell)\in \tilde{d}$. By the preceding discussion $\tilde{c}(i)$ and $\tilde{c}(i+\ell)$ also delimitate a maximal subpath $\tilde{d}[j,j+\varepsilon\ell]$ of $\tilde{d}$ whose endpoints belong to $\tilde{c}$, with $\varepsilon\in\{-1,1\}$. We are thus in the situation depicted on Figure~\ref{fig:strategy}, right, where $\tilde{c}[-\infty,i-1]$ and $\tilde{c}[i+\ell+1,+\infty]$ lie on either side of $\tilde{d}$ in $\D$. In this case we say that the pair of index paths $(\Ipath{i}{\ell}_c, \Ipath{j}{\varepsilon\ell}_d)$ is \define{crossing}. This happens precisely when $(\Ipath{i}{\ell}_c, \Ipath{j}{\varepsilon\ell}_d)$ is a maximal pair of index paths with homotopic image paths and when the circular ordering of the arcs $\tilde{c}[i,i+1], \tilde{c}[i,i-1]$ and $\tilde{d}[j,j-\varepsilon]$ is the same as the circular ordering of $\tilde{c}[i+\ell,i+\ell-1], \tilde{c}[i+\ell,i+\ell+1]$ and $\tilde{d}[j,j+\varepsilon(\ell+1)]$. This last condition can be checked by  comparing the direct circular orderings of the projection of those arcs on $\Sigma$, taking into account the parity of the number of half-twisted arcs in $c\Ipath{i}{\ell}$ (or equivalently in $d\Ipath{j}{\varepsilon\ell}$). We emphasize that the parameters $i,j,\ell,\varepsilon$ are attached to $c$ and $d$, not to their lifts. Those parameters determine a crossing pair of index paths
for $c$ and $d$ if there exist two lifts $\tilde{c}$ of $c$ and $\tilde{d}$ of $d$ such as described above. According to Lemma~\ref{lem:strategy} in the strategy section, we have
\begin{lem}\label{lem:counting}
  Let $c,d$ be primitive combinatorial closed geodesics, none of which is exceptional. We assume that $c$ is not homotopic to $d$ or its inverse, unless $c$ is precisely equal to $d$. Then, $i(c,d)$ is equal to the number of crossing pairs of index paths. 
\end{lem}
The case where $c$ or $d$ is exceptional can be dealt with separately. See the proof of  Lemma~\ref{cor:primitive-main-result} and Section~\ref{sec:non-primitive} for further details on the exceptional case.

\subsection{Thick double-paths}\label{subsec:thick-double-paths}
Following the description of Theorem~\ref{th:geodesic}, the two sides $c\Ipath{\Imath}{\ell}$ and $d\Ipath{\Jmath}{\varepsilon\ell}$ of a crossing pair of index paths $(\Ipath{i}{\ell}_c, \Ipath{j}{\varepsilon\ell}_d)$ label a disk diagram composed of a sequence of paths and staircases. In particular, the vertices of the two sides can be put in 1-1 correspondence and corresponding vertices satisfy a local condition: they stay at distance zero or two and bound a same quad in the disk diagram. Hence, when looking for crossing pairs of index paths of $(c,d)$, we can start from any double point $(i,j)$ and walk in parallel along the two sides from $i$ and $j$, checking the above local condition at each step. However, we do not know a priori if the two sides will meet again to form a pair of homotopic paths and the search may be unsuccessful. In order to analyze the complexity of the search we consider pairs of paths that could potentially be part of a crossing pair of index paths, but are not necessarily so. Formally, a pair $(\Ipath{\Imath}{\ell}_c,\Ipath{\Jmath}{\varepsilon\ell}_d)$ of distinct index paths such that $c\Ipath{\Imath}{\ell}$ and $d\Ipath{\Jmath}{\varepsilon\ell}$ are corresponding subpaths of homotopic geodesic paths is called a \define{thick double path} of length $\ell$. In other words, a thick double path is such that its image paths can be extended to form homotopic geodesic paths. A thick double path gives rise to a sequence of $\ell+1$ index pairs $\Dpoint{\Imath+k}{\Jmath+\epsilon k}$ for $k\in [0,\ell]$. Let $\Delta$ be a disk diagram whose boundary is labelled by homotopic geodesic paths extending the thick double path as illustrated on Figure~\ref{fig:thick}. 
\begin{figure}
  \centering
  \includesvg[.6\textwidth]{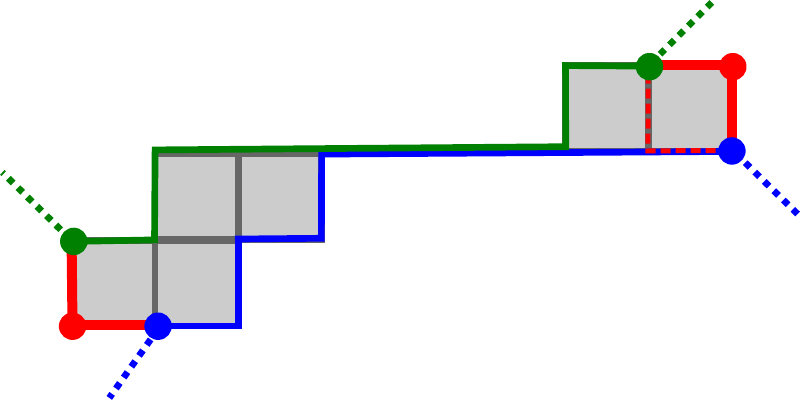}
  \caption{A thick double path extends to a pair of homotopic geodesics and bound a partial diagram. When $(i,j)$ and $(i+\ell,j+\varepsilon\ell)$ are in the same configuration, the corresponding  pairs of vertices in $\Delta$ are both joined by a path of length zero or two labelled by the same path $p$.}
  \label{fig:thick}
\end{figure}
The restriction of $\Delta$ to the part delimited by the thick double path is called a \define{partial diagram}. In a partial diagram each index pair $\Dpoint{\Imath}{\Jmath}$ may be in at most one of five \define{configurations}: either the corresponding vertices coincide in the partial diagram, or they appear diagonally opposite in a quad; in turn this quad may be labelled by one of the (at most) four quads incident to either $c[i,i+1]$ or $d[j,j+\varepsilon]$. A partial diagram  is \define{maximal} if it cannot be extended to form a partial diagram of a longer thick double path. 
\begin{lem}\label{lem:configurations}
  Let $c$ and $d$ be primitive geodesic curves, none of which is exceptional. Suppose in addition that either $c=d$, or $c$ is not homotopic to $d$ or its inverse. Then each configuration of an index pair may occur at most once in a partial diagram.
\end{lem}
\begin{proof}
  Suppose for a contradiction that the index pair $(i,j)$ occurs twice in the same configuration in a partial diagram of a thick double path of $(c,d)$. The two  occurrences of $(i,j)$  are separated by a path whose length $\ell$ is an integer multiple of $|c|$, say $\ell=k|c|$, and an integer multiple of $|d|$, say $\epsilon\ell=t|d|$.
In the partial diagram, $c(i)$ and $d(j)$ label vertices that are either identical or opposite in a quad. They are thus connected by a path $p$ of length zero or two. See Figure~\ref{fig:thick}. We infer that 
\begin{eqnarray}
  \label{eq:thick-double-path}
  c\Ipath{\Imath}{k|c|}\,\bhom\, p\cdot d\Ipath{\Jmath}{t|d|}\cdot p\inv
\end{eqnarray}
It ensues that $c^k\sim d^t$, considering free homotopy. Because $c$ and $d$ are primitive curves on a surface with negative Euler characteristic, this implies that $c$ is homotopic to $d$ or its inverse: every homotopy class has indeed a unique primitive root up to orientation on such a surface~\cite[Lem. 9.2.6]{r-ajccs-62},~\cite[p. 213]{b-gscrs-92}. By the hypotheses in the lemma, we thus have $c=d$. We first assume that the two index paths of the thick double path are both forward (i.e., $\varepsilon=1$). Exchanging the roles of $i$ and $j$ if necessary, we may further assume that $0< j-i \leq |c|/2$.
Equation~\eqref{eq:thick-double-path} now writes
\[c\Ipath{\Imath}{k|c|}\,\bhom\, p\cdot c\Ipath{\Jmath}{k|c|}\cdot p\inv \,\bhom\,  p\cdot c[i,j]\inv\cdot c\Ipath{\Imath}{k|c|}\cdot c[i,j]\cdot p\inv
\]
Equivalently, $c'^k\cdot q\,\bhom\,  q\cdot c'^k$ where $c':=c\Ipath{\Imath}{|c|}$ is a cyclic permutation of $c$ and $q := p\cdot c[i,j]\inv$ is a (closed) path of length at most $j-i+2$. Since they are  commuting, $c'^k$ and $q$ must admit a common primitive root~\cite[Lem. 9.2.6]{r-ajccs-62},~\cite[p. 213]{b-gscrs-92}. Whence $q\,\bhom\,c'^t$ for some $t\in \Z$, since $c'$ is itself primitive. By Remark~\ref{rem:geodesic-power}, $c'^t$ is geodesic, hence has minimal length in its homotopy class. It follows that $j-i+2 \geq |t|.|c|$. This is only possible if $t=0$ or if $|t|=1$ and $j-i=|c|-2$ (recall that the graph of the system of quads is bipartite so that indices in a pair have the same parity). 
\begin{itemize}
\item If $t=0$, then $q\sim 1$, so that $p\bhom c[i,j]$. Call $\Delta$ the partial diagram for $c\Ipath{\Imath}{k|c|}$ and $c\Ipath{\Jmath}{k|c|}$. The uniqueness of disk diagrams discussed after Corollary~\ref{cor:equal-length} implies that the initial point of $c\Ipath{\Jmath}{k|c|}$ and the point of index $j-i$ along $c\Ipath{\Imath}{k|c|}$ (which is also $c[j]$) are mapped to the same point in $\Delta$. However, coincident points in a (partial) diagram must have the same index along each side. This would imply $i=j$, which is impossible since we are only considering pairs of distinct index paths.
\item If $|t|=1$ and $j-i=|c|-2$, then $|c|=4$ (recall that $0<j-i\leq |c|/2$) and thus $j=i+2$. In the case $t=1$, we have $c'\bhom p\cdot c[i,j]\inv \bhom p\cdot c\Ipath{\Jmath}{2}\cdot {c'}\inv$, whence $(c')^2\sim p\cdot c\Ipath{\Jmath}{2}$. This is however impossible as the left member of the homotopy has minimal length 8 in its homotopy class while the right member has length 4. In the case $t=-1$, we have $c'\bhom c[i,j]\cdot p\inv$, whence $p\bhom c\Ipath{\Jmath}{2}\inv$. Similarly to the case $t=0$, we infer that the initial point of $c\Ipath{\Imath}{k|c|}$ and the point of index $2$ along $c\Ipath{\Jmath}{k|c|}$ are mapped to the same point in the partial diagram, which is also impossible.
\end{itemize}

We now assume that the thick double path is composed of a forward and a backward index paths. Similarly to the previous case, we infer that $c'\,\bhom\, q \cdot (c')\inv \cdot q\inv$, where $c':=c\Ipath{\Imath}{|c|}$  and $q:= p\cdot c[i,j]\inv$. Equivalently, $c$ and its inverse are freely homotopic. This is however impossible unless $c$ is contractible.
\end{proof}
\begin{cor}\label{cor:index-pairs}
  Each configuration of an index pair may occur at most once in the set of maximal partial diagrams.
\end{cor}
\begin{proof}
  By the preceding Lemma, we only need to check that a configuration of an index pair cannot occur twice in distinct partial diagrams of maximal thick double paths. This is essentially the unique lifting property of coverings.
\end{proof}

\subsection{Proof of Theorem~\ref{th:main-result}: the primitive case}
\begin{lem}\label{cor:primitive-main-result}
  Let $c,d$ be primitive curves of length at most $\ell$ on a system of quads. The crossing number $i(c,d)$ can be computed in $O(\ell^2)$ time.
\end{lem}
\begin{proof}
  We may assume that $c$ and $d$ are geodesic by Theorem~\ref{th:compute-geodesic}. We first consider the case where none of $c$ or $d$ is exceptional. According to Lemma~\ref{lem:counting}, we can compute $i(c,d)$ by enumerating the crossing pairs of index paths of $(c,d)$. Following the previous Section~\ref{subsec:thick-double-paths} we only need to list the maximal partial diagrams and to count those that correspond to crossing pairs of index paths. As described above Lemma~\ref{lem:counting}, this last step can be performed in time proportional to the total length of the maximal partial diagrams, which is $O(\ell^2)$ by Corollary~\ref{cor:index-pairs}. Since we are only interested in partial diagrams whose sides have coincident vertices, namely the endpoints of the crossing pairs of index paths, we can list the set of maximal partial diagrams as follows. Pick any index pair $(i,j)$ such that $c(i)=d(j)$ and form a trivial partial diagram with these vertices in coincident configuration. Extend the partial diagram maximally in both directions. By the unique lifting property of coverings, there is only one way to do so. Mark all the configurations of index pairs found in this diagram and  pick a new index pair which is not marked in the coincident configuration. Repeat until all the coincident configurations of index pairs  are marked. The total time spent by the procedure  is proportional to the total length of the maximal partial diagrams, hence is $O(\ell^2)$.

We next consider the case where $d$ is not exceptional but $c$ is an exceptional 1-sided curve.  Its square $c^2$ is a 2-sided curve with a unique homotopic geodesic whose turn sequence is $2^*$ or $\bar{2}^*$. In particular, $c^2$ cannot form a staircase with any curve since both sides of a staircase must have a vertex whose corresponding turn is $\pm 1$. (Note that it forms an infinite strip of width one with itself.) It follows that the crossing pairs of index paths of $(c^2,d)$ are \define{flat}, i.e., the index paths have identical image paths. This makes the search for crossing pairs of index paths even easier than in the general case so that $i(c^2,d)$ can be computed in $O(\ell^2)$ time. Applying the formulas in Proposition~\ref {prop:non-primitive-formulas} below (with $p=2$ and $q=1$) we finally derive $i(c,d)=i(c^2,d)/2$. Similar arguments apply when both $c$ and $d$ are exceptional to compute $i(c,d)=i(c^2,d^2)/4$, including the case where $c$ and $d$ are homotopic.
\end{proof}

\subsection{Non-primitive curves}\label{sec:non-primitive}
In order to finish the proof of Theorem~\ref{th:main-result} we need to handle the case of non-primitive curves. Thanks to canonical forms, computing the primitive root of a curve becomes extremely simple (see also Lustig~\cite[Rem. 5.2]{l-pggin2-87}) on orientable surfaces. 
\begin{lem}\label{lem:non-primitive}
  Let $c$ be a combinatorial curve of length $\ell>0$ in canonical form on an oriented system of quads. A primitive curve $d$ such that $c$ is homotopic to $d^k$ for some integer $k$ can be computed in $O(\ell)$ time.
\end{lem}
\begin{proof}
  By Theorem~\ref{th:canonical}, we may assume that $c$ is in canonical form.  Let $d$ be a primitive curve such that $c\sim d^k$ for some integer $k$. We look for $d$ in canonical form. By Remark~\ref{rem:geodesic-power}, the curve $d^k$ is also in canonical form. The uniqueness of the canonical form implies that $c=d^k$, possibly after some circular shift of $d$. It follows that $d$ is the smallest prefix of $c$ such that $c$ is a power of this prefix. It can be found in $O(\ell)$ time using a variation of the Knuth-Morris-Pratt algorithm to find the smallest period of a word~\cite{kmp-fpms-77}.  
\end{proof}
We now consider the case of a curve $c$ on a non-orientable system of quads. If $c$ is 1-sided we can reduce the computation of a primitive root to the 2-sided case by noting that $c$ and $c^2$ have the same primitive roots. 
If $c$ is 2-sided it can be given one of two canonical forms by fixing one of the two consistent orientations along the curve before removing spurs, brackets and $-1$'s in its turn sequence. Alternatively, we can build and orient the double cover of the system of quads and note that the 2-sided curves are precisely those that lift to closed curves in this cover. The projection of the canonical form of a lift provides one of two canonical forms corresponding to the two possible lifts. By Lemma~\ref{lem:non-primitive} we can compute in linear time a primitive root of the lift in canonical form. Its projection $d$ is thus a 2-sided geodesic. If $d$ is itself primitive, then $d$ is a primitive root of $c$. Otherwise, $d\sim e^2$ for some 1-sided primitive geodesic $e$, which is thus a primitive root of $c$.
\begin{itemize}
\item 
If $e$ is exceptional then $d$ has a constant turn sequence $2^*$ or $\bar{2}^*$ and bounds an immersed annulus wrapping twice around a M\"obius strip. Consider  any non-boundary arc $a$ of this annulus. Let $d(i)$ be the source point of $a$, and let $\ell=|d|/2$. Then $d\sim (d\Ipath{i+\ell}{\ell}\cdot a)^2$ so that $e$ must be homotopic to $d\Ipath{i}{\ell}\cdot a$. See Figure~\ref{fig:exceptional-case}.
\begin{figure}[h]
  \centering
  \includesvg[.4\linewidth]{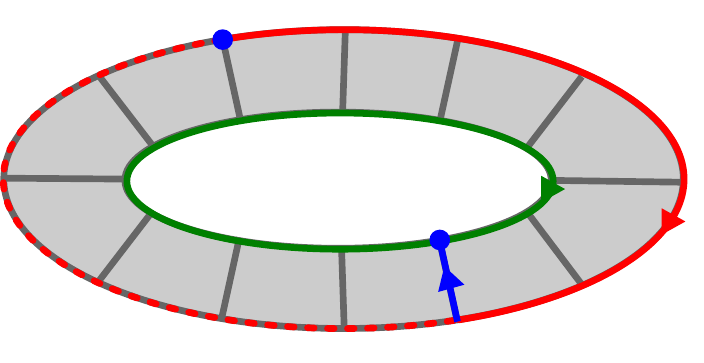}
  \caption{An annular diagram for $d$ with itself.}
  \label{fig:exceptional-case}
\end{figure}
\item If $e$ is not exceptional, then $e^2$ is geodesic so that $|e|=\ell$ with $\ell=|d|/2$. Consider an annular diagram $\Delta$ for $d$ and $e^2$. As for a thick partial diagram, the boundary loops of $\Delta$ can be put in 1-1 correspondence so that corresponding vertices appear in one of five possible configurations (in fact, only three are possible since $d$ is canonical). Let $p_0$ and $p_\ell$ be paths of length at most two connecting $d(0)$ and $d(\ell)$ to respectively $e^2(0)$ and $e^2(\ell)$. The possible configurations provide at most five candidates for $p_0$ and similarly for $p_\ell$. For each pair of candidates, say $(q_0,q_\ell)$, we consider the curve $e':=q\inv_0\cdot d\Ipath{0}{\ell}\cdot q_\ell$. We can confirm that the candidate paths are the right ones by checking whether $e'^2$ is indeed homotopic to $d$. This takes $O(\ell)$ time by Theorem~\ref{th:canonical}.
\end{itemize}

To summarize the discussion,
\begin{prop}\label{prop:primitive-root}
  Given a curve $c$ of length $\ell$ on a system of quads,  orientable or not, a primitive root of $c$ can be computed in $O(\ell)$ time.
\end{prop}

\subsection{End of proof of Theorem~\ref{th:main-result}}\label{subsec:end-of-proof}
We now have all the ingredients to complete the proof of Theorem~\ref{th:main-result}. We essentially need to extend Lemma~ \ref{cor:primitive-main-result} to non-primitive curves and to handle the case of curves on surfaces of non-negative Euler characteristic. 
The geometric intersection number of non-primitive curves is related to the geometric intersection number of their primitive roots by simple formulas. The next result is rather intuitive, see Figures~\ref{fig:non-primitive} and~\ref{fig:non-primitive-non-orient}, and is part of the folklore although we could only find references in some relatively recent papers. 
\begin{prop}[\cite{gs-mcmcr-97,gkz-amnip-05}]\label{prop:non-primitive-formulas}
  Let $c$ and $d$ be primitive curves such that $c$ is not homotopic to $d$ or its inverse and let $p,q$ be positive integers. Then,
\[
i(c^p,d^q) = pq \times i(c,d)
\]
and
\[i(c^p,c^q) = \left\{
              \begin{array}{ll}
                   2pq \times i(c) +\min\{p,q\} & \text{if } c\text{ is 1-sided and $p$ and $q$ are odd}, \\
                   2pq \times i(c) & \text{otherwise.}
              \end{array}
       \right. 
\]
Moreover, the geometric self-intersection number of a curve is related to the geometric self-intersection number of its root thanks to the formulas
\[i(c^p) =  \left\{
              \begin{array}{ll}
                p^2 \times i(c) + p-1 & \text{if $c$ is 2-sided,}\\
                p^2 \times i(c) + \lfloor \frac{p-1}{2}\rfloor & \text{otherwise.}
              \end{array}\right.
\]
\end{prop}
\begin{figure}
  \centering
  \includesvg[\linewidth]{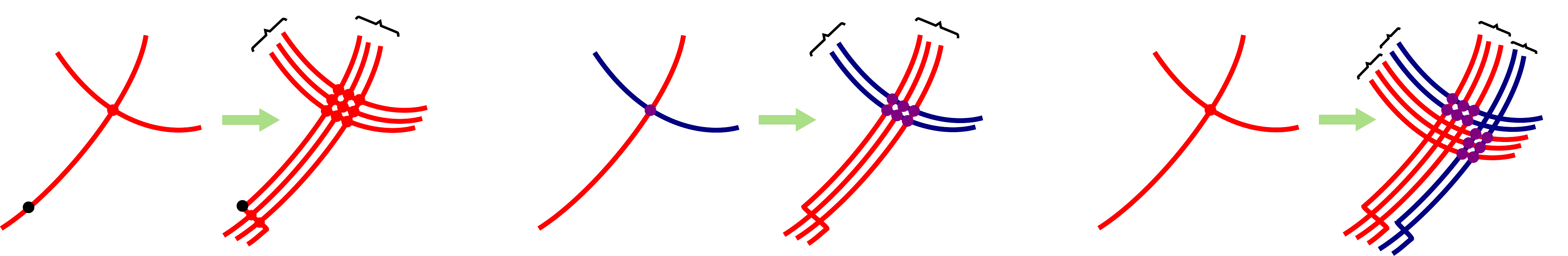}
  \caption{Left, each self-intersection of a 2-sided curve $c$ gives rise to $p^2$ self-intersections of its $p$th power obtained by wrapping $p$ times around $c$ in a small tubular neighborhood. $p-1$ additional self-intersections are needed to connect the end of the wrapping curve with the basepoint. Middle, each intersection of $c$ and $d$ gives rise to $pq$ intersections of (homotopic perturbations of) $c^p$ and $d^q$. Right, when $c$ and $d$ are homotopic and $c^p$ or $c^q$ is 2-sided, one should count $2pq$ intersections per self-intersection of $c$.}
  \label{fig:non-primitive}
\end{figure}
\begin{figure}
  \centering
  \includesvg[\linewidth]{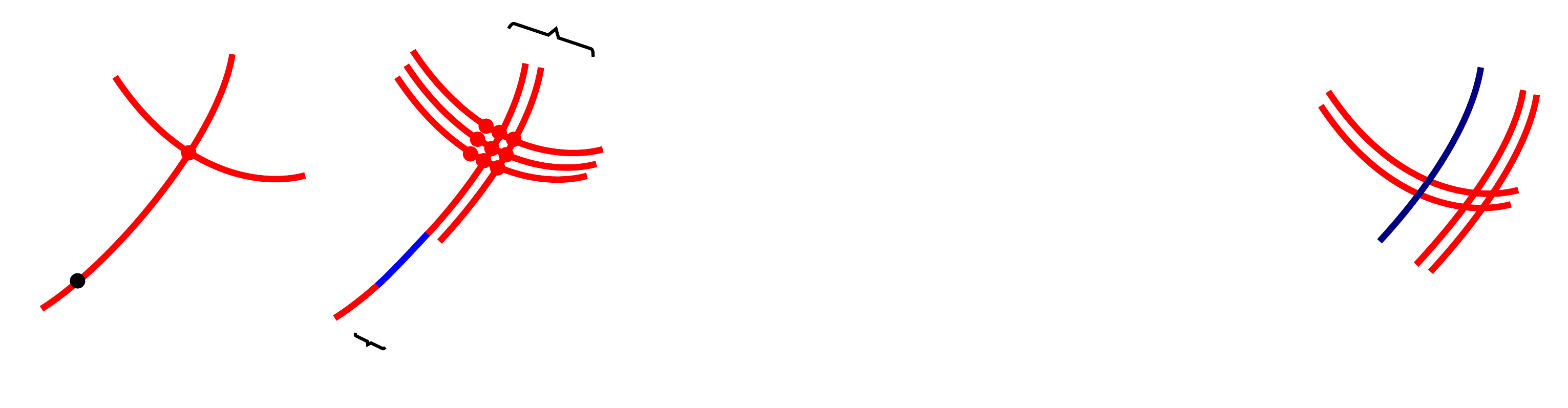}
  \caption{Left, the $p$th power of a 1-sided curve $c$ is obtained from $p$ copies of $c$ running parallel to $c$ and arriving at the starting basepoint of $c$ in reverse order. By ordering the copies conveniently we get $\lfloor \frac{p-1}{2}\rfloor$ self-intersections to reconnect the copies instead of $p-1$ for a 2-sided curve. Right, a minimal configuration for $c^p$ and $c^q$ with $p=3$ and $q=5$.}
  \label{fig:non-primitive-non-orient}
\end{figure}
\begin{proof}[Proof of Theorem~\ref{th:main-result}]
Let $c,d$ and $\Sigma$ be the two combinatorial curves and the combinatorial surface as in the Theorem. By Lemma~\ref{lem:systems-of-quads} we can assume after $O(n)$ time preprocessing that $\Sigma$ is a system of quads.
We first consider the case where $\Sigma$ has negative Euler characteristic. 
Thanks to Theorem~\ref{th:compute-geodesic} and Proposition~\ref{prop:primitive-root} we can determine geodesic primitive curves $c'$ and $d'$ and integers $p,q$ such that $c\sim c'^p$ and $d\sim d'^q$ in $O(\ell)$ time. We then compute $i(c',d')$ and $i(c', c')$ in $O(\ell^2)$ time according to Corollary~\ref{cor:primitive-main-result}. We finally use the formulas in the previous proposition to deduce $i(c,d)$ and $i(c)$ from $i(c',d')$ and $i(c',c')$. 

We next consider the case where $\Sigma$ has non-negative Euler characteristic. 
\begin{itemize}
\item If $\Sigma$ is a sphere or a disk, then every curve is contractible and $i(c,d)=i(c)=0$. 
\item If $\Sigma$ is a cylinder, then every two curves can be made non crossing so that $i(c,d)=0$ while $i(c)=p-1$. 
\item If $\Sigma$ is a projective plane, there is only one non-trivial homotopy class, say $c'$, for which $i(c',c') = 1$ and $i(c')=0$.
\item If $\Sigma$ is a torus, there is no need to build a system of quads; we may just compute the reduced surface which is composed of two loop edges. Those loop edges, say $\alpha$ and $\beta$, represent generators of the fundamental group of the torus and we can write $c\sim \alpha^x\cdot\beta^y$ and $d\sim \alpha^{x'}\cdot\beta^{y'}$. It results from classical formulas that: $i(c)=\gcd(x,y)-1$ and $i(c,d)=|\det\big((x,y),(x',y')\big)|$.
\item Finally, if $\Sigma$ is a Klein bottle, similarly to the torus case the reduced surface is composed of two loops $\alpha,\beta$ satisfying either the relation $\alpha^2\beta^2\sim 1$ or $\alpha\beta\alpha\beta\inv\sim 1$. We can transform the embedded graph in constant time so that the second relation holds. This relation can be written as a pseudo commutation relation $\alpha\beta=\beta\alpha\inv$ that preserves the cumulative sum of the exponents of $\beta$ and the parity of the cumulative sum of the exponents of $\alpha$. From this simple remark we easily derive that  every conjugacy class in the fundamental group has a unique normal form $\alpha^m\beta^{2n}$, $\beta^{2n+1}$ or  $\alpha\beta^{2n+1}$, where $m,n$ can take any integer value. Table~\ref{tab:KleinBottle} provides all the geometric intersection numbers for the primitive classes, which must have one of the forms $\beta$, $\alpha\beta$, or $\alpha^n\beta^{2m}$ with $m$ and $n$ coprime.
\begin{table}[h]
  \centering
\[
\begin{array}{c|c|c|c|}
\cline{2-4}
  & \alpha^k\beta^{2\ell} & \beta & \alpha\beta\\
\hline
\multicolumn{1}{ |c| }{\alpha^m\beta^{2n}} &  \parbox[c][2em][c]{0.2\textwidth}{ ${\scriptsize
\begin{Vmatrix}
  k & m  \\\ell & n
\end{Vmatrix}
}
+
{\scriptsize
\begin{Vmatrix}
  k & -m\\  \ell & n
\end{Vmatrix}
}
$
}
& |m| & |m|\\
\hline
\multicolumn{1}{ |c| }{\beta} & |k| & 1 & 0\\
\hline
\multicolumn{1}{ |c| }{\alpha\beta} & |k| & 0 & 1\\
\hline
\end{array}
\qquad
\begin{array}{|c|c|c|c|}
\hline
c & \alpha^k\beta^{2\ell} & \beta & \alpha\beta\\
\hline
i(c) & |k\ell| & 0 & 0\\
\hline
\end{array}
\]
  \caption[intersection numbers on the Klein bottle]{Left table, the geometric  intersection numbers for the pairs of distinct primitive homotopy classes. Here, $k$ and $\ell$ are coprime, as well as $m$ and $n$, and $(k,\ell)\neq \pm (m,n)$. We use the notation
${\scriptsize
\begin{Vmatrix} x & y\\  z & t \end{Vmatrix}
}:=|xt-yz|$. Right table, the geometric self-intersection numbers of primitive classes. Again, 
$k$ and $\ell$ are assumed to be coprime.
}
  \label{tab:KleinBottle}
\end{table}
We can combine this table with the formulas of Propositions~\ref{prop:non-primitive-formulas} to compute the geometric intersection number of any curves.
\end{itemize}
\end{proof}

\section{Combinatorial perturbations}\label{sec:perturbation}
A combinatorial curve on a combinatorial surface may be seen as a continuous curve in general position, i.e. with a finite number of transverse pairwise crossings, snapped to the graph $G$ of the combinatorial surface. When doing so several parts of the continuous curve may be mapped to the same edge of $G$. In order not to lose information, one needs to record their ordering transverse to that edge. 
Following the notion of a combinatorial set of loops as in~\cite{cl-oslos-05}, we define a \define{combinatorial \immersion{}} of a set $\mathcal{C}$ of combinatorial curves 
as the data for each arc $a$ in $G$  of a left-to-right order $\preceq_a$ over all the occurrences of $a$ or $a\inv$ in the curves of $\mathcal{C}$. The only requirement is that opposite arcs should be associated with inverse orders\footnote{This definition assumes that the combinatorial surface is orientable and consistently oriented. In the general case,  opposite arcs should be associated with inverse orders if the corresponding edge is untwisted and with the same order if the edge is half-twisted.}.
Let $A_v$ be the set of occurrences of $a$ or $a\inv$ in the curves of $\mathcal{C}$ where $a$ runs over all arcs of $G$ with origin $v$. A combinatorial \immersion{} induces for each vertex $v$ of $G$ a circular order $\preceq_{v}$ over $A_v$; if $a_1,\dots,a_k$ is the direct cyclic order of the arcs of $G$ with origin $v$, then $\preceq_{v}$ is the cyclic concatenation of the orders $\preceq_{a_1},\dots,\preceq_{a_k}$. See Figure~\ref{fig:combinatorial-immersion}.
\begin{figure}[h]
  \centering
  \includesvg[.3\linewidth]{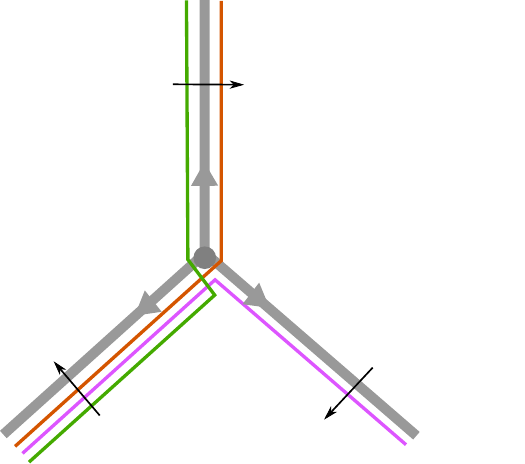}
  \caption{The circular order $\preceq_{v}$ is the concatenation of the orders $\preceq_{a_1},\preceq_{a_2},\preceq_{a_3}$.}
  \label{fig:combinatorial-immersion}
\end{figure}
\Animmersion{} of $\mathcal{C}$ restricts to \animmersion{} of each of its curves. Conversely, given \animmersion{} of each curve in $\mathcal{C}$ we can built \animmersion{} (in general many) of their union by merging for each arc $a$ the left-to-right orders $\preceq_a$ of each curve \immersion{}.

\paragraph{Combinatorial crossings.}
Recall that a double point of a pair $(c,d)$ of combinatorial closed curves is a pair of indices $\Dpoint{\Imath}{\Jmath}\in \Z/|c|\Z\times \Z/|d|\Z$ such that $c(\Ibar)=d(\Jbar)$, assuming $\Ibar\neq \Jbar$ when $c=d$.
Let $\Pi$ be \animmersion{} of $c$ and $d$. A double point $\Dpoint{\Imath}{\Jmath}$ of $(c,d)$ is a \define{crossing}\footnote{The crossings of \animmersion{} should not be confused with the crossing double points of a combinatorial curve appearing as crossing double paths of zero length in Section~\ref{subsec:double-paths}. Which notion of crossings is used should always be clear from the context.} in $\Pi$ if the pairs of arc occurrences $([\Ibar-1,\Ibar]_c,[\Ibar,\Ibar+1]_c)$ and $([\Jbar-1,\Jbar]_d,[\Jbar,\Jbar+1]_d)$ are linked in the  $\preceq_{c(\Ibar)}$-order, i.e. if they appear in the cyclic order 
\[\cdots [\Ibar,\Ibar-1]_c\cdots [\Jbar,\Jbar-1]_d\cdots [\Ibar,\Ibar+1]_c\cdots [\Jbar,\Jbar+1]_d\cdots, 
\]
with respect to $\preceq_{c(\Ibar)}$ or the opposite order. An analogous definition holds for a self-crossing of a single curve, taking $c=d$ in the above definition. Note that the notion of (self-)crossing is independent of the traversal directions of $c$ and $d$. The number of crossings of $c$ and $d$ and of self-crossings of $c$ in $\Pi$ is denoted respectively by $i_{\Pi}(c,d)$ and $i_{\Pi}(c)$.

We define the \define{combinatorial self-crossing number} of $c$, denoted by $i(c)$, as the minimum of $i_{\Pi}(c')$  over all the combinatorial \immersion{}s $\Pi$ of any combinatorial curve $c'$ freely homotopic to $c$. The \define{combinatorial crossing number} of two combinatorial curves $c$ and $d$ is defined the same way taking into account crossings between $c$ and $d$ only.
\begin{lem}\label{lem:geometric-vs-combinatorial}
  The combinatorial (\-self\--\-)\-crossing number coincides with the geometric (\-self\--\-)\-intersection number, i.e. for every combinatorial closed curves $c$ and $d$ on a combinatorial surface with graph $G$:
\[i(c)=i(\rho(c))\qquad \text{ and } \qquad i(c,d)=i(\rho(c),\rho(d)).
\]
where $\rho$ is any cellular embedding of $G$.
\end{lem}
\begin{proof}
 Every combinatorial \immersion{} $\Pi$ of $c$ can be realized by a continuous curve $\gamma\sim \rho(c)$ with the same number of self-intersections as the number of self-crossings of $\Pi$. To construct $\gamma$ we consider small disjoint disks centered at the images of the vertices of $G$ in $S$. We connect those vertex disks by disjoint strips corresponding to the edges of $G$. For every arc $a$ of $G$ we  draw inside the corresponding edge strip parallel curve pieces labelled by the occurrences of $a$ or $a\inv$ in $c$ in the left-to-right order $\preceq_a$. The endpoints of those curve pieces  appear on the boundary of the vertex disks in the circular vertex orders induced by  $\Pi$. It remains to connect those endpoints by straight line segments inside each disk (via a parametrization over a Euclidean disk) according to the labels of their incident curve pieces. The resulting curve $\gamma$ can be homotoped to $\rho(c)$ and its self-intersections may only appear inside the vertex disks. Clearly, an intersection of two segments of $\gamma$ in a disk corresponds to linked pairs of arc occurrences, i.e. to a combinatorial crossing, and vice-versa. It follows that $i(c) \geq i(\rho(c))$. To prove the reverse inequality we show that for any continuous curve $\gamma$ in general position there exists a combinatorial \immersion{} of a curve $c'$ such that $\rho(c')\sim \gamma$ and the number of self-crossings of $c'$ is at most the number of self-intersections of $\gamma$. By an isotopy we can enforce $\gamma$ to stay in the neighborhood of $G$ composed of the edge strips and vertex disks, and to have all its self-intersections inside the vertex disks. After removing the vertex disks from $S$, we are left with a set of disjoint and simple pieces of $\gamma$ lying inside the edge strips. If a piece of curve in a strip has its endpoints on the same end of the strip we can further isotope the piece inside the incident vertex disk. This way all the curve pieces join the two ends of their strip and can be ordered from left-to-right inside each strip (assuming a preferred direction of the strip). Replacing each curve piece by an arc occurrence we thus define a combinatorial \immersion{} of a curve $c'$ with the required properties.
The same constructions apply to \animmersion{} of two curves showing that  $i(c,d)=i(\rho(c),\rho(d))$. 
\end{proof}

\section{Computing a minimal immersion}\label{sec:computing-immersions}
For the rest of the paper we assume that \textbf{all considered system of quads are orientable} (and consistently oriented) and we only deal with the self-intersection number of a single curve. We thus drop the subscript $c$ to denote an index path $\Ipath{\Imath}{\ell}$ or an arc occurrence $\Focc{\Imath}$. \emph{We also write intersections for self-intersections}.
By Theorem~\ref{th:main-result} we can compute the geometric intersection number of a curve efficiently. Here, we describe a way to compute a minimal immersion, that is an actual immersion with the minimal number of intersections. We refer to the framework of the previous Section~\ref{sec:perturbation} to describe such an immersion as a combinatorial perturbation. Thanks to Lemma~\ref{lem:geometric-vs-combinatorial}, we know that this framework faithfully encodes the topological configurations.
\subsection{Bigons and monogons}
A \define{bigon} of \animmersion{} $\Pi$ of  $c$ is a  pair of index paths $(\Ipath{\Imath}{\ell},\Ipath{\Jmath}{k})$ whose \define{sides} $c\Ipath{\Imath}{\ell}$ and $c\Ipath{\Jmath}{k}$ have strictly positive lengths, are homotopic, and whose \define{tips}
$(i,j)$ and $(i+\ell,j+k)$ are combinatorial crossings for $\Pi$. See Figure~\ref{fig:bigonAndMonogon}.
\begin{figure}[h]
  \centering
  \includesvg[.8\linewidth]{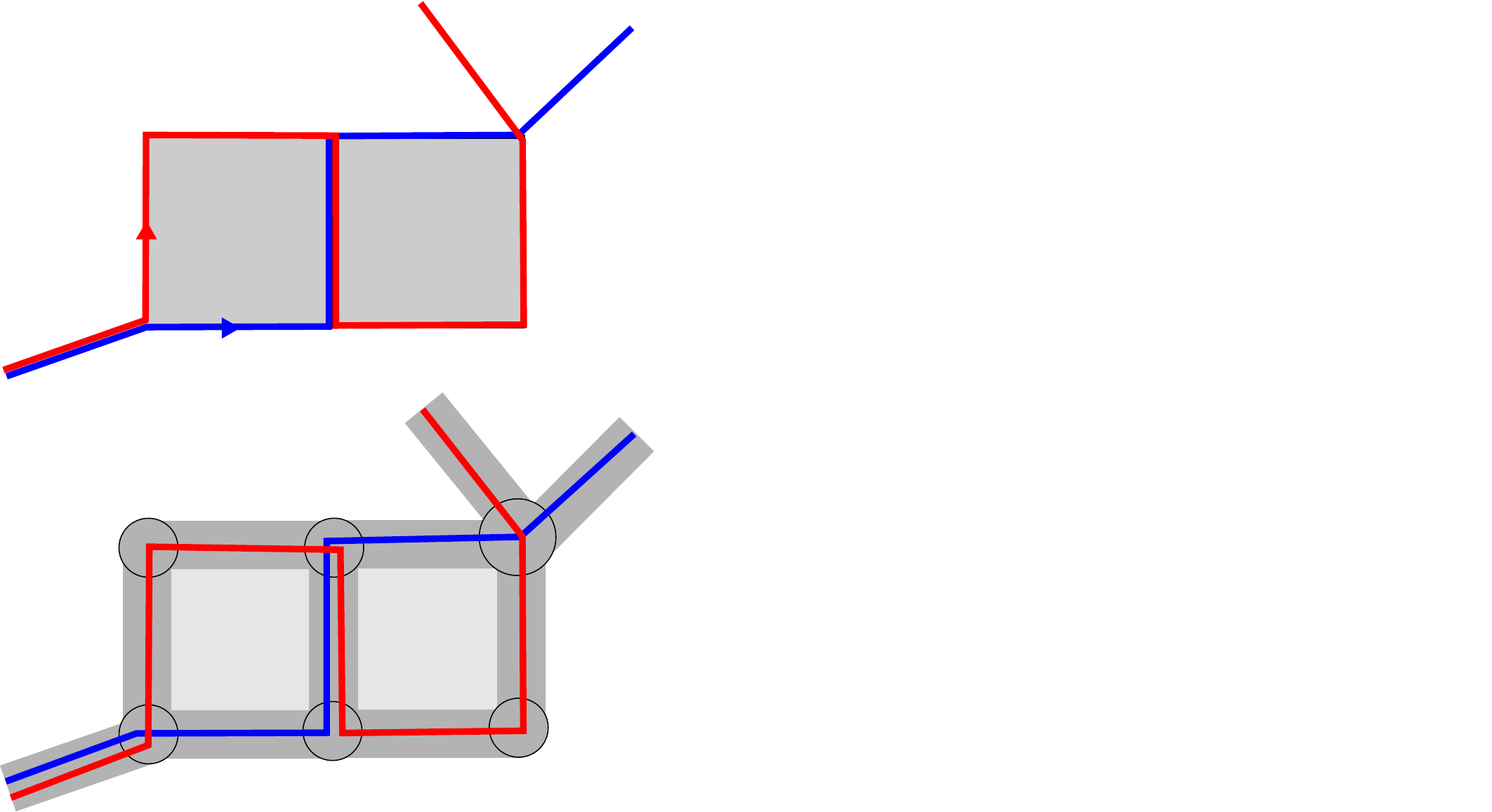}
  \caption{Left, a pair of paths in system of quads and a combinatorial perturbation with three bigons $(\Ipath{\Imath}{2},\Ipath{\Jmath}{2})$, $(\Ipath{\Imath+2}{1},\Ipath{\Jmath+2}{3})$ and $(\Ipath{\Imath}{3},\Ipath{\Jmath}{5})$. Right, a path and a combinatorial perturbation with a monogon of length 6.}
  \label{fig:bigonAndMonogon}
\end{figure}
A \define{monogon} of $\Pi$ is an index path $\Ipath{\Imath}{\ell}$ of strictly positive length such that $(\Ibar,\Overline{\Imath+\ell})$ is a combinatorial crossing and the image path $c\Ipath{\Imath}{\ell}$ is contractible.
\begin{prop}\label{prop:bigon}
  A combinatorial \immersion{} of a primitive curve has excess self-crossing if and only if it contains a bigon or a monogon.
\end{prop}
\begin{proof}
We just need to prove the existence of a bigon or monogon for a continuous realization $\gamma: \R/\Z\to S$ of the combinatorial \immersion{}. This bigon or monogon corresponds to a combinatorial bigon or monogon as claimed in the Proposition. To prove the topological counterpart, we first note as in Section~\ref{sec:strategy} that the minimal number of intersections is counted by the lifts of $\gamma$ whose limit points alternate along $\partial\mathbb{D}$. 
It follows that $\gamma$ is minimally crossing when its lifts are pairwise disjoint unless the corresponding limit points alternate in $\partial\mathbb{D}$ and they should intersect exactly once in that case. When $\gamma$ is not minimally crossing  there must exist two lifts intersecting more than necessary, thus forming a bigon in $\D$. The projection of such a bigon on $S$ gives the desired result. Note that a statement similar to the one in the Proposition holds for \animmersion{} of a pair of primitive curves $c$  and $d$ with excess crossing.  In Appendix~\ref{app:bigon}, we provide a purely combinatorial proof of the direct implication of the Proposition in the case of two curves. 
\end{proof}

Suppose that $\Pi$ is a combinatorial \immersion{} of a primitive geodesic $c$. It cannot have a monogon by Corollary~\ref{cor:no-monogon}. Hence, according to Proposition~\ref{prop:bigon} the \immersion{} $\Pi$ is minimally crossing unless it contains a bigon. One could eliminate such a bigon by swapping its sides. However, it may happen that the index paths of the bigon share some common part, which would make the swapping ambiguous.  In agreement with the terminology of Hass and Scott~\cite{hs-ics-85}, a bigon $(\Ipath{\Imath}{\ell},\Ipath{\Jmath}{k})$ is said \define{singular} if 
\begin{enumerate}
\item its two index paths have disjoint interiors, i.e. they do not share any arc occurrence;
\item when $\Jbar=\Overline{\Imath+\ell}$ the following arc occurrences
\[\Bocc{\Imath}, [j,j-1],\Focc{\Imath},[\Jmath+k,\Jmath+k+1],[j,j+1],[\Jmath+k,\Jmath+k-1]
\]
do not appear in this order or its opposite in the circular ordering induced by $\Pi$ at $c(j)$;
\item when $\Ibar=\Overline{\Jmath+k}$ the following arc occurrences
\[\Bocc{\Imath}, \Bocc{\Jmath},\Bocc{\Imath+\ell},\Focc{\Imath},\Focc{\Imath+\ell},\Focc{\Jmath}
\]
do not appear in this order or its opposite in the circular ordering  induced by $\Pi$ at $c(i)$.
\end{enumerate}
Figure~\ref{fig:singular-def} shows that when a bigon does not satisfy condition 2 or 3, swapping its sides does not reduce the number of crossings.
\begin{figure}
  \centering
  \includesvg[\linewidth]{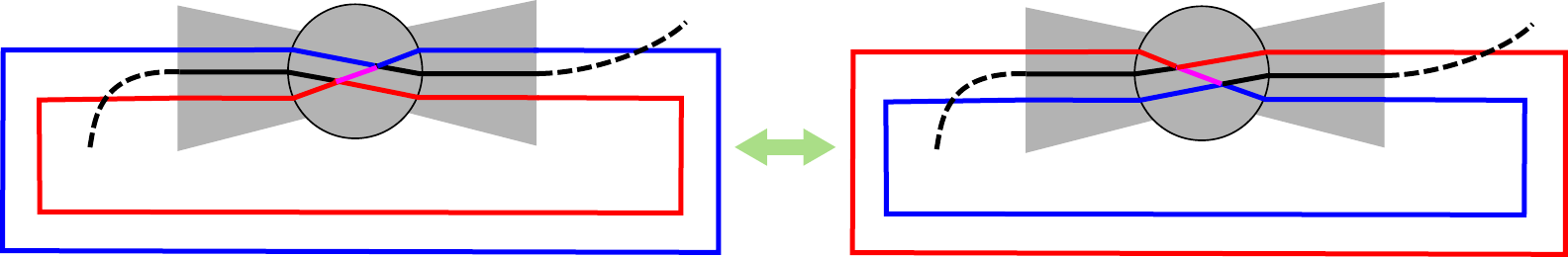}
  \caption{Right, the realization of the bigon $(\Ipath{\Imath}{\ell},\Ipath{\Jmath}{k})$ where $\Jbar=\Overline{\Imath+\ell}$ and $\Bocc{\Imath}$ is in-between $\Bocc{\Imath+\ell}$ and $\Bocc{\Jmath+k}$, and $\Focc{\Jmath+k}$ is  in-between $\Focc{\Imath}$ and $\Focc{\Jmath}$.  The small purple part is at the same time the beginning of the red side of the bigon and the end of the blue side. Swapping this bigon does not reduce the number of crossings.}
  \label{fig:singular-def}
\end{figure}
Note that when $c$ is geodesic and $k=-\ell$ there cannot be any identification between $\{\Ibar, \Overline{\Imath+\ell}\}$ and $\{\Jbar, \Overline{\Jmath-\ell}\}$. For instance, $\Ibar=\Jbar$ implies $c\Ipath{\Imath}{\ell}\sim c\Ipath{\Jmath}{-\ell}=c\Ipath{\Imath}{-\ell}$ so that $c\Ipath{\Imath-\ell}{2\ell}\sim 1$, while $\Jbar=\Overline{\Imath+\ell}$ implies $c\Ipath{\Imath}{\ell} \sim c\Ipath{\Jmath}{-\ell} = c\inv\Ipath{\Imath}{\ell}$ so that $c^2\Ipath{\Imath}{\ell}\sim 1$. In both cases $c$ would have a monogon which would contradict Corollary~\ref{cor:no-monogon}.

When $c$ is primitive and geodesic, we cannot have $\Jmath=\Imath+\ell$ and $\Imath=\Jmath+\ell$ at the same time. Otherwise, we must have $k=\ell$ by the preceding paragraph, and $c$ would be a circular shift of $c\Ipath{\Imath}{\ell}\cdot c\Ipath{\Jmath}{\ell}$ and thus homotopic to a square. We also remark that when one of the last two conditions in the definition is not satisfied, the bigon maps to a non-singular bigon in the continuous realization of $\Pi$ as in Lemma~\ref{lem:geometric-vs-combinatorial}. See Figure~\ref{fig:singular-def}.

When the bigon is singular we can swap its two sides by exchanging the two arc occurrences $[\Overline{\Imath+p},\Overline{\Imath+p+1}]$ and $[\Overline{\Jmath+\varepsilon p},\Overline{\Imath+\varepsilon(p+1)}]$ in the left-to-right order associated to their supporting arc, for $0\leq p< \ell$ and $k=\varepsilon\ell$. 
\begin{lem}\label{lem:swap}
   Swapping the two sides of a singular bigon of \animmersion{} of a geodesic primitive curve decreases its number of crossings by at least two. 
\end{lem}
This is relatively obvious if one considers a continuous realization of the \immersion{}, performs the swapping and comes back to a combinatorial \immersion{} as in the proof of Lemma~\ref{lem:geometric-vs-combinatorial}. We nonetheless provide a purely combinatorial proof in Appendix~\ref{app:swap}. 

Hence, by swapping singular bigons we may decrease the number of crossings until there is no more singular bigons. It follows from the next theorem that the resulting \immersion{} has no excess crossing.
\begin{thm}[Hass and Scott~{\cite[Th. 4.2]{hs-ics-85}}]\label{th:singular-bigon}
   \Animmersion{} of a primitive geodesic curve has excess crossing if and only if it contains a singular bigon.
\end{thm}
\begin{proof}
We realize $\Pi$ by a continuous curve $\gamma$ with the same construction as in Lemma~\ref{lem:geometric-vs-combinatorial}. By Theorem 4.2 in~\cite{hs-ics-85} the curve $\gamma$ has a singular bigon. In turn this bigon corresponds to a singular bigon in $\Pi$. 
\end{proof}
In the next section we describe how to detect bigons in practice.

\subsection{Proof of Theorem~\ref{th:compute-immersion}}\label{subsec:finding-bigons}
If \animmersion{} $\Pi$ of a primitive geodesic $c$ has a bigon then, by the same arguments as for a crossing pair of index paths in Section~\ref{subsec:thick-double-paths}, the two sides label a disk diagram and their vertices can be put in 1-1 correspondence such that corresponding vertices stay at distance zero or two and bound a same quad in the disk diagram.
Hence, we can look for bigons among the maximal partial diagrams of $c$. This leads to a quadratic algorithm for finding a bigon of $\Pi$, or deciding that $\Pi$ has no bigon.

\begin{proof}[Proof of Theorem~\ref{th:compute-immersion}]
  By Theorem~\ref{th:canonical} we can assume that $c$ is geodesic. We first consider the case where $c$ is primitive. Let $\Pi$ be a randomly chosen \immersion{} of $c$. We store the induced vertex orderings of the arc occurrences in arrays so that we can test if a double point is a crossing in constant time using pointer arithmetic. By Theorem~\ref{th:singular-bigon}, if $\Pi$ has excess crossing it has a singular bigon. This singular bigon defines a thick double path that must appear in some maximal partial diagram of $\Pi$ as two boundary paths with common extremities.
  Hence, in order to find a singular bigon we just need to scan in each maximal partial diagram the index pairs that are combinatorial crossings and test whether the two sides between any two consecutive such index pairs define a singular bigon.  According to the proof of Corollary~\ref{cor:index-pairs} in Section~\ref{sec:counting}, the cost for scanning all the maximal partial diagrams is $O(\ell^2)$. By the very definition of a singular bigon, testing whether a pair of index paths define a singular bigon can be performed in constant time. It follows that we can compute a singular bigon, if any, in $O(\ell^2)$ time. Once a singular bigon is found we swap its sides in $O(\ell)$ time. This preserves the geodesic character and, by Lemma~\ref{lem:swap}, the number of crossings is reduced by at least two.
Since $\Pi$ may have $O(\ell^2)$ excess crossings, we need to repeat the above procedure $O(\ell^2)$ times before the curve has no more singular bigon, hence is minimally crossing. This concludes the proof of the theorem in the case of primitive curves. When $c=d^p$, with $d$ primitive, we  first find a minimally crossing \immersion{} for $d$ by the above procedure. We further traverse the \immersion{} $p$ times duplicating $d$ as many times. As we start each traversal we connect the last arc occurrence of the previously traversed copy with the first occurrence of the next copy. We continue this copy by duplicating each traversed arc occurrence to its right. It is easily seen that the number of crossings of the final \immersion{} of $c$, after connecting the last traversed arc with the first one, satisfies the formula in Proposition~\ref{prop:non-primitive-formulas}. This concludes the proof for non-primitive curves.
\end{proof}
\section{Detecting simple curves}\label{sec:simple-curves}
As noticed in Section~\ref{sec:strategy}, geodesics on a hyperbolic surface possess a minimal number of self-intersections in their homotopy class but this is no longer true for combinatorial geodesics. In particular, it is easy to construct examples of combinatorial geodesics homotopic to simple curves that nonetheless exhibit crossings. This remains true even if we assume the geodesic to be in canonical form. In Section~\ref{subsec:unzip}, we present a simple and efficient algorithm to recognize curves homotopic to simple curves. Given a combinatorial curve, this algorithm constructs a combinatorial perturbation of some homotopic geodesic. The aim is to obtain, if possible, a perturbation without crossings. After computing the canonical form of the given curve, the construction starts with the perturbation of an empty path and adds the arcs of the canonical geodesic incrementally, each time extending the path perturbation. In order to ensure at each step that this perturbation has no crossings we may need to transform the canonical geodesic into another homotopic geodesic. Those transformations amounts to unzipping the potential bigons of the canonical geodesic. Lemma~\ref{lem:zipper-complexity} states that this unzip algorithm can be implemented to run in $O(\ell\log\ell)$ time, where $\ell$ is the length of the input curve. Although the returned geodesic need not be in canonical form, the algorithm strongly requires that the given curve be in canonical form.
The rest of Section~\ref{subsec:unzip} is devoted to the proof of Proposition~\ref{prop:zipper} asserting that the perturbation returned by the unzip algorithm has no combinatorial crossings if and only if the input curve is homotopic to a simple curve. The main argument is to show that if the unzip algorithm fails to produce a perturbation without crossings then one of its crossings gives rise to a pair of lifts with alternating limit points. This certifies that the crossing number of the input curve is at least one.

\subsection{Some preparatory lemmas for Theorem~\lowercase{\ref{th:simple-curve}}}\label{subsec:preparatory-lemmas}
Before studying the recognition of curves homotopic to simple curves, we state some refinements of Lemma~\ref{lem:configurations} and its Corollary that will be useful for the proof of Theorem~\ref{th:simple-curve}. We first highlight a simple property of the system of quads.
\begin{lem}\label{lem:facial-walk}
  The facial walk of a quad cannot contain an arc twice, either with the same or the opposite orientations.
\end{lem}
\begin{proof}
  If the facial walk of a quad contains two occurrences of an arc with the same orientation, then their identification creates a M\"obius strip in the system of quads, in contradiction with its orientability. If the facial walk contains two consecutive occurrences of an arc with opposite orientations, then their common endpoint must have degree one in the system of quads. This contradicts the minimal degree (8 on a closed orientable surface) in our system of quads. Finally, if the facial walk contains non consecutive occurrences of an arc with opposite orientations, then their identification creates a cylinder bounded by loop edges. However, the radial graph being bipartite, this cannot occur.
\end{proof}
\begin{lem}\label{lem:no-common-index}
  Let $c$ be a primitive geodesic curve. If a forward index path $\Ipath{i}{k}$ and a backward index path $\Ipath{j}{-k}$ have homotopic image paths $c\Ipath{i}{k}\sim c\Ipath{j}{-k}$, then they cannot share any index.
\end{lem}
\begin{proof}
  We remark that we cannot have $i=j$ for otherwise $c\Ipath{j-k}{2k} = (c\Ipath{j}{-k})\inv.c\Ipath{i}{k} = 1$ in contradiction with Corollary~\ref{cor:no-monogon}. Likewise, we cannot have $i+k=j-k$. By way of contradiction, suppose that the index paths $\Ipath{i}{k}$ and $\Ipath{j}{-k}$ have a common index $i+r=j-t$ for some integers $0\leq r,t\leq k$. By the previous remark we have $0<\frac{r+t}{2}<k$. It follows that $(m,m)$ is an interior index pair of $(\Ipath{i}{k}, \Ipath{j}{-k})$, where
$m=i+\frac{r+t}{2} = j-\frac{r+t}{2}$ (recall that $i$ and $j$ have the same parity). Consider a disk diagram for $c\Ipath{i}{k}\sim c\Ipath{j}{-k}$. The index pair $(m,m)$ cannot correspond to the same vertex in the diagram since the path $c[i,j]=c[i,m].c[m,j]$ would then label a closed path in the diagram, hence be contractible. This would again contradict  Corollary~\ref{cor:no-monogon}. It ensues that $(m,m)$ corresponds to opposite vertices in a quad of the diagram. This quad is part of a staircase where the arc $c[m,m+1]$ labels two arcs on either sides of the partial diagram. Figure~\ref{fig:m-m-opposite-1} depicts the four possible configurations. 
\begin{figure}
  \centering
  \includesvg[\linewidth]{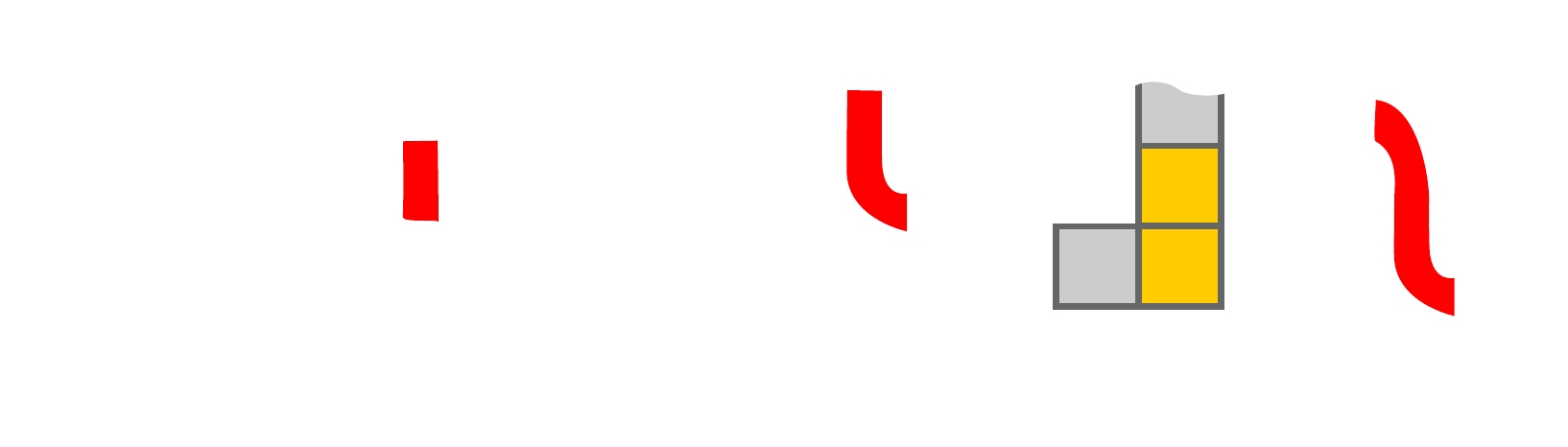}
  \caption{When the index pair $(m,m)$ corresponds to distinct vertices (the thick $m$ dots) in the diagram, the arc $c[m,m+1]$ labels two arcs that may be in one of four possible configurations. In each configuration the yellow quads must be bounded twice by the same arc, contradicting Lemma~\ref{lem:facial-walk}. }
  \label{fig:m-m-opposite-1}
\end{figure}
In each case we infer the existence of a quad bounded twice by the same arc, in contradiction with the previous lemma. 
\end{proof}
\begin{lem}\label{lem:ell-plus-un}
  Two distinct index paths of a nontrivial primitive geodesic curve $c$ having homotopic image paths have length distinct from $|c|$ and at most $|c|+1$.
\end{lem}
\begin{proof}
When the index paths have opposite orientations, the previous lemma directly implies the result. 
Put $\ell=|c|$ and let $\Ipath{i}{k},\Ipath{j}{k}$ be two index paths such that 
$c\Ipath{i}{k}\sim c\Ipath{j}{k}$. Clearly, $k\neq \ell$ since otherwise $c$ would be homotopic to a conjugate of itself by the shorter curve $c[i,j]$, in contradiction with the primitivity of $c$.
For the sake of a contradiction, suppose that $k>\ell+1$. Consider a diagram $\Delta$ with sides $\Delta_R$ and $\Delta_L$ for the homotopic paths $c\Ipath{i}{k}\sim c\Ipath{j}{k}$. By Lemma~\ref{lem:configurations} the vertices $\Delta_R(i+\ell)$ and $\Delta_L(j+\ell)$ must be diagonally opposite in a quad of $\Delta$. Using that $\Delta_R\Ipath{i}{2}$ and $\Delta_R\Ipath{i+\ell}{2}$ are labelled by the same edges and similarly for $\Delta_L\Ipath{j}{2}$ and $\Delta_L\Ipath{j+\ell}{2}$ we easily deduce that the system of quads has a quad with two occurrences of a same arc in its facial walk, or has a vertex of degree at most $5$, or is non-orientable. This would contradict Lemma~\ref{lem:facial-walk} or the hypotheses on the system of quads. In details, Figure~\ref{fig:lengthbigon} depicts the five possible configurations (a,b,c,d,e) for $\Delta_R\Ipath{i}{2}$  and $\Delta_L\Ipath{j}{2}$ in the diagram $\Delta$ and the eight possible configurations (A1-2, B1-3, C1-3) for $\Delta_R\Ipath{i+\ell}{2}$  and $\Delta_L\Ipath{j+\ell}{2}$.
  \begin{figure}
    \centering
    \includesvg[\textwidth]{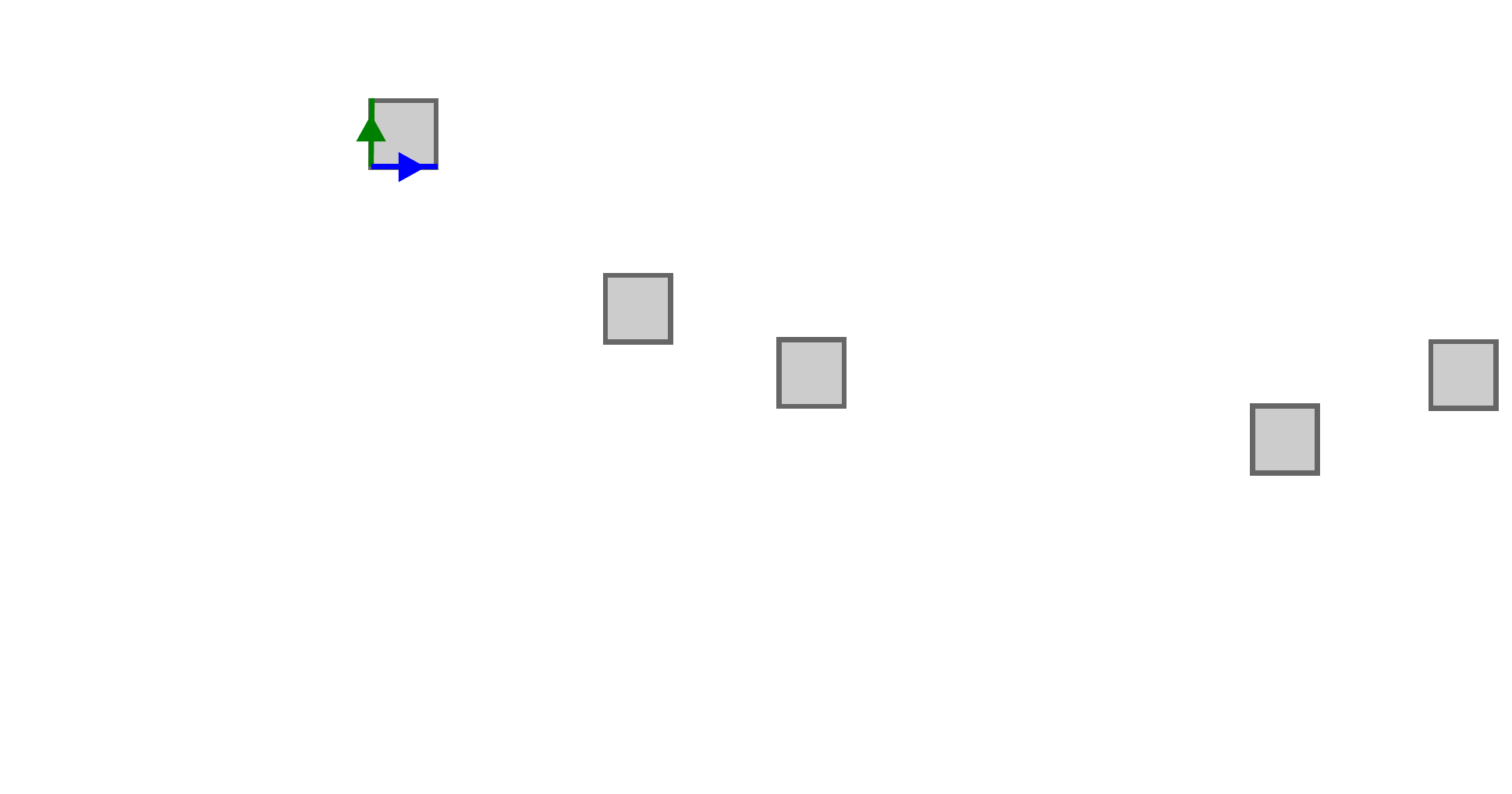}
    \caption{Upper row, the five possible configurations at the beginning of the disk diagram for $c\Ipath{i}{k}\sim c\Ipath{j}{k}$. Middle row, the eight possible configurations in the neighborhood of the index pair $(i+\ell,j+\ell)$ in this diagram. Lower row, case A2 splits into two subcases as we fix the relative orientations of the two quads.}
    \label{fig:lengthbigon}
  \end{figure}
 We remark that cases C1-3 are symmetric to cases B1-3, so that we only need to consider the five configurations A1-2, B1-3. We denote by $(\text{x,Y}) \in \{\text{a,b,c,d,e}\}\times\{\text{A1-2, B1-3}\}$ the conjunction of configurations $x$ and $Y$ in $\Delta$. The quads in configurations a,b,c,d are part of a staircase $\sigma_0$ in $\Delta$. Similarly, the vertices $\Delta_R(i+\ell)$ and $\Delta_L(j+\ell)$
are contained in a staircase $\sigma_\ell$ in $\Delta$ (possibly the same). Being drawn in the plane the staircases of $\Delta$ inherit an orientation from the plane. This orientation is said to be positive or negative according to whether or not it is consistent with some default orientation of the quad system. 
\begin{itemize}
\item 
If the orientations of the two staircases $\sigma_0$  and $\sigma_\ell$ have the same sign, then the face to the right of $\Delta_L[j,j+1]$ has the same facial walk as the face to the right of $\Delta_L(j+\ell, j+\ell+1)$.  The induced identifications imply that for $(\text{x,Y})\in \{\text{a,b,c,d}\}\times\{\text{A1-2, B1-3}\}\cup \{\text{e}\}\times\{\text{A1-2,B1}\}$ one quad is bounded twice by the same arc (with or without the same orientation) in contradiction with Lemma~\ref{lem:facial-walk}. If $(\text{x,Y}) \in \{\text{e}\}\times\{\text{B2-3}\}$, we easily deduce that $c(i+1)=c(j+1)$ has degree four in the system of quads, again a contradiction. 
\item
When the orientations of $\sigma_0$  and $\sigma_\ell$ have opposite signs the faces to the right of $\Delta_L[j,j+1]$ and to the right of $\Delta_L(j+\ell, j+\ell+1)$ (resp. to the left of $\Delta_R[i,i+1]$ and  to the left $\Delta_R[i+\ell,i++\ell+1]$) correspond to the two faces incident to $c[j,j+1]$ (resp. $c[i,i+1]$). We further split case $A2$ into two variants A$2'$ and A$2''$ according to whether the two quads (see Figure~\ref{fig:lengthbigon}) have orientations with the same sign or not, respectively. 
By the induced identifications we derive the following forbidden situations: two occurrences of a same arc in a quad for $(\text{x,Y}) \in \{\text{(a,A1),(e,B1)}\} \cup\{\text{d,e}\}\times\{\text{A1,A}2'\}\cup \{\text{a-e}\}\times \{\text{A}2''\}$, a degree two vertex (namely $c(j+1)$) for $(\text{x,Y}) \in \{\text{a,c}\}\times\{\text{B1,B3}\}$, a degree three vertex for $(\text{x,Y}) \in \{\text{b}\}\times\{\text{B1-3}\}\cup \{\text{a-c}\}\times\{\text{B2}\}$, a degree four vertex for $(\text{x,Y}) \in \{\text{b,c}\}\times\{\text{A1}\}\cup \{\text{e}\}\times\{\text{B2-3}\}\cup \{\text{(a,A$2'$),(d,B1)}\}$ and a degree five vertex for $(\text{x,Y}) \in \{\text{b,c}\}\times\{\text{A}2'\}\cup \{\text{d}\}\times\{\text{B2-3}\}$.
\end{itemize}
\end{proof}
When the orientations of the staircases $\sigma_0$  and $\sigma_\ell$ in the above proof have the same sign it is sufficient to use that $\Delta_R\Ipath{i}{1}$ and $\Delta_R\Ipath{i+\ell}{1}$ are labelled by the same edge, and similarly for $\Delta_L\Ipath{j}{1}$ and $\Delta_L\Ipath{j+\ell}{1}$, in order to reach a contradiction for cases $A1$ and $A2$. In other words, the two  homotopic paths
$c\Ipath{i}{k}\sim c\Ipath{j}{k}$ cannot have length $\ell+1$. This leads to the following refinement.
\begin{lem}\label{lem:short-bigon}
  A bigon $(\Ipath{i}{k}, \Ipath{j}{k})$ of \animmersion{} of a primitive geodesic curve $c$ without intermediate crossings (i.e., no index pair $(i+r, j+r)$, $1<r<k$, is a combinatorial crossing) has length $k<|c|$.
\end{lem}

\begin{lem}\label{lem:flat-bigons}
  If two forward index paths of a canonical curve have homotopic image paths then they actually have the same image paths, i.e. they form a double path. In particular, bigons composed of two forward paths must be flat.
\end{lem}
\begin{proof}
  By Theorem~\ref{th:geodesic} the image paths bound a disk diagram composed of paths and staircase. Remark that a staircase with both sides directed forward must have a $\bar{1}$ turn on its left side. Since the curve is canonical we infer that the diagram cannot have any staircase, implying that the two sides coincide. 
\end{proof}

\subsection{The {\zipper} algorithm}\label{subsec:unzip}
We now turn to the original problem of Poincar\'e~\cite[\S 4]{p-ccal-04}, deciding whether a given curve $c$ is homotopic to a simple curve. In the affirmative we know by
Lemma~\ref{lem:geometric-vs-combinatorial} 
that some geodesic homotopic to $c$ must have a (combinatorial) \define{embedding}, i.e. \animmersion{} without crossings. This also results from Lemma~\ref{lem:swap} and Theorem~\ref{th:singular-bigon} by swapping bigons as much as possible. Rather than swapping the sides of a singular bigon as in Lemma~\ref{lem:swap} we choose to switch one side along the other side. This will also decrease the number of crossings if the bigon contains no other interior bigons. This time, however, one side of the bigon is left unchanged and this allows us to enforce a given edge of $c$ to stay fixed as we remove crossings: each time the edge is involved in a bigon switch, we switch the side that does not contain the edge.
This suggests an incremental computation of an embedding in which a larger and larger part of the curve is fixed as we switch bigons: we assume that $c$ is canonical and consider the trivial embedding of its first arc occurrence $[0,1]$. We next insert the successive arc occurrences incrementally to maintain, if possible, an embedding of the path formed by the already inserted arcs. Note that given an embedding of the subpath $c\Ipath{0}{i}$, either this embedding cannot be extended to an embedding of  $c\Ipath{0}{i+1}$, or there is exactly one way to do so.
When inserting the occurrence $[i,i+1]$ we need to compare its left-to-right order with each  already inserted arc occurrence $\beta$ of its supporting arc.
\begin{itemize}
\item 
  If $\beta\neq [0,1]$ we can use the comparison of the occurrence $[i-1,i]$ with the occurrence $\gamma$ preceding $\beta$ (or succeeding $\beta$ if it is a backward occurrence). If $[i-1,i]$ and $\gamma$ have the same supporting arc, we just propagate their relative order to $[i,i+1]$ and $\beta$. Otherwise, we use the circular ordering of the supporting arcs of $[i-1,i]$, $\gamma$ and $[i,i+1]$ to deduce the relative order of $[i,i+1]$ and $\beta$ that avoids a crossing at the double point $(i,j)$, where $j$ is the common index of $\beta$ and $\gamma$.
 \item When $\beta= [0,1]$, we cannot use the occurrence preceding $[0,1]$ as it is not yet inserted. We rather compare $[i,i+1]$ and $[0,1]$ as follows. In the Poincar\'e disk, we consider two lifts $\tilde{d}_i$ and $\tilde{d}_0$ of $c$ such that $\tilde{d}_i[i,i+1]=\tilde{d}_0[0,1]$. We decide to insert $[i,i+1]$ to the left (right) of  $[0,1]$ if one of the limit points of $\tilde{d}_i$ lies to the left (right) of $\tilde{d}_0$. Note that when $c$ is homotopic to a simple curve the two limit points of $\tilde{d}_i$ should lie on the same side of $\tilde{d}_0$.
 \end{itemize}

After comparing $[i,i+1]$ with all the occurrences of its supporting arc, we can insert it in the correct place. If no crossings were introduced this way, we proceed with the next occurrence $[i+1,i+2]$. It may happen, however, that no matter how we insert  $[i,i+1]$ in the left-to-right order of its supporting arc, the resulting \immersion{} of $c\Ipath{0}{i+1}$ will have a combinatorial crossing. In order to handle this case, we first check if $[i,i+1]$ is \define{switchable}, i.e. if for some $k\geq 0$ and some turns $t,u$ the subpath $p:=c\Ipath{\Imath}{k+2}$ has turn sequence $t2^k1u$ and the index path  $\Ipath{\Imath}{k+2}$ does not contain the arc occurrence $[0,1]$. See Figure~\ref{fig:switch-path}.
When $[i,i+1]$ is switchable we can switch $p$ to a new subpath $p'$ with turn sequence $(t-1)\bar{1}\bar{2}^{k}(u-1)$ such that  $p$ and $p'$ bound a diagram composed of a single horizontal staircase. 
\begin{figure}
  \centering
  \includesvg[\textwidth]{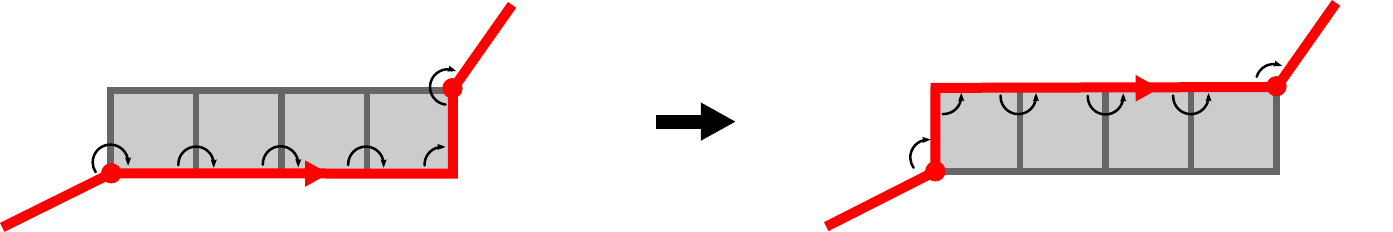}
  \caption{The arc $[i,i+1]$ is switchable.}
  \label{fig:switch-path}
\end{figure}
We indeed perform the switch \emph{if some intersection is actually avoided this way}. More formally, this happens when there is a crossing $(i,j)$ such that the turn of $(c[j,j+\varepsilon], c[i,i+1])$ is one with either $\varepsilon = 1$ or $\varepsilon = -1$ as illustrated on Figure~\ref{fig:switch-path-bis}.
\begin{figure}
  \centering
  \includesvg[\textwidth]{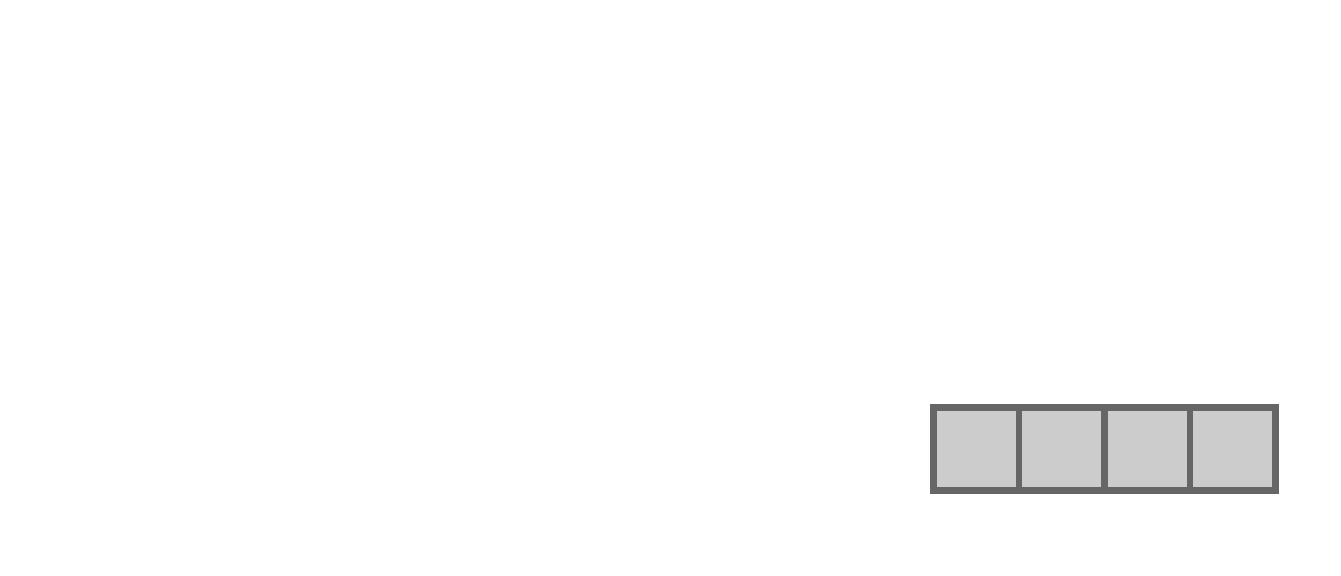}
  \caption{Top, the arc occurrence $[j,j+1]$ triggers a switch. Bottom, $[i,i+1]$ is switchable but not switched in these two situations.}
  \label{fig:switch-path-bis}
\end{figure}
In such a case we say that $[j,j+\varepsilon]$ has \define{triggered} the switch.
Note that replacing $p$ by $p'$ leads to a curve $c'$ that is still geodesic and homotopic to $c$. Moreover, \emph{the non traversed part of $c'$ starting at index $i+1$ is in canonical form} as it contains no $\bar{1}$ turns. We then insert the arc occurrence $[i,i+1]$ and proceed with the algorithm using $c'$ in place of $c$. 
The successive switches in the course of the computation untangle  $c$ incrementally. We call our embedding procedure the \define{{\zipper} algorithm}. For further reference we make some observations that easily follows from the above switch procedure and the fact that $c$ is initially canonical.
\begin{remark}\label{rem:switchable}
  If $[i,i+1]$ is switchable but not switched at the time of being processed,  then it remains not switched until the end of the algorithm, meaning that the vertex of index $i$ is directly followed by a subpath $p$ with a turn sequence of the form $2^k1$ ($k$ may get smaller if some other arc of $p$ is switched). Similarly, if $[i,i+1]$ is not switchable because of an inappropriate turn sequence, this remains true until the end of the algorithm.
\end{remark}
\begin{remark}\label{rem:01-right-side}
   Let $c'$ be a curve homotopic to $c$ obtained after a certain number of switches during the execution of the {\zipper} algorithm. Then $c'$ cannot contain a subpath with turn sequence $\bar{1}\bar{2}^*$ that ends at index $0$, nor can it contain a subpath with turn sequence $\bar{2}^*\bar{1}$ that starts at index $0$.
\end{remark}
\begin{lem}\label{lem:zipper-complexity}
  The {\zipper} algorithm applied to a canonical primitive curve $c$ of length $\ell$ can be implemented to run in $O(\ell\log\ell)$ time.
\end{lem}
\begin{proof}
  We shall describe an implementation of the unzip algorithm with the announced time complexity. This algorithm constructs a combinatorial perturbation of $c$ incrementally adding its arcs one after the other. The proof that this perturbation is an embedding if and only if $c$ is homotopic to a simple curve is deferred to Proposition~\ref{prop:zipper}. Before we start the incremental construction, we need to perform some preprocessing, namely marking the switchable arcs of $c$ and precomputing a left-to-right comparison between the first arc of $c$ and the other occurrences of that arc in $c$.
\paragraph{Marking the switchable arcs.}
  We first traverse $c$ in reverse direction to mark all the switchable arcs. In the course of the algorithm,
each time a switch applies to a subpath $p$ we unmark all the arcs in $p$ except the last one that may become switchable (depending on the last turn of $p$ and on the status of the arc following $p$). It easily follows that an arc can be switched at most twice and that the amortized cost for the switches is linear. Alternatively, we could use the run-length encoded turn sequence of $c$ as defined in~\cite{ew-tcsr-13} to detect each switchable arc  and update the turn sequence in constant time per switches.
\paragraph{Comparisons with the first arc of $c$.}
In the preprocessing phase we also compute the relative order of $[0,1]$ with all the other occurrences of $c[0,1]$ as follows. If $c[i,i+1]=c[0,1]$, the corresponding arc occurrences form a double path of length one and we compute their relative order by extending maximally this double path in the backward direction. Looking at the tip of this double path, say $\Dpoint{j}{k}$, we can decide which side is to the left of the other. Indeed, the three arcs $c\Bocc{\Jmath}$, $c\Bocc{k}$ and $c\Focc{j}=c\Focc{k}$ must be pairwise distinct  and their circular order about their common origin vertex in the system of quads provides the necessary information as follows from Lemma~\ref{lem:flat-bigons}. 
The computation of the maximal extensions in the backward direction amounts to evaluating the longest common prefix of $c\inv$ with all its circular shifts. Overall, this can be done in $O(\ell)$ time thanks to a simple variation of the Knuth-Morris-Pratt algorithm.
\paragraph{The incremental construction.}
Recall that a combinatorial perturbation of $c$ is the data for each arc of the system of quads of an ordering of the arc occurrences of $c$ equal to that arc or its opposite. In practice, we store each of these orderings into a (balanced) binary search tree to allow the left-to-right comparison of arcs in logarithmic time. 
Although an edge is made of two opposite arcs, we shall store only one ordering per edge, assuming a default orientation of the edge. The orderings associated to opposite arcs should indeed be opposite, so that there is no need to store both of them. 

We start the incremental construction of the combinatorial perturbation by initializing an empty search tree for every edge of the system of quads.
We now traverse $c$ in the forward direction starting with  $[0,1]$ and insert each traversed arc occurrence $\alpha=[i,i+1]$ in its tree, possibly after switching $\alpha$. The switches may modify $c$ so that we denote by $c_k$ the geodesic homotopic to $c$ resulting from the first switches in the algorithm up to and including the insertion of the arc occurrence $[k,k+1]$. The insertion of $\alpha$ is composed of two steps.
We check in a first step if $\alpha$ needs to be switched, that is if $\alpha$ is switchable and some already inserted arc occurrence indeed triggers its switch. This last condition happens when one of  the occurrences of the arc that makes a one turn with $c_{i-1}(\alpha)$ defines a crossing with $\alpha$ at their common endpoint (see Figure~\ref{fig:switch-path-bis}). 
This can be  easily tested in $O(\log\ell)$ time using a dichotomy over these occurrences. We perform the switch if the test is positive and obtain a possibly new curve $c_i$. In a second step we insert $\alpha$ in the tree of its supporting arc $c_i(\alpha)$. If this search tree is empty, we just insert $\alpha$ at the root. Otherwise, we perform the usual tree insertion; as we go down the tree we need to perform $O(\log\ell)$ comparisons, each time comparing $\alpha$ with some already inserted arc occurrence $\beta$ such that $c_i(\beta)=c_i(\alpha)$. There are three cases to consider.
\begin{itemize}
\item If $\beta\neq [0,1]$ we can use the comparison of the occurrence $[i-1,i]$ preceding $\alpha$ with the occurrence $\beta'$ preceding $\beta$ (or succeeding $\beta$ if it is a backward occurrence) in order to decide the relative order of $\alpha$ and $\beta$. If moreover $c_i(\beta')\neq c_i[i-1,i]$ this comparison can be done in constant time: we can index the arcs around each vertex of the system of quads in a prepossessing step and use index arithmetic to make the necessary comparisons between the supporting arcs of  $\beta'$, $[i-1,i]$ and $\alpha$.
\item If $\beta\neq [0,1]$ as above, while $\beta'$ and $[i-1,i]$ have the same supporting arc, we  can use the search tree of $c_i(\beta')$ to perform the comparison. This would cost $O(\log\ell)$ time for each such $\beta'$, implying a $O(\log^2\ell)$ cost for inserting $\alpha$ in its tree. However, we can gather the comparisons of $\alpha$ with all such $\beta'$'s into a single search. Indeed, if the computed perturbation for the first $i-1$ arc occurrences is an embedding, then the $\beta'$'s must be contiguous in the left-to-right order of $c_i([i-1,i])$ and must appear in the same order as the corresponding $\beta$'s in the search tree of $c_i(\alpha)$. It follows that we can determine the relative order of $\alpha$ with all the $\beta$'s thanks to a single search of  $[i-1,i]$ in the search tree of $c_i([i-1,i])$. The overall cost is $O(\log\ell)$. If the computed perturbation for $\Ipath{0}{i}$ happens to be non-simple, then the $\beta'$'s need not be contiguous and
this procedure may return an incoherent answer. However, this would certify that the unzip algorithm has already failed to produce an embedding and we need not worry about incoherent answers.
\item 
When $\beta=[0,1]$ we can use our  precomputed comparisons unless $\alpha$ was previously switched, thus not compared with $[0,1]$ during the preprocessing phase. Since $c$ is canonical, arcs may only be switched to their left in one of four possible configurations as on Figure~\ref{fig:under_a0} 
\begin{figure}
  \centering
  \includesvg[.8\textwidth]{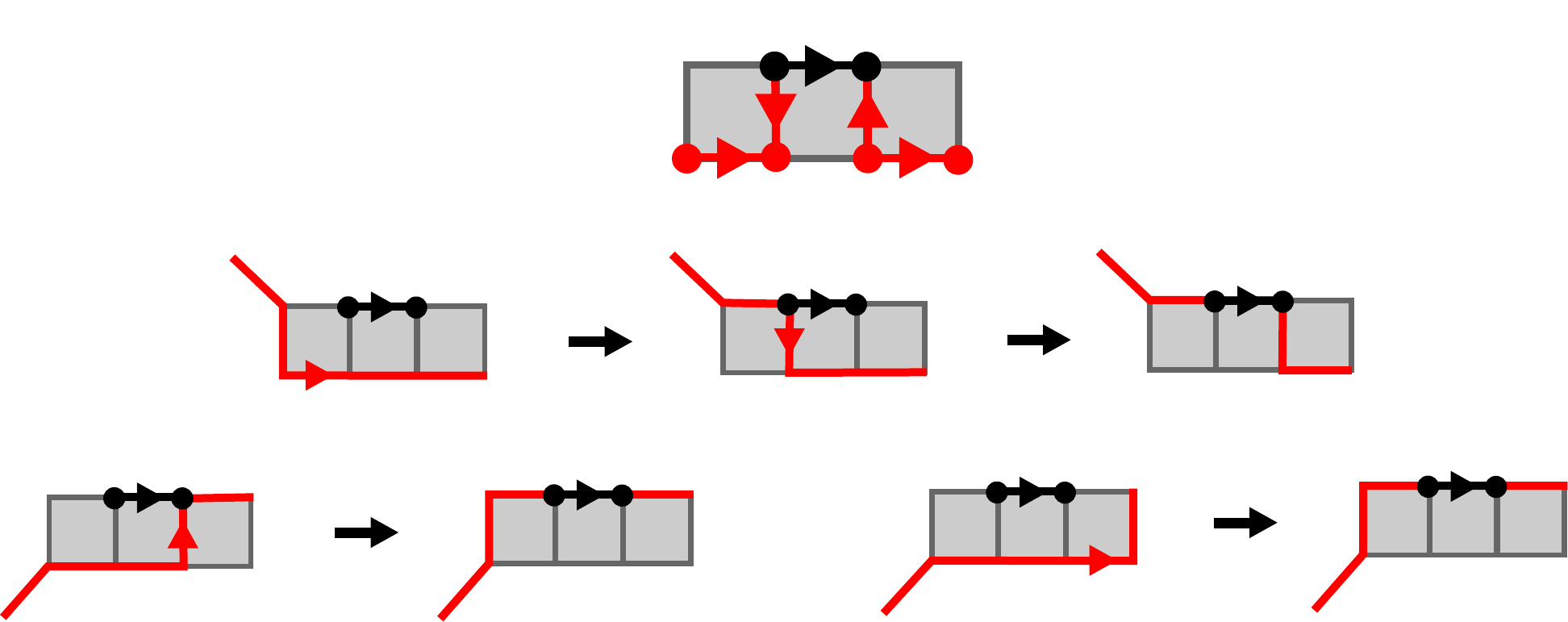}
  \caption{Top: the arc $c(\alpha)$ may be in one of four possible configurations with respect to the first arc $c[0,1]$. Middle and bottom row: from its initial canonical position, $\alpha$ is switched to the same arc as $[0,1]$.}
  \label{fig:under_a0}
\end{figure}
and we infer that $\alpha$ must lie to the right of $[0,1]$ as justified by Lemma~\ref{lem:relative-pos} below.  
\end{itemize}
In conclusion, $\alpha$ can be inserted in its search tree in $O(\log\ell)$ time. Inserting each arc occurrence in turn thus leads to an $O(\ell\log\ell)$ time algorithm.

Note that after $\alpha$ is inserted we do not try to determine if the current \immersion{} has a crossing or not. This will be checked in a second  step after the {\zipper} algorithm is completed. 
\end{proof}
\begin{lem}\label{lem:relative-pos}
  Let $\tilde{c}$ and $\tilde{d}$ be two lifts of a primitive canonical geodesic in the Poincar\'e disk $\D$. The limit points of $\tilde{c}$ cut the boundary circle $\partial \D$ into two pieces. By the \emph{right piece}, we mean the piece of $\partial \D$ that bounds the part of $\D\setminus \tilde{c}$ to the right of $\tilde{c}$. Suppose that $\tilde{c}$ contains an arc $a$ and that $\tilde{d}$ contains an arc $b$ such that $a$ and $b$ are in one of the four relative positions depicted on Figure~\ref{fig:relative-pos}. Then one of the limit points of $\tilde{d}$ is in the right piece of $\partial \D$.
\end{lem}
\begin{figure}
  \centering
  \includesvg[.7\linewidth]{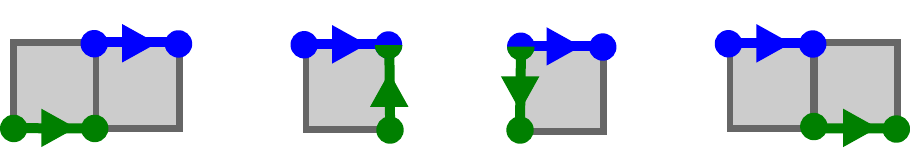}
  \caption{The relative positions of $a$ and $b$.}
  \label{fig:relative-pos}
\end{figure}
\begin{proof}
  We first note that $b$ itself lies to the right of $\tilde{c}$. Indeed, $\tilde{c}$ would have to use a $\bar{1}$ turn to see $b$ on its left or to pass along $b$, in contradiction with its canonicity. If both limit points of $\tilde{d}$ were to the left of $\tilde{c}$, then  $\tilde{c}$ and $\tilde{d}$ would form a bigon with one staircase part containing $b$. The $\tilde{c}$ side of this staircase would thus see $b$ on its right, which is impossible since $\tilde{c}$ is canonical (here, we have not used that $\tilde{d}$ is canonical but only geodesic).
\end{proof}

\begin{prop}\label{prop:zipper}
  If $i(c)=0$ the {\zipper} algorithm returns an embedding of a geodesic homotopic to $c$.
\end{prop}
\begin{proof}
  Let $\Pi$ be the \immersion{} computed by the {\zipper} algorithm. We denote by $c_k$ the geodesic homotopic to $c$ resulting from the first switches in the algorithm up to and including the insertion of the arc occurrence $[k,k+1]$. Note that $k>r$ implies $c_k\Ipath{0}{r+1}=c_r\Ipath{0}{r+1}$.
Suppose that $\Pi$ has a crossing. For the rest of this proof, we denote by $i$ the smallest index such that the insertion of $[i,i+1]$ creates a crossing, i.e. the restriction of $\Pi$ to $\Ipath{0}{i+1}$ has a crossing double point while its restriction to $\Ipath{0}{i}$ is an embedding. By convention we set $i=\ell$ if the crossing appeared after the last arc insertion. We shall show that $c_{\ell-1}$ has two lifts whose limit points are alternating on the boundary of the Poincar\'e disk. It will follow from Section~\ref{sec:strategy}  that $i(c)=i(c_{\ell-1})>0$ thus proving the Proposition. We first establish some preparatory claims.
  \begin{claim}\label{claim:branched-crossing}
    If $(i,j)$ is a crossing of $\Pi$ with $0<j<i<\ell$, then the backward arc $c_i[j,j-1]$ and the forward arc $c_i[j,j+1]$ are distinct from the supporting arc $c_i[i,i+1]$ of $[i,i+1]$. Moreover, if $[i,i+1]$ was switched just before its insertion, then $c_i[j,j-1]$ and $c_i[j,j+1]$ are also distinct from the supporting arc $c_{i-1}[i,i+1]$ of $[i,i+1]$ just before its switch.
  \end{claim}
  \begin{proof}
    The first part of the claim follows directly from our insertion procedure and the fact that the restriction of $\Pi$ to $\Ipath{0}{i}$ is an embedding. 
Moreover, assume that $[i,i+1]$ was switched just before its insertion.
By the insertion procedure, this means that the insertion of $[i,i+1]$ before the switch would have induced a removable crossing 
$(i,k)$ where $0<k<i$ and both $c_{i-1}[k,k+1]$ and $c_{i-1}[k,k-1]$ are distinct from $c_{i-1}[i,i+1]$ by the first part of the claim. By the same first part we have $c_i[j,j-1], c_i[j,j+1]\neq c_i[i,i+1]$. Since the restriction of $\Pi$ to $c_i\Ipath{0}{i}$ has no crossing, the length 2 path $c_{i-1}\Ipath{k-1}{2}=c_{i}\Ipath{k-1}{2}$ separates $c_{i}\Ipath{j-1}{2}$ from $c_{i-1}[i,i+1]$, implying $c_i[j,j-1], c_i[j,j+1]\neq c_{i-1}[i,i+1]$ as desired. Figure~\ref{fig:claim1-2} depicts the situation.
  \end{proof}
  \begin{claim}\label{claim:parallel-right}
    If $c_k$ has a $\bar{1}$ turn at index $k+1$, with $0<k< i-1$, then there is an index $r$ with $0< r<k$ such that
    \begin{itemize}
    \item $c_{k}[r,r-1]=c_{k}[k,k+1]$,
    \item the arc $c_k[r,r-1]$ lies to the right\footnote{More formally, the arc occurrence $[r,r-1]$ lies to the right of $[k,k+1]$ in the restriction of $\Pi$ to $\Ipath{0}{k+1}$.} of $c_k[k,k+1]$,
\item if $r\neq 1$, then $c_k$ has a 1 turn at $r-1$.
    \end{itemize}
  \end{claim}
  \begin{proof}
Refer to Figure~\ref{fig:claim1-2} for the claim and its proof.
\begin{figure}[h]
  \centering
  \includesvg[.8\textwidth]{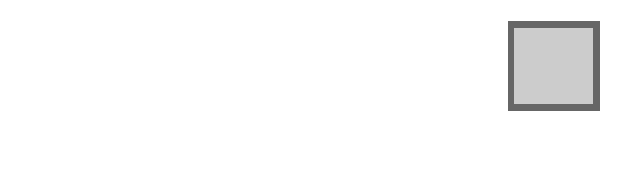}
  \caption{Left, illustration for the proof of claim 1 when $[i,i+1]$ was switched just before its insertion. $c_{i}$ and $c_{i-1}$ appear respectively in plain and doted lines. Right, illustration for claim 2.}
  \label{fig:claim1-2}
\end{figure}
A $\bar{1}$ turn can only occur at the destination of an arc that has been switched. It follows that the arc $c_{k}[k,k+1]$ must have been switched just before its insertion, thus witnessing the existence of a triggering arc $a$ of the form $c_{k}[u,u-1]$ or $c_{k}[u,u+1]$, with $0<u<k$, that lies parallel to and to the right of $c_{k}[k,k+1]$ (with respect to $\Pi$).
In the case $a=c_{k}[u,u-1]$, we may set $r=u$. 
Otherwise, $a=c_{k}[u,u+1]$ and
recalling that the restriction of $\Pi$ to $\Ipath{0}{i}$, hence to $\Ipath{0}{k+2}$, has no crossing\footnote{When $i=\ell$, a crossing of the form $(0,j)$ is not considered as a crossing of the restriction of $\Pi$ to $\Ipath{0}{\ell}$.}, it must be that $c_{k}\Ipath{u}{2}$ also lies parallel to and to the right of $c_{k}\Ipath{k}{2}$.  So $c_i$ has a $\bar{1}$ turn at index $u+1$, and we are back to the hypothesis of the claim, decreasing the value of $k$ to $u$. We can repeat the same arguments inductively, each time decreasing the value of $k$ in the claim. Since $k\geq 1$, the process must stop, implying that we have reached the former case. It remains to note that whenever $r>1$ the double point $(k+1,r-1)$ cannot be a crossing by the choice of $i$, so that $c_{k}$ makes a 1 turn at $r-1$.
 \end{proof}

Recall that $i$ is the smallest index such that the restriction of $\Pi$ to $\Ipath{0}{i+1}$ has a crossing. We denote by $c'=c_{\ell-1}$ the geodesic homotopic to $c$ resulting from all the switches in the course of the algorithm execution.  Let $j$ be the minimal index such that $(i,j)$ is a crossing of $\Pi$.
We consider two lifts $\tilde{d_i}$ and $\tilde{d}_j$ of $c'$ in the Poincar\'e disk such that $\tilde{d}_i(i) = \tilde{d}_j(j)$. We first suppose $i<\ell$. 
\begin{claim}\label{claim:no-crossing-after-i}
  The piece of lift $\tilde{d}_i\Ipath{i}{+\infty}$ has no crossing with $\tilde{d}_j$. (Crossings are defined with respect to the lift of $\Pi$ in the Poincar\'e disk.)
\end{claim}
\begin{proof}
  For the sake of a contradiction, suppose that $\tilde{d}_i\Ipath{i}{+\infty}$ and $\tilde{d}_j$ cross. Let $s$ be the smallest positive integer such that $\tilde{d}_i$ and $\tilde{d}_j$ crosses at $\tilde{d}_i(i+s)$. We thus have a bigon of the form $(\Ipath{i}{s},\Ipath{j}{\varepsilon s})$ for some $\varepsilon\in \{-1,1\}$. Moreover, this bigon has no intermediate crossings by the choice of $s$. Let $\Delta$ be a disk diagram for this bigon, oriented consistently with the system of quads.  $\Delta$ must start with a staircase part by Claim~\ref{claim:branched-crossing}. In particular, the turn $t$ between $c'[j,j+\varepsilon]$ and $c'[i,i+1]$ should be $\pm 1$.
The {\zipper} algorithm may have run across four possible situations at step $i$.
\begin{enumerate}
\item\label{enum:1} Either $[i,i+1]$ was switchable just before its insertion but was not switched, 
\item \label{enum:2} or it was switchable and switched, 
\item \label{enum:3} or it was not switchable because of an inappropriate turn sequence, 
\item \label{enum:4} or it was not switchable because the part to be switched contains $[0,1]$. 
\end{enumerate}
In the first situation, we know by Remark~\ref{rem:switchable} that $c'(i)$ is followed by a turn sequence of the form $2^*1$. Hence, $t$ is exactly $1$; but this contradicts the fact that $[i,i+1]$ was not switched though switchable. Indeed, $[j,j+\varepsilon]$ should have triggered the switch. The third situation together with Remark~\ref{rem:switchable} also lead to a contradiction as the inappropriate turn sequence prevents $\Ipath{i}{s}$ from being part of any staircase. Thanks to Claim~\ref{claim:branched-crossing}, the second situation equally prevents $\Ipath{i}{s}$ from being part of any staircase. It remains to consider the fourth situation. We first suppose that $\varepsilon=-1$, i.e. that the diagram $\Delta$ corresponds to the bigon $(\Ipath{i}{s},\Ipath{j}{-s})$. By Lemma~\ref{lem:no-common-index}, we have $s<\ell$ and the fourth situation implies that the $\Ipath{i}{s}$ side of the bigon contains $[0,1]$. These two properties imply that $0<i+s-\ell<i$. It ensues that the $\Ipath{j}{-s}$ side contains index $i$, for otherwise $(i+s,j-s)$ would be a crossing occurring before step $i$. However, the occurrence of $i$ on both sides of the bigon again contradicts Lemma~\ref{lem:no-common-index}. Hence, it must be that $\varepsilon=1$, i.e. that $\Delta$ is bounded by two forward paths.
Since the $\Ipath{i}{s}$ side of $\Delta$ contains $[0,1]$, the $\Ipath{j}{s}$ side
must lie to the left of $\Delta$ and thus contain $\bar{1}$ turns.
Let $k+1$ be the first index of a $\bar{1}$ turn along this side. Note that $j\leq k<i-1$ because $i$ is followed by a turn sequence of the form $2^*1$.
Claim~\ref{claim:parallel-right} ensures the existence of a smaller index $r<k$ such that $c_{k}[k,k+1]=c_{k}[r,r-1]$ with $c_{k}[r,r-1]$ to the right of $c_{k}[k,k+1]$. See Figure~\ref{fig:claim-3}. 
\begin{figure}[h]
  \centering
  \includesvg[\textwidth]{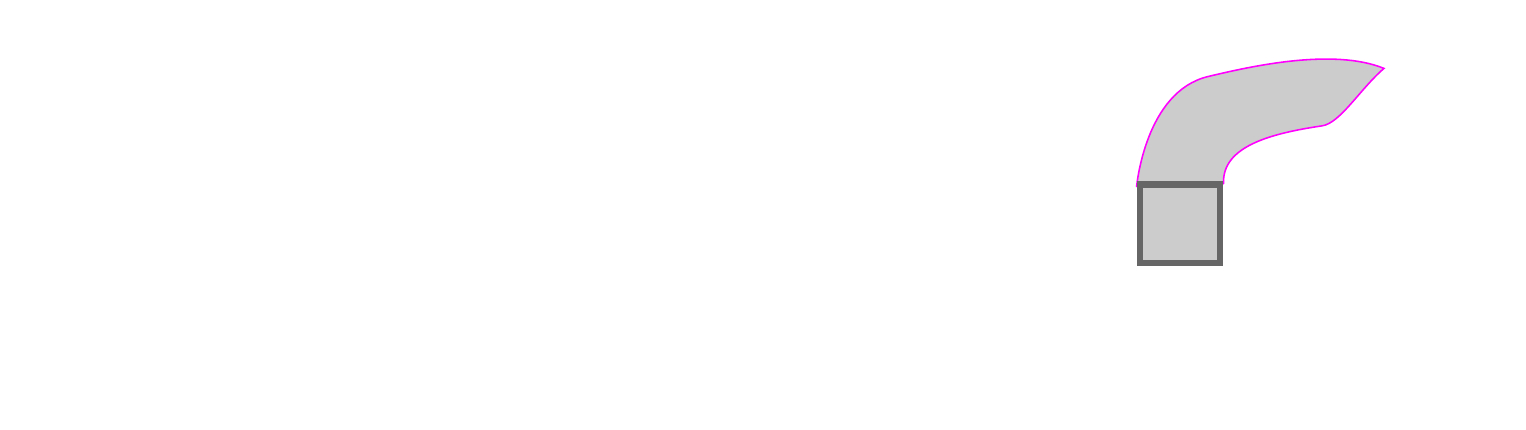}
  \caption{$\Delta$ may start with a vertical step (left) or a horizontal step (right).}
  \label{fig:claim-3}
\end{figure}
If $r>1$, Claim~\ref{claim:parallel-right} further states that $c_k$ has a $1$ turn at index $r-1$. This is also true when $r=1$ and $\Delta$ starts with a vertical step: situation (4) implies a  $1$ turn at index $0$ in this case. It follows in both cases that $c_i\Ipath{r}{}$ is stuck inside the first staircase step of $\Delta$: it cannot turn left as this would create a bracket; it cannot turn right as this would imply a crossing lexicographically smaller than $(i,j)$; and it can neither go straight as it would also imply a crossing lexicographically smaller than $(i,j)$. To see this we remark that the indices in $\Ipath{r}{k+1-j}$ are each smaller than $j$ for otherwise it would be followed by a 1 turn, inducing a bracket. Similar arguments holds in the case where $r=1$ and $\Delta$ starts with a horizontal step. We conclude that situation (4) also leads to a contradiction and that $\tilde{d}_i\Ipath{i}{+\infty}$ and $\tilde{d}_j$ indeed have no crossing.
\end{proof}
We next consider the smallest positive integer $r$ such that $\tilde{d}_i$ and $\tilde{d}_j$ crosses at $\tilde{d}_i(i-r)$. If no such $r$ exists, then by the above Claim~\ref{claim:no-crossing-after-i}, $\tilde{d}_j$ and $\tilde{d}_i$ have a unique intersection point and we may conclude that their limit points alternate. We can thus assume the existence of a bigon  $(\Ipath{i-r}{r},\Ipath{j-\varepsilon r}{\varepsilon r})$ with $r>0$ and $\varepsilon\in\{-1,1\}$, and without intermediate crossings.
We first examine the case $\varepsilon=-1$ where the bigon $(\Ipath{i-r}{r},\Ipath{j+ r}{-r})$ has oppositely oriented index paths. We must have $r\geq i-j$ for otherwise the tip $(i-r,j+r)$ of the bigon would define a crossing occurring before step $i$, contradicting the choice of $i$. Hence, $\Ipath{i-r}{r}$ contains $j$. This is however impossible by Lemma~\ref{lem:no-common-index}.
We now look at the case of two forwards index paths ($\varepsilon=1$). 
We must have $r\geq j$ for otherwise the tip $(i-r,j-r)$ of the bigon would define a crossing occurring before step $i$, again contradicting the choice of $i$. It follows that $\Ipath{j-r}{r}$ contains $[0,1]$ and that the forward branches $\tilde{d}_i\Ipath{i-j}{+\infty}$ and $\tilde{d}_j\Ipath{0}{+\infty}$ have a unique intersection point. The bigon labels a disk diagram $\Delta$ composed of paths and staircases as described in Theorem~\ref{th:geodesic}. 
\begin{itemize}
\item 
If $[0,1]$ and $\alpha:= [i-j,i-j+1]$ label the same arc of a path part in $\Delta$ there are two possibilities: either  it holds initially that $c[0,1]=c(\alpha)$ or $\alpha$ was switched in the course of the algorithm. In the former case, we know by the preprocessing phase and Lemma~\ref{lem:flat-bigons} that the left-to-right order of $[0,1]$ and $\alpha$ is coherent with its extension in the backward direction. This implies that ultimately in the backward direction $\tilde{d}_i$  lies on the same side of $\tilde{d}_j$ as does $\alpha$. In the latter case, as described in the end of the proof of Lemma~\ref{lem:zipper-complexity}, we know by the insertion procedure that at least one of the limit points of $\tilde{d}_i$  lies on the same side of $\tilde{d}_j$ as does $\alpha$. 
Since $(i,j)$ is the only crossing in the forward direction, we conclude in both cases that $\tilde{d}_j$  and $\tilde{d}_i$ have alternating limit points. 
\item Otherwise, $[0,1]$ and $\alpha$ label two distinct arcs, say $a_0$ and $a_{i-j}$, of a staircase part $\sigma$ of $\Delta$.  Let $\Ipath{i-v}{v-u}$ and  $\Ipath{j-v}{v-u}$, $0\leq u<v\leq r$, be the index paths corresponding to the sides $\sigma_L$ and $\sigma_R$ of $\sigma$ as pictured in Figure~\ref{fig:push-1}. By Remark~\ref{rem:01-right-side}, $\Ipath{j-v}{v-u}$ must label the right side $\sigma_R$ while  $\Ipath{i-v}{v-u}$ cannot contain $[0,1]$. It follows that $v<i$, whence $\Ipath{i-v}{v-u}\subset \Ipath{1}{i-1}$. If $a_{i-j}$ belongs to a horizontal part of $\sigma_L$, the first vertex in this  part has a $\bar{1}$ turn and we let $m$ be the index of this vertex. Otherwise, $a_{i-j}$ belongs to a vertical part whose last vertex has a $\bar{1}$ turn. It follows that the vertex of index
$i-j+1$ had a $\bar{1}$ turn at step $i-j$ of the algorithm. We set $m=i-j+1$ in this case. See Figure~\ref{fig:push-1}. 
  \begin{figure}
    \centering
    \includesvg[\textwidth]{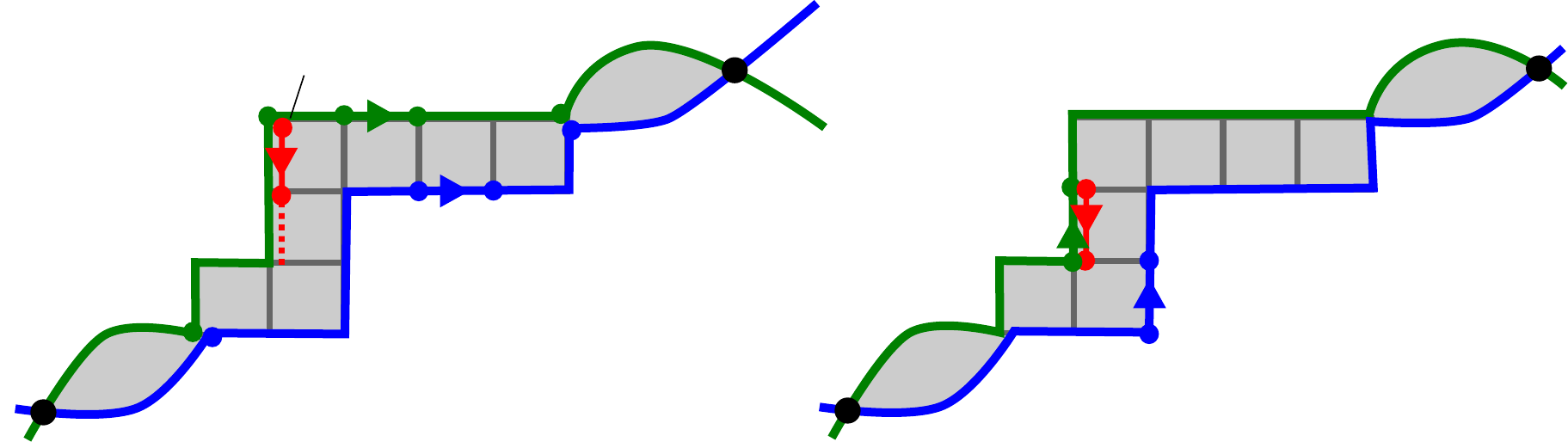}
    \caption{$a_{i-j}$ may belong to a horizontal (left figure) or vertical (right figure) part of $\sigma_L$.}
    \label{fig:push-1}
  \end{figure}
In both cases, $c_{m-1}$ had a $\bar{1}$ turn at index $m$ and by Claim~\ref{claim:parallel-right}, there is an arc occurrence $[x,x-1]$, with $x< m-2$, that lies to the right of $a_{i-j}$. We view $\Delta$ as a subset of the Poincar\'e disk so that $\Delta_L$ and  $\Delta_R$ can be seen as portions of $\tilde{d}_i$ and $\tilde{d}_j$ respectively. Let $\tilde{q}$ be the lift of $c'$ that extends the above occurrence $[x,x-1]$ in $\D$. 
We denote by ${q}_+:=\tilde{q}\Ipath{x}{+\infty}$ and ${q}_-:=\tilde{q}\Ipath{x}{-\infty}$ the portion of $\tilde{q}$ respectively after and before its vertex with index $x$. See Figure~\ref{fig:push-2}.
  \begin{figure}
    \centering
    \includesvg[.6\textwidth]{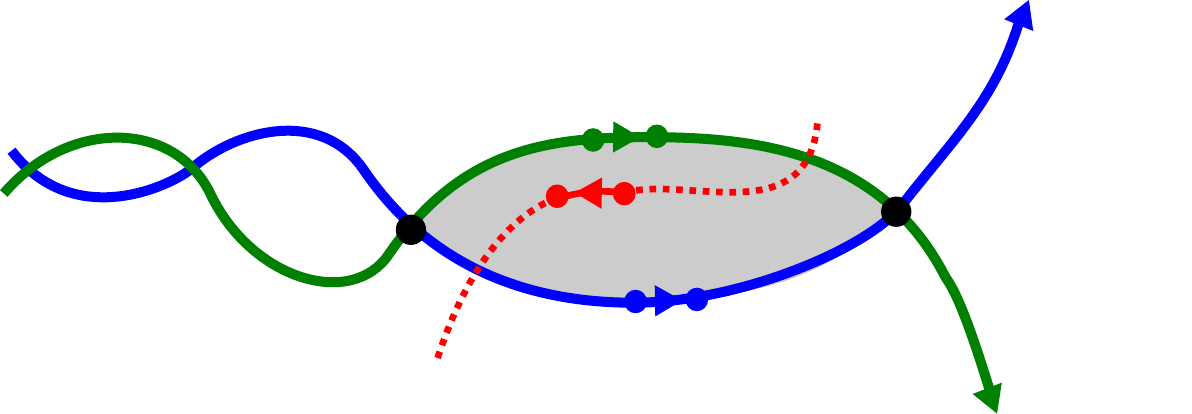}
    \caption{$q_+$ cannot cross $\tilde{d}_i$, $q_-$ cannot cross $\tilde{d}_j$, and $\tilde{d}_j$ cannot cross $\tilde{d}_i\Ipath{i}{+\infty}$.}
    \label{fig:push-2}
  \end{figure}

We claim that $q_+$ cannot cross $\tilde{d}_i$.
Otherwise, we would get a pair of homotopic paths, one piece of $q_+$ starting at index $x$ and one piece of 
$\tilde{d}_i$ ending at index $m-1$, with opposite orientations. By Lemma~\ref{lem:no-common-index}, the $q_+$ piece would not contain index $m-1$ and would thus have length at most $m-1-x < m-1$. In turn, this would imply that the $\tilde{d}_i$ piece would only contain vertices with indices in $[1,m]$. The crossing of $q_+$ and $\tilde{d}_i$ would thus occur before step $i$, a contradiction.

We next claim that $q_-$ cannot cross $\tilde{d}_j$. 
To see this, first note that $q_+$ crosses $\tilde{d}_j$ before $a_0$ (i.e., at a point with index in $\Ipath{j-r}{r-j}$) by the previous claim. If $q_-$ crossed $\tilde{d}_j$, say at  $q_-(s)=\tilde{d}_j(t)$, then we would have a bigon between $\tilde{q}$ and $\tilde{d}_j$ whose $\tilde{d}_j$ side contains $a_0$. Since no crossing occurs before step $i$, we would have either $s\geq i$ or $t\geq i$. We infer in both cases that the index path of the $\tilde{q}$ side of the bigon contains index 0. This would contradict Lemma~\ref{lem:no-common-index}.

It follows from the above claims that the two lifts $\tilde{q}$ and $\tilde{d}_i$ of $c'$ have alternating limit points. 
\end{itemize}
It remains to consider the case $i=\ell$ where the first crossing appears after the last arc insertion. Since all the arcs have been inserted without introducing crossings it means that the crossings of the computed \immersion{} $\Pi$  have the form $(0,j)$. We first claim that each bigon of $\Pi$ must have its two sides oriented the same way. Otherwise, the tips of the bigon have the form $(0,j)$ and $(j,0)$ for some index $j\neq 0$, in contradiction with Lemma~\ref{lem:no-common-index}. 
Let $\tilde{d}_\ell$ and $\tilde{d}_j$ be two intersecting lifts of $c'$ in the Poincar\'e disk. If  $\tilde{d}_\ell$ and $\tilde{d}_j$ intersect only once, then we are done as they must have alternating limit points. We now suppose that $\tilde{d}_\ell$ and $\tilde{d}_j$ intersect at least twice and we consider the bigon $\Delta$ between their \emph{last two} intersections in the forward direction. 
By Lemma~\ref{lem:short-bigon}, the length of $\Delta$ is smaller than $\ell$. Let $(\ell-j,0)$ and $(0,j)$ be the tips of $\Delta$, so that $\tilde{d}_j\Ipath{0}{j}\sim \tilde{d}_\ell\Ipath{\ell-j}{j}$. 
\begin{itemize}
\item 
If $[0,1]$ and $\alpha=[\ell-j,\ell-j+1]$ label the same arc $\tilde{d}_j[0,1] = \tilde{d}_\ell(\alpha)$ then, by the insertion procedure, one of the limit points of $\tilde{d}_\ell$ is on the same side of $\tilde{d}_j$ as $\alpha$. The same argument as in the general case $i<\ell$ allows us to conclude that $\tilde{d}_\ell$ and $\tilde{d}_j$ have alternating limit points. 
\item
We finally suppose that $[0,1]$ and $\alpha$ label distinct arcs, say $a_0$ and $a_{\ell-j}$, of a staircase part $\sigma$ of $\Delta$. 
As in the general case (i.e., $i<\ell$),  $a_0$ must see $\sigma$ on its left, while $a_{\ell-j}$ must see $\sigma$ on its right. Hence, the $\tilde{d}_j$ side of $\sigma$ is canonical while the side along $\tilde{d}_\ell$ is not and $\alpha$ must have been switched. By Claim~\ref{claim:parallel-right}, there is an arc occurrence $\beta$ to the right of $\tilde{d}_\ell[\ell-j,\ell-j+1](\alpha)$ with the opposite orientation. Let $\tilde{q}$ be the lift of $c'$ that extends $\beta$. As in the general case $i<\ell$, we can show that the part of $\tilde{q}$ after $\beta$ cannot cross $\tilde{d}_\ell$, while the part before $\beta$ cannot cross $\tilde{d}_j$. Using that $(0,j)$ is the last crossing along $\tilde{d}_\ell$ and $\tilde{d}_j$, we equally conclude that $\tilde{d}_\ell$ and $\tilde{q}$ have alternating limit points. See Figure~\ref{fig:push-2} with $i=\ell$.
\end{itemize}
\end{proof}

\begin{proof}[Proof of Theorem~\ref{th:simple-curve}]
  Let $c$ be a combinatorial curve of length $\ell$ on a combinatorial surface of size $n$. We compute its canonical form in $O(n+\ell)$ time and check in linear time that $c$ is primitive. In the negative, we conclude that either $c$ is contractible, hence reduced to a vertex, or that  $c$ has no embedding by Proposition~\ref{prop:non-primitive-formulas}. In the affirmative, we apply the {\zipper} algorithm to compute \animmersion{} $\Pi$ of some geodesic $c'$ homotopic to $c$. According to Proposition~\ref{prop:zipper}, we have $i(c)=0$ if and only if $\Pi$ has no crossings. This is easily verified in $O(n+\ell)$ time by checking for each vertex $v$ of the system of quads that the set of paired arc occurrences with $v$ as middle vertex form a well-parenthesized sequence with respect to the local ordering $\prec_v$ induced by $\Pi$. We conclude the proof thanks to  Lemma~\ref{lem:zipper-complexity}.
\end{proof}

\section{Concluding remarks}
The existence of a singular bigon claimed in Theorem~\ref{th:singular-bigon} relies on Theorem 4.2 of Hass and Scott~\cite{hs-ics-85}. As noted by the authors themselves this result is ``surprisingly difficult to prove''. Except for this result and the recourse to some hyperbolic geometry in the general strategy of Section~\ref{sec:strategy}
our combinatorial framework allows us to provide simple algorithms and to give simple proofs of results whose known demonstrations are rather involved. 
Concerning Proposition~\ref{prop:bigon}, the existence of \animmersion{} without bigon could be achieved in our combinatorial viewpoint by showing that if \animmersion{} has bigons, then one of them can be swapped to reduce the number of crossings. (The example in Figure~\ref{fig:no-singular-bigon} shows that such a bigon need not be singular.) What is more, if those swappable bigons could be found easily this would provide an algorithm to compute a minimally crossing \immersion{} of two curves by iteratively swapping bigons as in Section~\ref{sec:computing-immersions} for the case of a single curve. However, we were unable to show the existence or even an appropriate definition of a swappable bigon. Note that in the analogous approaches using Reidemeister-like moves by de Graff and Schrijver~\cite{gs-mcmcr-97} or by Paterson~\cite{p-cails-02}, the number of moves required to reach a minimal configuration is unknown\footnote{In a very recent work, Chang et al.~\cite{celms-tcsvl-18} were able to show that a \emph{single} curve with $n$ self-intersections on an orientable surface could be put in minimal configuration by applying $O(n^4)$ Reidemeister-like moves. In our framework a combinatorial perturbation of a closed walk of length $n$ may have $\Omega(n^2)$ self-intersections, so that their approach leads to $O(n^8)$ moves to reach a minimal configuration.}. Comparatively, the number of bigon swaps would be just half the excess crossing of a given \immersion{}. We would thus obtain a polynomial time algorithm for computing a minimally crossing \immersion{} of two curves (there is an exponential time algorithm by a result of Neumann-Coto~\cite[Prop. 2.2]{n-csgs-01}).

Although the geometric intersection number of a combinatorial curve of length $\ell$ may be $\Omega(\ell^2)$, it is not clear that the complexity in Theorem~\ref{th:main-result} is optimal. In particular, it would be interesting to see if the {\zipper} algorithm of Section~\ref{subsec:unzip} yields minimally crossing curves even with curves that are not homotopic to simple curves, thus improving Theorem~\ref{th:compute-immersion}. This would yield a quasi-linear algorithm for computing the geometric intersection number based on the counting of inversions in a permutation~\cite{p-blrw-84}. It is also tempting to check whether the {\zipper} algorithm applies to compute the geometric intersection number of two curves rather than a single curve.
Another intriguing question related to the computation of minimally crossing \immersion{}s comes from the fact that they are not unique. Given a combinatorial \immersion{} we can construct a continuous realization as described in the proof of Lemma~\ref{lem:geometric-vs-combinatorial}. Say that two realizations are in the \define{same configuration} if there is an ambient isotopy of the surface where they live that brings one realization to the other. It was shown by Neumann-Coto~\cite{n-csgs-01} that every minimally crossing \immersion{} is in the configuration of shortest geodesics for some Riemaniann metric $\mu$, but Hass and Scott~\cite{hs-ccgs-99} gave counterexamples to the fact that we could always choose $\mu$ to be hyperbolic. We conclude with the following open question:
\textit{Is there an algorithm to construct or detect combinatorial \immersion{}s that have a realization in the configuration of geodesics for some hyperbolic metric?}

\bibliographystyle{alpha}
\newcommand{\etalchar}[1]{$^{#1}$}

\newpage
\appendix
\noindent
\textbf{\huge Appendix}
\section{Proof of the direct implication of Proposition~\lowercase{\ref{prop:bigon}}}\label{app:bigon}
\begin{nonumberprop}
  If a combinatorial \immersion{} of one or two primitive curves has excess (self-)crossing then it contains a bigon or a monogon.
\end{nonumberprop}
We first introduce some terminology. 
An \define{elementary homotopy} on a combinatorial curve $c$ consists in adding or removing a spur, or replacing in $c$ a possibly empty part of a facial walk by its complementary part. For a free elementary homotopy we can also apply a circular shift to the indices of the closed curve $c$. The equivalence relation generated by (free) homotopies is called combinatorial (free) homotopy. 

Let $\rho: G\to S$ be a cellular embedding of $G$ in a topological surface $S$ corresponding to the  combinatorial surface we are working with. Every combinatorial curve $c$ can be realized by a continuous curve on $S$ by replacing each arc in $c$ with a continuous arc deduced from the restriction of $\rho$ to the corresponding arc. We denote this continuous realization by $\rho(c)$. It is part of the folklore that (free)  combinatorial homotopy coincides with continuous (free) homotopy, meaning that $c$ is (freely) homotopic to $d$ if an only if $\rho(c)$ is (freely) homotopic to $\rho(d)$. 

An \define{elementary move} of a combinatorial \immersion{} consists either in an elementary homotopy or in an \define{adjacent transposition}, i.e. in exchanging the left-to-right order of two occurrences in a same arc where one occurrence is next to the right of the other. We further require before performing an elementary homotopy that the \immersion{} is in \define{good position}. This means, if the elementary homotopy applies to a nonempty part $u$ of a facial walk of some face, that each arc occurrence in $u$ should be the rightmost element of its arc, i.e. the most interior to the face. When removing a spur, we just require that the two arc occurrences to be removed are adjacent in left-to-right order.  See Figure~\ref{fig:elementary-moves}.  
\begin{figure}[h]
  \centering
\includesvg[.85\linewidth]{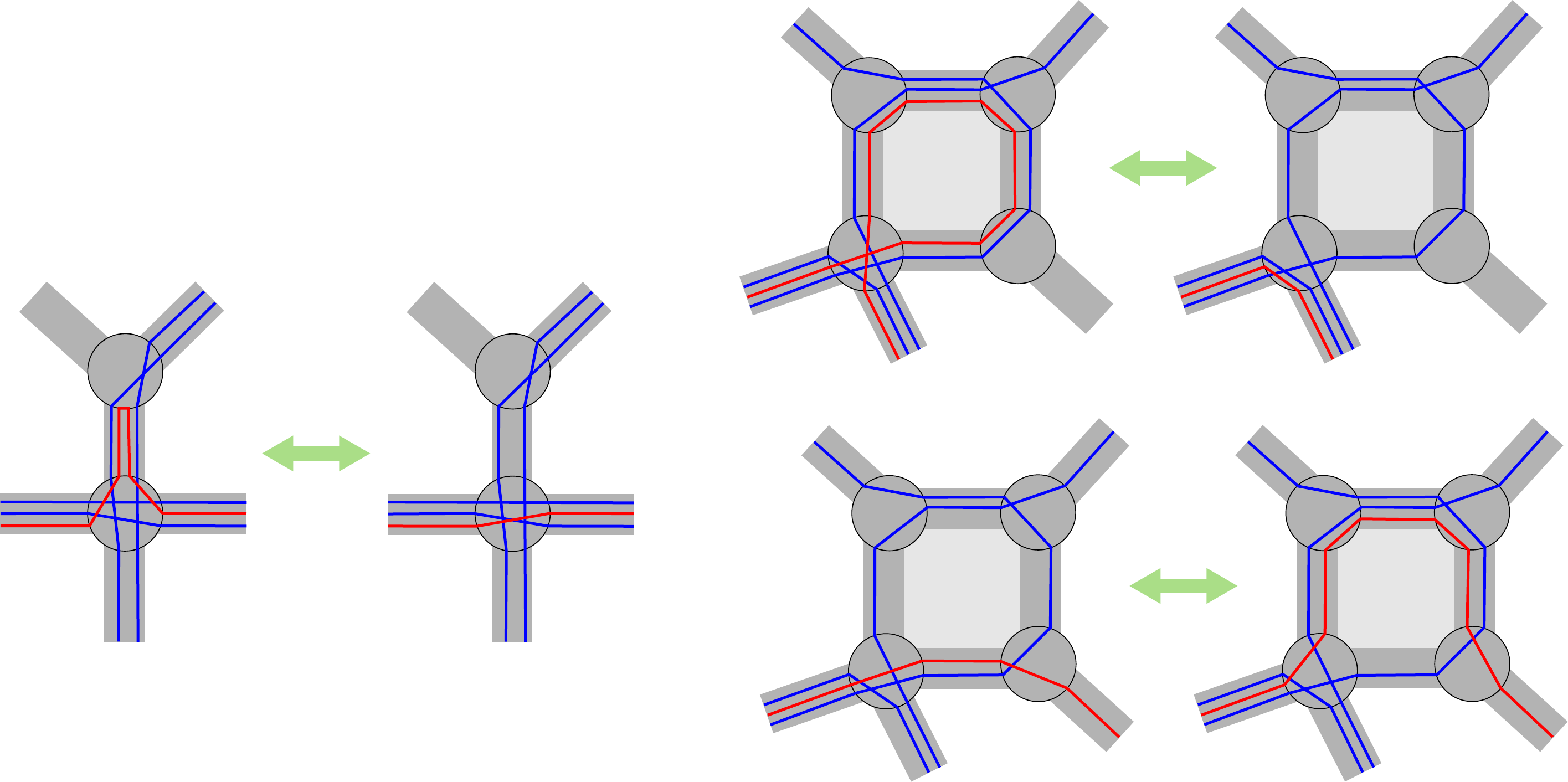}
\caption{Left, the removal of a spur in good position. Right, two elementary homotopies replacing a subpath $u$ of a facial walk by the complementary subpath $v$. In the upper right, $v$ is empty. The \immersion{}s being in good positions, crossings may only appear or disappear by pairs.
}
  \label{fig:elementary-moves}
\end{figure}
The elementary homotopy does not modify the order of the remaining arc occurrences. When inserting a spur we make the inserted arc occurrences adjacent and when inserting part of a facial walk we insert each arc occurrence after the rightmost element of its arc. Remark that by a sequence of adjacent transpositions we can always enforce \animmersion{} to be in good position. 

\begin{proof}[Proof of the Proposition]
We consider the case of \animmersion{} $\Pi$ of two curves $c$ and $d$. Suppose that $\Pi$ has $p$ excess crossings. We shall prove the stronger claim that $\Pi$ has at least $\lceil p/2\rceil$ bigons whose tips are pairwise distinct. Let $\Psi$ be \animmersion{} of two curves $c'$ and $d'$ respectively homotopic to $c$ and $d$ such that $\Psi$ has no excess crossings. Consider a sequence of elementary homotopies from $(c', d')$ to $(c,d)$. Following the above remark, we can insert adjacent transpositions between each elementary homotopy in order to obtain a sequence of elementary moves from $\Psi$ to $\Pi$. The claim is trivially verified by $\Psi$. We now check that the claim remains true after each elementary move. If the move is an adjacent transposition we may assume that we exchange an arc occurrence of $c$ with an arc occurrence of $d$ since we only consider intersections between $c$ and $d$. Without loss of generality we also assume exchanging the forward occurrences $[\Ibar,\Ibar+1]_c$ and $[\Jbar,\Jbar+1]_d$. There are three cases to consider. 
  \begin{enumerate}
  \item If none of $\Dpoint{\Imath}{\Jmath}$ and $\Dpoint{\Imath+1}{\Jmath+1}$ is a crossing then the transposition adds two crossings that we may pair as the tips of a new bigon. We now have $p+2$ excess crossings with at least $\lceil (p+2)/2\rceil$ bigons as required. 
\item\label{it:one-of-two} If $\Dpoint{\Imath}{\Jmath}$ is a crossing but $\Dpoint{\Imath+1}{\Jmath+1}$ is not and if $\Dpoint{\Imath}{\Jmath}$ is paired to form the tips of a bigon, say $(\Ipath{\Imath}{\ell}_c,\Ipath{\Jmath}{k}_d)$, we may just replace this bigon by $(\Ipath{\Imath+1}{\ell-1}_c,\Ipath{\Jmath+1}{k-1}_d)$ sliding its tip  $\Dpoint{\Imath}{\Jmath}$ to $\Dpoint{\Imath+1}{\Jmath+1}$. A similar procedure applies when $\Dpoint{\Imath+1}{\Jmath+1}$ is a crossing but $\Dpoint{\Imath}{\Jmath}$ is not. In each case the number of  excess crossings and bigons is left unchanged. 
\item\label{it:two-of-two} It remains the case where both $\Dpoint{\Imath}{\Jmath}$ and $\Dpoint{\Imath+1}{\Jmath+1}$ are crossings. If none of the two is paired to form the tips of a bigon, then the transposition removes two crossings and no pairing, so that the claim remains trivially true. If exactly one of the two is paired or if the two are paired together, we loose one bigon and two crossings after the transposition so that the claim remains true. Otherwise, there are two bigons of the form $(\Ipath{\Imath}{\ell}_c,\Ipath{\Jmath}{k}_d)$ and $(\Ipath{\Imath+1}{\ell'}_c,\Ipath{j+1}{k'}_d)$ that we can recombine to form the bigon $(\Ipath{\Imath+\ell}{1+\ell'-\ell}_c,\Ipath{j+k}{1+k'-k}_d)$. We again have one less bigon and two less crossings. 
  \end{enumerate}
We now consider the application of an elementary homotopy as described before the proof. There are again three possibilities.
\begin{enumerate}
\item 
If the homotopy replaces a nonempty subpath $u$ of a facial walk $uv\inv$ by the nonempty complementary part $v$ then no crossing may appear or disappear as we assume the \immersion{} in good position. In particular, no crossing may use an internal vertex of $u$ and $u$ is either entirely included in or excluded from any side of any bigon. We can replace $u$ by $v$ in any bigon side where $u$ occurs to obtain valid bigons in the new \immersion{} after the elementary homotopy is applied. The number of  excess crossings and bigons is left unchanged. 
\item When $u$  is empty in the above replacement, or when inserting a spur, we may only add crossings by pairs forming bigons with one zero-length side and the complement of $u$ or the spur as the other side. A similar analysis as in case~\eqref{it:one-of-two} of a transposition applies to take care of each pair.
\item When $v$ is empty in the above replacement, or when removing a spur, we may only remove crossings by pairs and a similar analysis as in case~\eqref{it:two-of-two} of a transposition applies to take care of each pair.
\end{enumerate}
This ends the proof for the case of two curves.
A similar proof holds for the existence of a bigon or a monogon in \animmersion{} with excess self-intersection. This time the excess crossings are either paired to form bigons or left alone as the tips of monogons.
 \end{proof}

\section{Proof of Lemma~\lowercase{\ref{lem:swap}}}\label{app:swap}
\begin{nonumberlem}
   Swapping the two sides of a singular bigon of \animmersion{} of a geodesic primitive  curve decreases its number of crossings by at least two. 
\end{nonumberlem}
\begin{proof}
  Consider a singular bigon $(\Ipath{\Imath}{\ell},\Ipath{\Jmath}{\ell})$ of \animmersion{} $\Pi$ of a closed primitive canonical curve $c$. Let $\Psi$ be the \immersion{} after the bigon has been swapped. We shall partition the set of potential double points and show that the net change of the number of crossings with respect to $\Pi$ and $\Psi$ is non positive in each part. As the tips of the bigon are not crossings in $\Psi$, this will prove the lemma. 
We set for $\Overline{x}\in\Z/|c|\Z\setminus (\Ipath{\Imath}{\ell}\cup\Ipath{\Jmath}{\ell})$ and $0< p,q<\ell$:
\[\small
  \begin{array}{ll}
\D_{p,q}=\{\Dpoint{\Imath+p}{\Jmath+q}, \{\Dpoint{\Imath+q}{\Jmath+p}\},& \D_{x,p}=\{\Dpoint{x}{\Imath+p}, \Dpoint{x}{\Jmath+p}\}, \\
\D_{p,0}=\{\Dpoint{\Imath}{\Imath+p},\Dpoint{\Imath}{\Jmath+p},\Dpoint{\Jmath}{\Imath+p}, \Dpoint{\Jmath}{\Jmath+p}\}, &
\D_{x,0}=\{\Dpoint{x}{\Imath}, \Dpoint{x}{\Jmath}\}, \\
\D_{p,\ell}=\{\Dpoint{\Imath+\ell}{\Imath+p},\Dpoint{\Imath+\ell}{\Jmath+p},
\Dpoint{\Jmath+\ell}{\Imath+p}, \Dpoint{\Jmath+\ell}{\Jmath+p}\}, &
 \D_{x,\ell}=\{\Dpoint{x}{\Imath+\ell}, \Dpoint{x}{\Jmath+\ell}\},\\
\D'_{p,q}=\{\Dpoint{\Imath+p}{\Imath+q}, \{\Dpoint{\Jmath+q}{\Jmath+p}\}, 
  & \hspace{-1cm} \D_1=\{\Dpoint{y}{z}\mid y,z \not\in(\Ipath{\Imath}{\ell}\cup\Ipath{\Jmath}{\ell}) \}, \\
\D_2=\{\Dpoint{\Imath}{\Imath+\ell}, \Dpoint{\Imath}{\Jmath+\ell}, \Dpoint{\Jmath}{\Jmath+\ell}, \Dpoint{\Jmath}{\Imath+\ell}\},&
\D_3 =\{\Dpoint{\Imath}{\Jmath}, \Dpoint{\Imath+\ell}{\Jmath+\ell}\}
\end{array}
\]
When $\Ibar=\Overline{\Jmath+\ell}$ or $\Jbar=\Overline{\Imath+\ell}$, the set $\D_2$ consists of three double points only (or zero if $c(\Ibar)\neq c(\Overline{\Imath+\ell})$). Note that we cannot have both equalities as this would imply $\Jbar=\Overline{\Jmath+2\ell}$, whence $\ell=|c|/2$ and $c$ would be a square, in contradiction with the hypothesis that $c$ is primitive.
For each double point in $\D_1$ the four incident arc occurrences are  left in place in $\Pi$ and $\Psi$. It ensues that $\D_1$ has the same crossings in $\Pi$ and $\Psi$.  The double points in $\D_3$ are the tips of the bigon. They are crossings in $\Pi$ and not in $\Psi$, whence a net change of $-2$ crossings. For each of $\D_{x,p}, \D_{p,q}$ or $\D'_{p,q}$, the first double point in their above definition is a crossing in $\Pi$ (resp. $\Psi$) if and only if the second one is a crossing in $\Psi$ (resp. $\Pi$). Their net change of crossings is thus null. For $\D_{x,0}$, $\D_{x,\ell}$, $\D_{p,0}$ and $\D_{p,\ell}$  there are a few case analysis depending on the relative ordering of the two arc occurrences incident to $\Overline{x}$ (resp. $\Overline{\Imath+p}$, $\Overline{\Jmath+p}$) with respect to the crossing $\Dpoint{\Imath}{\Jmath}$ (resp. $\Dpoint{\Imath+\ell}{\Jmath+\ell}$. In each case (see Figure~\ref{fig:Dx1}) the number of crossings cannot increase from $\Pi$ to $\Psi$.
\begin{figure}
  \centering
  \includegraphics[width=\linewidth]{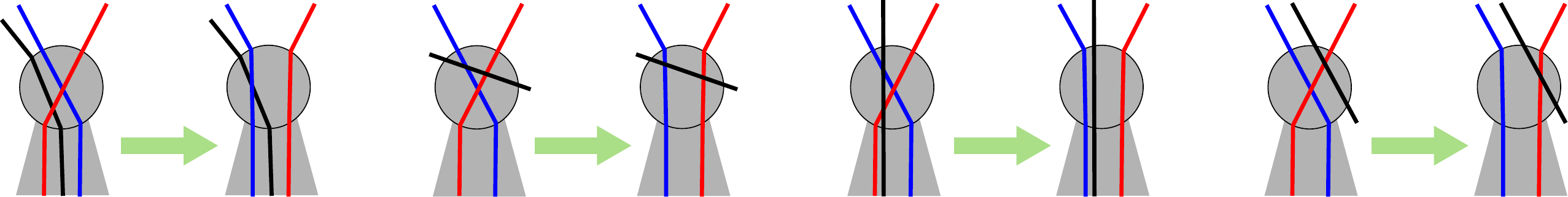}
  \caption{Up to some obvious symmetries the case $\D_{x,0}$, $\D_{x,\ell}$, $\D_{p,0}$ or $\D_{p,\ell}$ has four distinct possible configurations. The blue and red strands represent the crossing $\Dpoint{\Imath}{\Jmath}$ (resp. $\Dpoint{\Imath+\ell}{\Jmath+\ell}$) and each arrow links the left configuration before the bigon swap (in $\Pi$) with the right configuration after the swap is applied (in $\Psi$).}
  \label{fig:Dx1}
\end{figure}
For $\D_2$ there are three cases according to whether $\Overline{\Imath+\ell}=\Jbar$, $\Overline{\Jmath+\ell}=\Ibar$, or none of the two identifications occurs. Recall from the paragraph before the lemma that we cannot have both identifications. The first two  cases can be treated the same way. See Figure~\ref{fig:D2special}. 
\begin{figure}
  \centering
  \includegraphics[width=.8\linewidth]{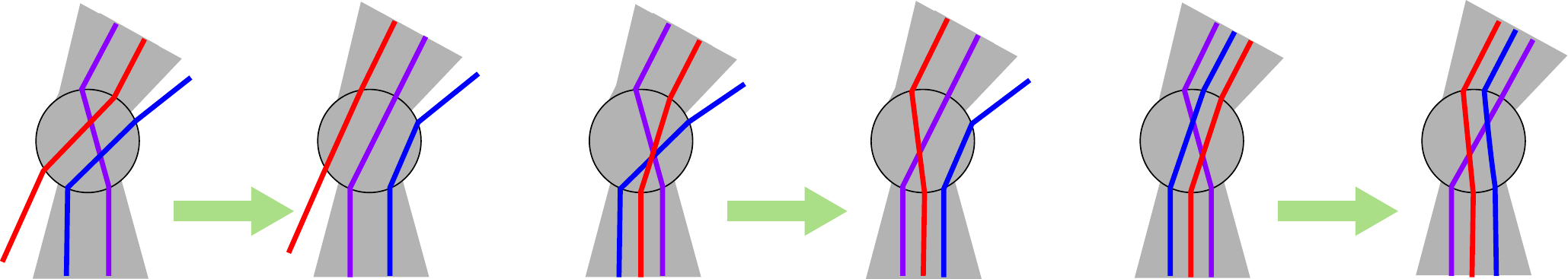}
  \caption{When $\Overline{\Imath+\ell}=\Jbar$, the end of $\Ipath{\Imath}{\ell}$ overlaps with the beginning of $\Ipath{\Jmath}{\ell}$ (in the purple strand). The third configuration on the right is actually forbidden by the definition of a singular bigon! It corresponds to the non-singular bigon on Figure~\ref{fig:singular-def}.}
  \label{fig:D2special}
\end{figure}
The number of possible configurations is much larger in the last case, i.e. when $\Overline{\Imath+\ell}\neq\Jbar$ and $\Overline{\Jmath+\ell}\neq\Ibar$. We essentially have to shuffle two circular orderings of length four corresponding to the two crossing tips. Without loss of generality we can assume that the path from $\Ibar$ to $\Overline{\Imath+\ell}$ is the right side of the bigon. We consider the following arc occurrences: 
\[
\begin{array}{llll}
  \alpha=\Focc{\Imath},&\alpha'=\Bocc{\Imath}, & \beta=\Focc{\Jmath},& \beta'=\Bocc{\Jmath}\\
\gamma=\Focc{\Imath+\ell},& \gamma'=\Bocc{\Imath+\ell},&  \delta=\Focc{\Jmath+\ell},& \delta'=\Bocc{\Jmath+\ell}
\end{array}
\]
Since $\Dpoint{\Imath}{\Jmath}$ and $\Dpoint{\Imath+\ell}{\Jmath+\ell}$ are crossings we must see $(\alpha,\beta,\alpha',\beta')$ in this counterclockwise order around the vertex $c(\Ibar)=c(\Overline{\Imath+\ell})$ and similarly for $(\gamma,\delta,\gamma',\delta')$. We denote by $S_1$ and $S_2$, respectively, these two circular sequences. We need to consider the effect of the bigon swapping on all the possible shuffles of $S_1$ and $S_2$. The restriction of these shuffles to $\alpha,\beta,\gamma$ and $\delta$ gives the 6 possible shuffles of $(\alpha,\beta)$ and $(\gamma,\delta)$. Among them the order $(\alpha,\delta,\gamma,\beta)$ cannot occur. Indeed, since the bigon is a thick double path the arcs $c(\alpha)$ and $c(\beta)$ either coincide or form a corner of a quad. This would force  $c(\delta)$ and $c(\gamma)$ to lie in a similar configuration. In turn, the constrained order $S_2$ would also enforce  $c(\delta')$ and $c(\delta)$ to coincide or form a corner of a quad in contradiction with the hypothesis that $c$ has no spurs or brackets. Similar arguments show that the orders $(\alpha,\gamma,\beta,\delta)$, $(\alpha,\delta,\beta,\gamma)$, $(\alpha,\gamma,\delta,\beta)$ and $(\alpha,\beta,\delta,\gamma)$ can only occur as factors in the possible shuffles of $S_1$ and $S_2$. These 4 orders thus leads to 24 distinct shuffles of $S_1$ and $S_2$ by factoring with the 6 shuffles of $(\alpha',\beta')$ with $(\gamma',\delta')$. By exchanging the roles of $(\alpha,\beta)$ and $(\gamma,\delta)$ and by turning clockwise instead of counterclockwise we see that $(\alpha,\delta,\beta,\gamma)$ leads to the same orders as $(\alpha,\gamma,\beta,\delta)$ and a similar correspondence holds for $(\alpha,\gamma,\delta,\beta)$ and $(\alpha,\beta,\delta,\gamma)$. We thus only need to check the 12 configurations depicted on Figure~\ref{fig:D2quater}. 
\begin{figure}
  \centering
  \includegraphics[width=\linewidth]{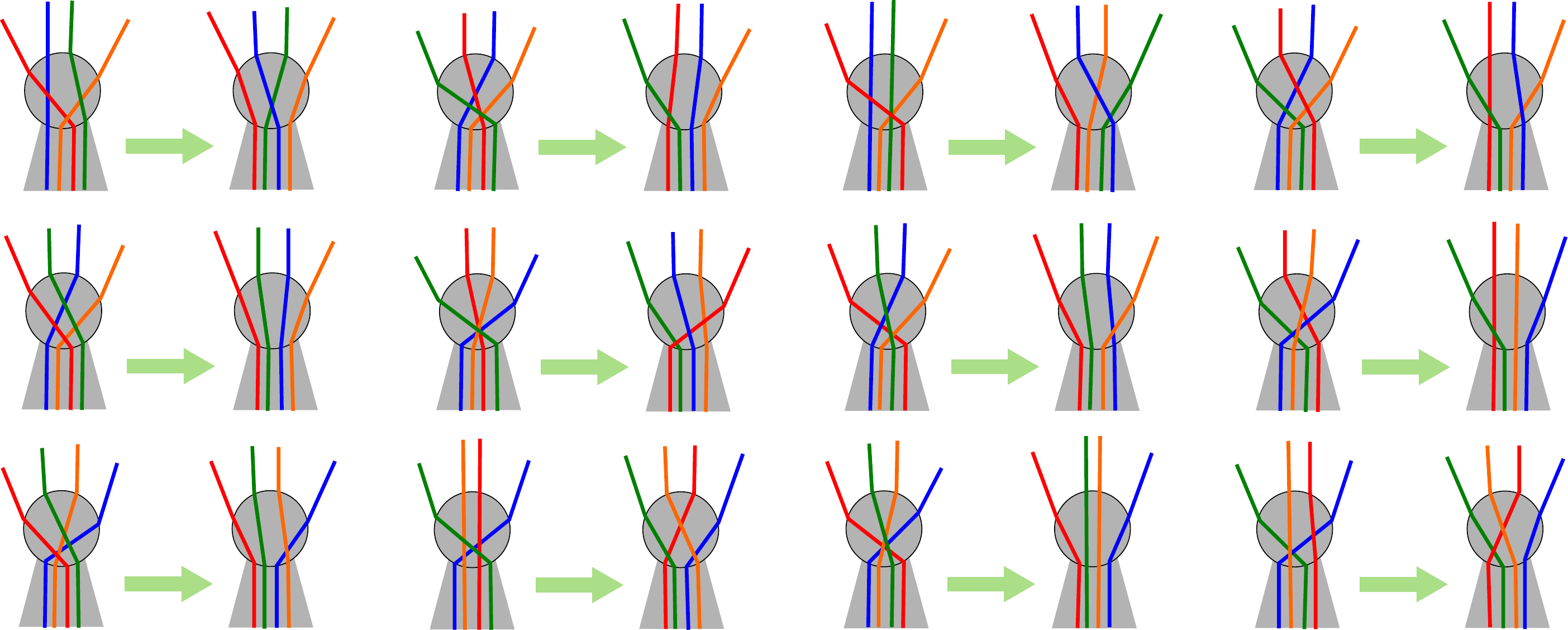}
  \caption{Effect of swapping the singular bigon $(\Ipath{\Imath}{\ell},\Ipath{\Jmath}{\ell})$ on $\D_2$ for all the orders including the factor $(\alpha,\gamma,\beta,\delta)$ or  $(\alpha,\gamma,\delta,\beta)$. The blue strand represents $(\alpha, \alpha')$, the red one $(\beta, \beta')$, the orange one $(\gamma, \gamma')$ and the green one $(\delta, \delta')$.}
  \label{fig:D2quater}
\end{figure}
It remains to consider the order $(\alpha,\beta,\gamma,\delta)$.  Because of the constrained order $S_2$, $\gamma'$ and $\delta'$ cannot lie between $\gamma$ and $\delta$. This leaves out $\binom{4}{2}=6$ possible shuffles of $(\gamma',\delta')$ and $(\alpha,\beta,\gamma,\delta)$:
\[
\begin{array}{lll}
S'_1: (\alpha,\beta,\gamma,\delta,\gamma',\delta') &
S'_2: (\alpha,\beta,\delta',\gamma,\delta,\gamma') &
S'_3: (\alpha,\beta,\gamma',\delta',\gamma,\delta) \\
S'_4: (\alpha,\delta',\beta,\gamma,\delta,\gamma') &
S'_5: (\alpha,\gamma',\beta,\delta',\gamma,\delta) &
S'_6: (\alpha,\gamma',\delta',\beta,\gamma,\delta)
\end{array}
\]
We finally shuffle each of these orders with $(\alpha', \beta')$. Since $\alpha'$ and $\beta'$ cannot lie between $\alpha$ and $\beta$, we obtain $\binom{6}{2}=15$ possible shuffles when considering either $S'_1, S'_2$ or $S'_3$, $\binom{5}{2}=10$ possible shuffles with $S'_4$ or $S'_5$, and $\binom{4}{2}=6$ possible shuffles with  $S'_6$. By swapping left and right and turning clockwise instead of counterclockwise, we remark that $S'_1$ and $S'_3$ lead to the same orders and similarly for $S'_4$ and $S'_5$. We thus only need to consider the shuffles of $(\alpha', \beta')$ and $S'_1, S'_2, S'_4$ or $S'_6$ to complete the inspection of all the cases. Those shuffles are represented on Figures~\ref{fig:D2-S1},~\ref{fig:D2-S2},~\ref{fig:D2-S4},~\ref{fig:D2-S6} respectively.
\begin{figure}
  \centering
  \includesvg[\linewidth]{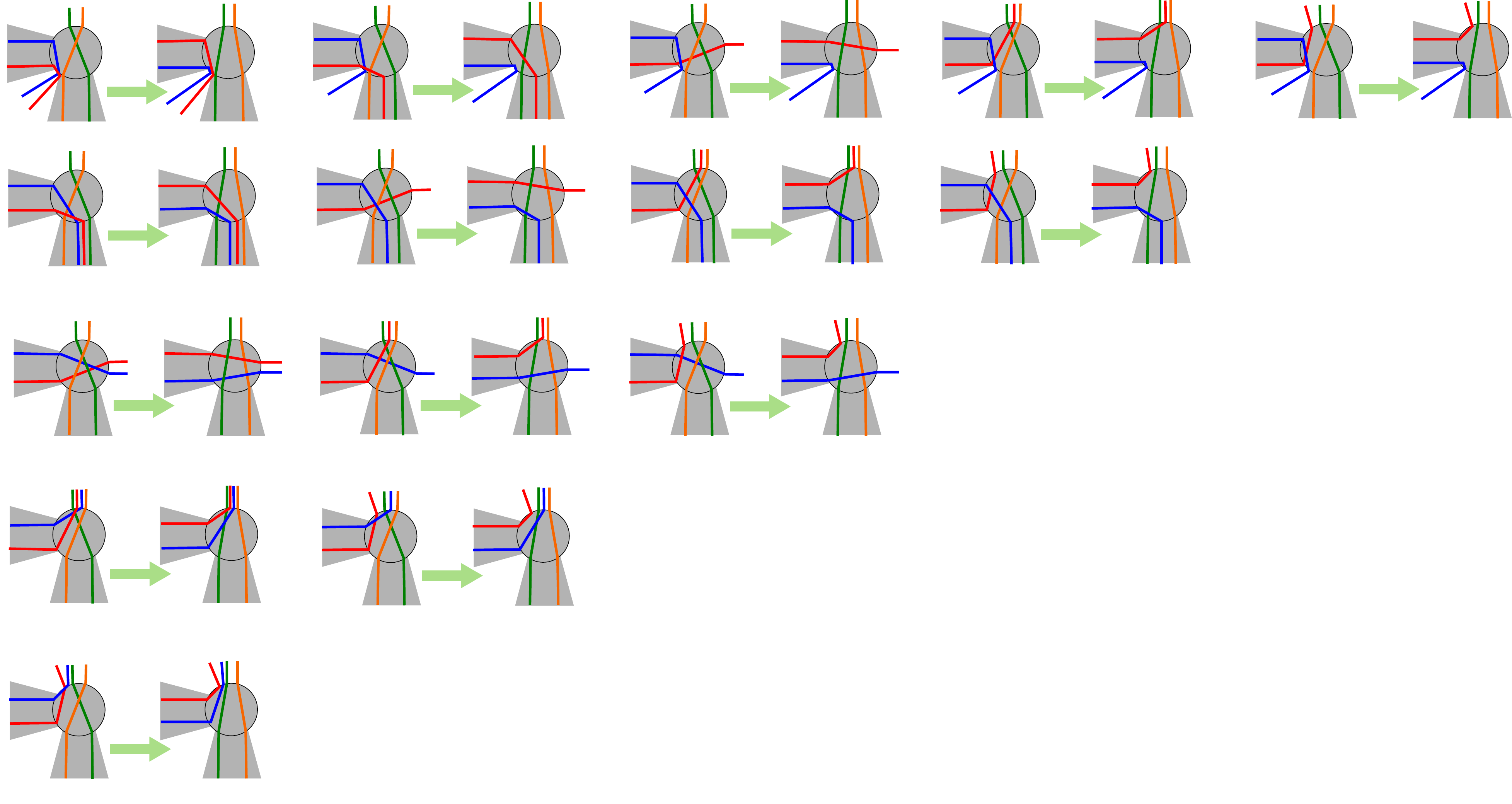}
  \caption{Effect of the bigon swapping on the 15 shuffles of $S'_1$ and $(\alpha', \beta')$.}
  \label{fig:D2-S1}
\end{figure}
\begin{figure}
  \centering
  \includesvg[\linewidth]{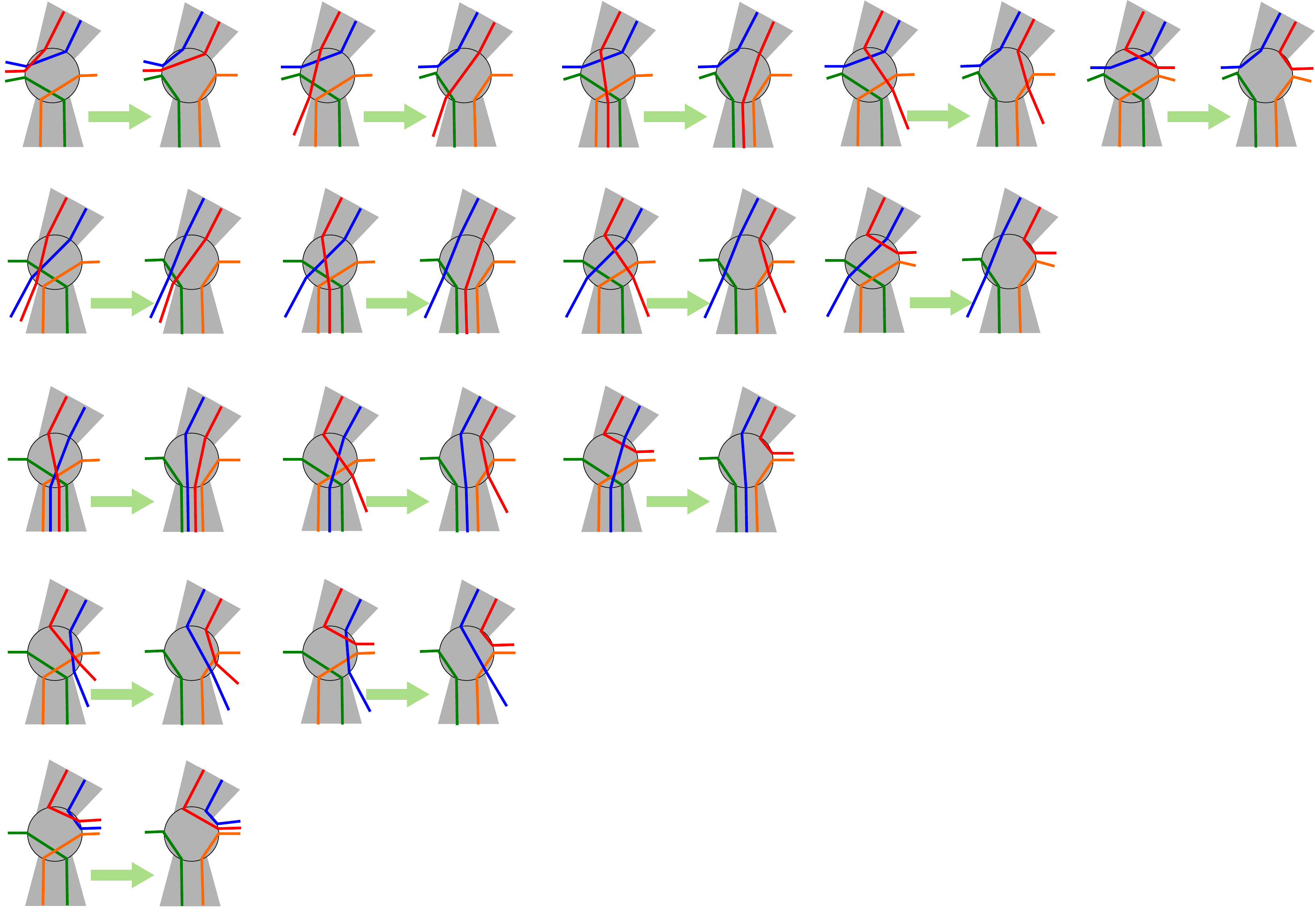}
  \caption{Effect of the bigon swapping on the 15 shuffles of $S'_2$ and $(\alpha', \beta')$.}
  \label{fig:D2-S2}
\end{figure}
\begin{figure}
  \centering
  \includesvg[\linewidth]{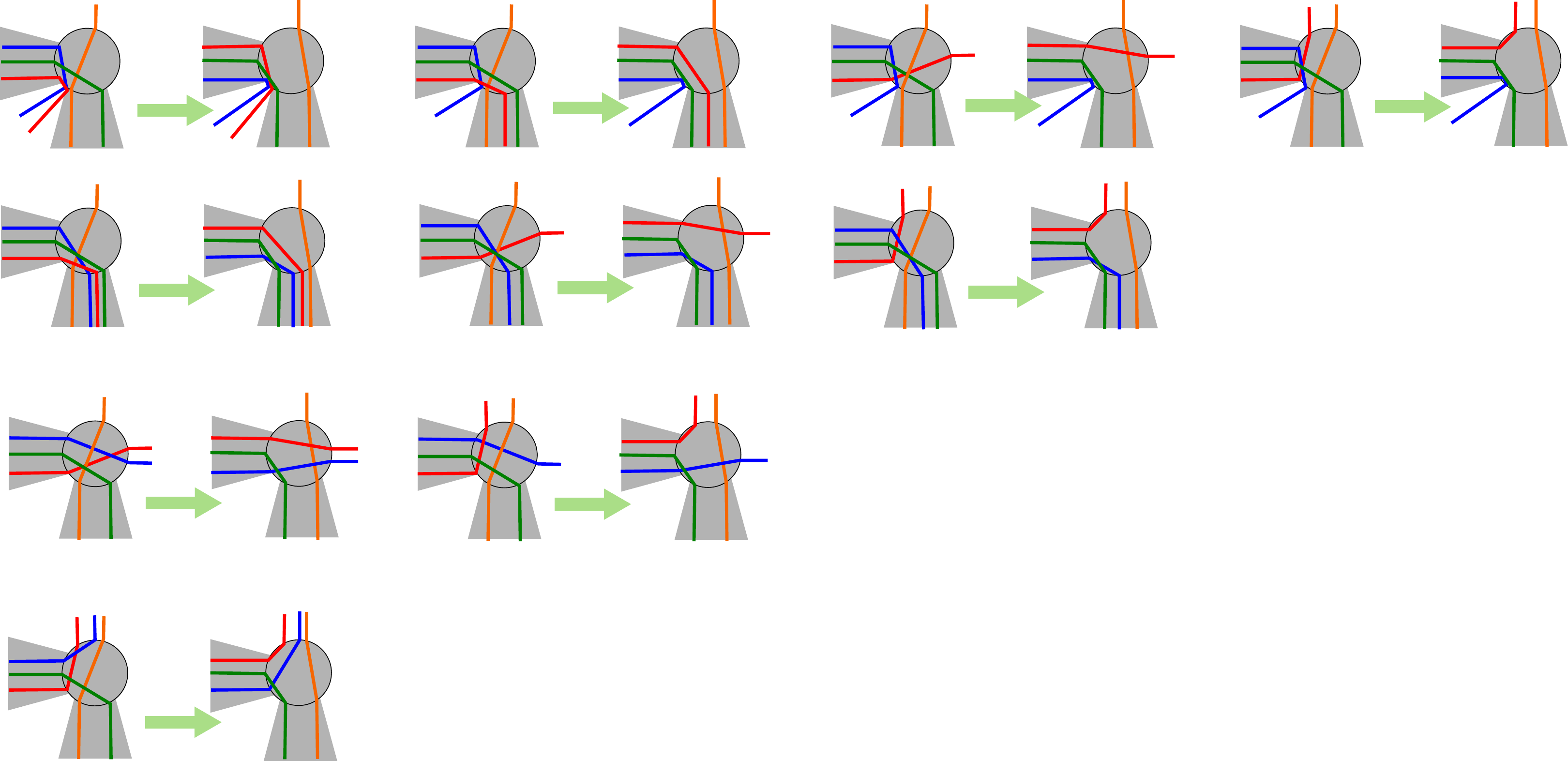}
  \caption{Effect of the bigon swapping on the 10 shuffles of $S'_4$ and $(\alpha', \beta')$.}
  \label{fig:D2-S4}
\end{figure}
\begin{figure}
  \centering
  \includesvg[\linewidth]{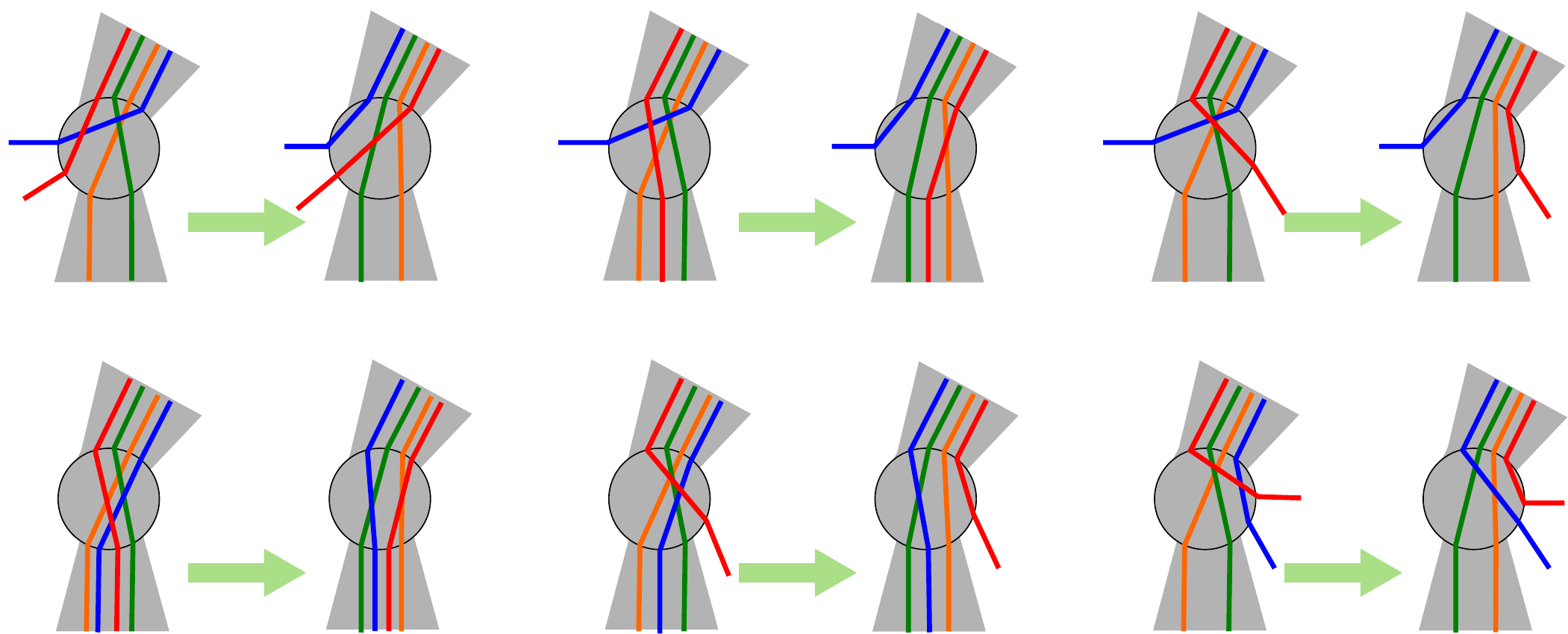}
  \caption{Effect of the bigon swapping on the 6 shuffles of $S'_6$ and $(\alpha', \beta')$.}
  \label{fig:D2-S6}
\end{figure}

In each of the configurations, we trivially check that the number of crossings is not increasing. This allows to conclude the lemma when the two index paths of a singular bigon are directed the same way. A similar analysis can be made when their directions are opposite, that is when the considered  bigon has the form $(\Ipath{\Imath}{\ell},\Ipath{\Jmath}{-\ell})$.
\end{proof}

\end{document}